\documentclass{article}

\usepackage{fullpage}
\usepackage{booktabs} 
\usepackage{amssymb}
\usepackage{pifont}
\usepackage{authblk}

\usepackage{soul}
\usepackage{amsthm}
\usepackage{amsmath,amssymb}

\usepackage{tikz}
\usetikzlibrary{arrows,positioning,shapes} 
\tikzset{
	>=stealth',
	punkt/.style={
		rectangle,
		rounded corners,
		draw=black, very thick,
		minimum width=8.5em,
		minimum height=2em,
		text centered},
	pil/.style={
		->,
		thick,
		shorten <=2pt,
		shorten >=2pt,}
}
\tikzstyle{line} = [draw, -latex,thick,
shorten <=2pt,
shorten >=2pt]

\usepackage[ruled]{algorithm2e} 

\SetAlFnt{\small}
\SetAlCapFnt{\small}
\SetAlCapNameFnt{\small}
\SetAlCapHSkip{0pt}
\IncMargin{-\parindent}

\usepackage[T1]{fontenc}
\usepackage{microtype}
 
\usepackage{amsmath}
\usepackage{amssymb}
\usepackage{mathtools}
\usepackage{verbatim} 
\usepackage{stmaryrd}
\usepackage{url}
\allowdisplaybreaks
\usepackage{basic}
\usepackage{ambients}
\usepackage{paralist}
\usepackage{graphicx}
\usepackage{proof}

\usepackage{enumitem}

\usepackage{epsf}
\usepackage{graphics} 
\usepackage{epsfig} 

\newif\ifdraft\drafttrue

\newcommand{\ntrans}[1]{\mathrel{{\trans{#1}}\makebox[0em][r]{$\not$\hspace{2ex}}}{\!}}
\newcommand{\ttrans}[1]{\stackrel{\, {#1} \,}{\Longrightarrow}}
\newcommand{\ttranst}[1]{\Longrightarrow\trans{#1}\Longrightarrow}

\newcommand{\Uppaal}{{\sc Uppaal}}

\newcommand{\on}{\mathsf{on}}
\newcommand{\off}{\mathsf{off}}

\newcommand{\bool}[1]{\llbracket #1 \rrbracket}

\newcommand{\confCPS}[2]{#1 \, {\Join} \, #2}

\newcommand{\sniff}[2]{\mathsf{sniff}\, #2(#1)}
\newcommand{\drop}[2]{\mathsf{drop}\, #2(#1)}
\newcommand{\forge}[2]{\mathsf{forge}\, #2
\langle #1 \rangle}

\newcommand{\rsens}[2]{\mathsf{read}\, #2(#1)}

\newcommand{\wact}[2]{\mathsf{write}\, #2 \langle #1 \rangle}

\newcommand{\env}{\mathit{Env}}
\newcommand{\stateSys}{\mathit{State}}

\newcommand{\statefun}{\xi_{\mathrm{x}}}
\newcommand{\actuatorfun}{\xi_{\mathrm{a}}}
\newcommand{\uncertaintyfun}{\xi_{\mathrm{w}}}
\newcommand{\evolmap}{\mathit{evol}}
\newcommand{\errorfun}{\xi_{\mathrm{e}}}
\newcommand{\measmap}{\mathit{meas}}
\newcommand{\invariantfun}{\mathit{inv}}
\newcommand{\safefun}{\mathit{safe}}

\newcommand{\sensorfun}{\xi_{\mathrm{s}}}

\newcommand{\CPS}{CPS}

\renewcommand{\cal}{\mathcal}

\newtheorem{definition}{Definition}
\newtheorem{theorem}{Theorem}
\newtheorem{proposition}{Proposition}
\newtheorem{remark}{Remark}
\newtheorem{lemma}{Lemma}
\newtheorem{example}{Example}
\newtheorem{corollary}{Corollary}
\newtheorem{notation}{Notation}

\usepackage{marvosym} 

\begin{document}

\title{A Formal Approach to Physics-Based Attacks \\ in Cyber-Physical Systems\\ (Extended Version)\thanks{This document extends the paper ``A Formal Approach to Physics-Based Attacks in Cyber-Physical Systems'' that will appear in \emph{ACM Transactions on Privacy and Security} by providing proofs that are worked out in full details.}}

\author[1]{Ruggero Lanotte}
\author[2]{Massimo Merro}
\author[2]{Andrei Munteanu}
\author[3]{Luca Vigan\`o}
\affil[1]{%
 DISUIT, Universit\`a dell'Insubria, Italy\\
  {ruggero.lanotte@uninsubria.it}
}
\affil[2]{%
  Dipartimento di Informatica, Universit\`a degli Studi di Verona, Italy\\
  {\{massimo.merro,andrei.munteanu\}@univr.it}
}
\affil[3]{%
  Department of Informatics, King's College London, UK\\
  {luca.vigano@kcl.ac.uk}
}

\maketitle

\begin{abstract}
We apply formal methods to lay and streamline theoretical foundations to reason about Cyber-Physical Systems (CPSs) and physics-based attacks, i.e., attacks targeting physical devices. 
We focus on a formal treatment of both integrity and denial of service  attacks to sensors and actuators of CPSs, and on the timing aspects of these attacks. 
Our contributions are fourfold. 
(1)~We define a hybrid process calculus to model both CPSs and physics-based attacks.
(2)~We formalise a threat model that  specifies
MITM attacks that can manipulate sensor readings or control commands in order to drive a CPS into an undesired state;  we group these attacks into classes, and provide the means to assess attack tolerance/vulnerability with respect to a
given class of attacks,  based on a proper notion of most powerful physics-based attack. 
(3)~We formalise how to estimate the impact of a successful attack on a CPS and investigate possible quantifications of the success chances of an attack.
(4)~We illustrate our definitions and results by formalising a non-trivial running example in \Uppaal{} SMC, the statistical extension of the \Uppaal{} model checker;
we 
use \Uppaal{} SMC as an automatic tool for carrying out a static security analysis of our running example in isolation and when exposed to three different physics-based attacks with different impacts. 
\end{abstract}



\section{Introduction}
\label{sec:introduction}

\subsection{Context and motivation} 
\emph{Cyber-Physical Systems (CPSs)} are integrations of networking and
distributed computing systems  with physical processes that monitor and
control entities in a physical environment, with feedback loops where
physical processes affect computations and vice versa. For example, in
real-time control systems, a hierarchy of \emph{sensors}, \emph{actuators} and \emph{control components} are connected  to  control stations. 

In recent years there has been a dramatic increase in the number of
attacks to the security of CPSs, e.g., manipulating sensor readings and, in general, influencing physical processes to bring the system into a state desired by the attacker.
Some notorious examples are: 
\begin{inparaenum}[(i)]
	\item the \emph{STUXnet} worm, which reprogrammed PLCs of nuclear
	centrifuges in Iran~\cite{stuxnet}; 
	\item  the attack on a sewage treatment
	facility in Queensland, Australia, which manipulated the SCADA system to
	release raw sewage into local rivers and parks~\cite{SlMi2007}; 
	\item the \emph{BlackEnergy} cyber-attack on the Ukrainian power grid, again  compromising the SCADA system~\cite{ICS15}. 
\end{inparaenum}

A common aspect of these attacks  is that they all compromised \emph{safety critical systems}, i.e., systems whose failures may cause catastrophic consequences. Thus, as stated in~\cite{GGIKLW2015,GollmannK16}, the concern for consequences at the physical level puts \emph{\CPS{} security} apart from standard \emph{information security}, and demands for ad hoc solutions to properly address such novel research challenges. 

These ad hoc solutions must explicitly take into consideration a number of specific issues of attacks tailored for \CPS{s}. One main critical issue is the \emph{timing of the attack}: the physical state of a system changes continuously over time and, as the system evolves, 
some states might be more vulnerable to attacks than others~\cite{KrCa2013}. For example, an attack launched when the target state
variable reaches a local maximum (or minimum) may have a great impact on
the whole system behaviour, whereas the system might be able to tolerate the same attack if launched when that variable is far from its local maximum or minimum~\cite{BestTime2014}.
Furthermore, not only the timing of the attack but also the \emph{duration of the attack} is an important parameter to be taken into consideration in order to achieve a successful attack. For example, it may take minutes for a chemical reactor to rupture~\cite{chemical-reactor}, hours to heat a tank of water or burn out a motor, and days to destroy centrifuges~\cite{stuxnet}. 

Much progress has been done in the last years in developing formal approaches to aid the \emph{safety verification} of CPSs (e.g., \cite{HYTECH, PHAVer, SpaceEx, KeYmaera, PlatzerBook, SurveyBartocciSTL2018}, to name a few). However, there is still a relatively small number of works that use formal methods in the context of the \emph{security analysis} of \CPS{s} (e.g., \cite{BuMaCh2012, LICS-Platzer2018,VNN2013,RocchettoTippenhauer2016b,Nigam-Esorics2016,Akella2010,Bodei2018,Nielson2018}). 
In this respect, to the best of our knowledge, a systematic formal approach to study \emph{physics-based attacks}, that is, attacks targeting the physical devices (sensors and actuators) of \CPS{s}, is still to be fully developed. Our paper moves in this direction by relying on a process calculus approach. 

\subsection{Background}
The dynamic behaviour of the \emph{physical plant} of a \CPS{} is often 
represented by means of a \emph{discrete-time state-space
	model\/}\footnote{See~\cite{survey-CPS-security-2016,survey-CPS-security-2019} for a taxonomy of the time-scale models used to represent \CPS{s}.} 
consisting of two equations of the form
\begin{displaymath} 
\begin{array}{rcl}
x_{k+1} & = & Ax_{k} + Bu_{k} + w_{k} \enspace \\[1pt] 
y_k & = & Cx_{k} + e_k 
\end{array}
\end{displaymath}
where $x_k \in \mathbb{R}^n$ is the current \emph{(physical) state}, $u_k \in \mathbb{R}^m$ is the \emph{input} (i.e., the control actions implemented through actuators) and $y_k \in \mathbb{R}^p$ is the \emph{output} (i.e., the measurements from the sensors). The \emph{uncertainty} $w_k \in \mathbb{R}^n$ and the \emph{measurement error} $e_k \in \mathbb{R}^p$ represent perturbation and sensor noise, respectively, and $A$, $B$, and $C$ are matrices modelling the dynamics of the physical system. Here, the \emph{next state} $x_{k+1}$ depends on the current state $x_k$ and the corresponding control actions $u_k$, at the sampling instant $k \in \mathbb{N}$. The state $x_k$ cannot be directly observed: only its measurements $y_k$ can be observed.

The physical plant is supported by a communication network through which
the sensor measurements and actuator data are exchanged with controller(s) and supervisor(s) (e.g., IDSs), which are the \emph{cyber} components (also called \emph{logics}) of a CPS.

\subsection{Contributions} 
In this paper, we focus on a formal treatment of both \emph{integrity} and \emph{Denial of Service (DoS)} attacks to \emph{physical devices} (sensors and actuators) of \CPS{s}, paying particular attention to the \emph{timing aspects} of these attacks. The overall goal of the paper is to apply \emph{formal methodologies} to lay \emph{theoretical foundations} to reason about and 
formally detect attacks to physical devices of \CPS{s}. A straightforward utilisation of these methodologies is for \emph{model-checking} (as, e.g., in~\cite{SpaceEx})  or \emph{monitoring} (as, e.g., in~\cite{SurveyBartocciSTL2018}) in order to be able to verify security properties of \CPS{s} either before system deployment or, when static analysis is not feasible, at runtime to promptly detect undesired behaviours. In other words, we aim at providing an essential stepping stone for formal and automated analysis techniques for checking the security of CPSs (rather than for providing defence techniques, i.e., \emph{mitigation}~\cite{Mitigation2018}).

Our contribution is fourfold. The \emph{first contribution} is the definition of a \emph{hybrid process calculus}, called \cname{}, to formally specify both \CPS{s} and \emph{physics-based attacks}. In \cname{}, \CPS{s} have two components: 
\begin{itemize}[noitemsep]
	\item a \emph{physical component} denoting the \emph{physical plant}
	(also called environment) of the system, and containing information on
	state variables, actuators, sensors, evolution law, etc., and 
	\item a \emph{cyber component} that governs access to sensors and actuators, and channel-based communication with other cyber components. 
\end{itemize}
Thus, channels are used for logical interactions between cyber components, whereas sensors and actuators make possible the interaction between cyber and physical components.

\cname{} adopts a \emph{discrete notion of time}~\cite{HR95} and it is
equipped with a \emph{labelled transition semantics (LTS)} that allows us to observe both \emph{physical events} (system deadlock and violations of safety conditions) and \emph{cyber events} (channel communications). Based on our LTS, we define two \emph{compositional} trace-based system preorders: a deadlock-sensitive \emph{trace preorder}, $\sqsubseteq$, and a \emph{timed variant}, $\sqsubseteq_{m..n}$, for $m \in \mathbb{N}^{+}$ and $n \in \mathbb{N}^{+}\cup \{\infty\}$, which takes into account discrepancies of execution traces within the discrete time interval $m..n$. Intuitively, given two \CPS{s} $\mathit{Sys}_1$ and $\mathit{Sys}_2$, we write $\mathit{Sys}_1 \sqsubseteq_{m..n} \mathit{Sys}_2$ if $\mathit{Sys}_2$ simulates the execution traces of $\mathit{Sys}_1$, except for the time interval $m..n$; if $n=\infty$ then the simulation only holds for the first $m-1$ time slots.

As a \emph{second contribution}, we formalise a \emph{threat model} that specifies \emph{man-in-the-middle (MITM) attacks} that can manipulate \emph{sensor readings} or \emph{control commands} in order to drive a \CPS{} into an undesired state~\cite{TeShSaJo2015}.\footnote{Note that we focus on attackers
	who have already entered the CPS, and we do not consider how they gained access to the system, e.g., by attacking an Internet-accessible controller or 
	one of the communication protocols as a Dolev-Yao-style attacker~\cite{dolev1983security} would do.} Without loss of generality, MITM attacks targeting physical devices (sensors or actuators) can be assimilated to \emph{physical attacks}, i.e., those attacks that directly compromise physical devices (e.g., electromagnetic attacks).
As depicted in Figure~\ref{fig:threat-model}, our attacks may affect directly the sensor measurements or the controller commands: 
\begin{itemize}[noitemsep]
	\item \emph{Attacks on sensors} consist of reading and eventually 
	replacing $y_k$ (the sensor measurements) with $y^a_k$. 
	\item \emph{Attacks on actuators} consist of reading, dropping and eventually replacing the controller commands $u_k$ with $u^a_k$, affecting directly the actions the actuators may execute.
\end{itemize}
We group attacks into classes. A \emph{class of attacks} takes into account both the potential malicious activities  $\I$ on physical devices and the \emph{timing parameters} $m$ and $n$ of the attack: begin and end of the attack. We represent a class $C$ as a total function $C \in [\I \rightarrow {\cal P}(m..n)]$. Intuitively, for $\iota \in \I$, $C(\iota) \subseteq m..n$ denotes the set of time instants when an attack of class $C$ may tamper with \nolinebreak  the \nolinebreak device \nolinebreak  $\iota$.

\begin{figure}[t]
	\centering
	\begin{minipage}[c]{0.4\textwidth}
		\begin{tikzpicture}[scale=0.4, every node/.style={scale=0.5}]
		\node[punkt] (engine) {\huge Plant};
		
		\node[punkt, below=0.7cm of engine]
		(controller) {\huge Logics};
		
		\node[punkt, right=of engine] (sensor) {\huge Sensors};
		\node[punkt, left=of engine] (actuator) { \huge Actuators}
		edge[pil, right=45] (engine);
		
		\path[line] (engine) -- node[above, pos=0.5] {{\huge${x_k}$}} (sensor);
		\path [line] (sensor) |- node[above, pos=0.7] {{\huge${y_k^a}$}} (controller.east);
		\path [line] (sensor) |- node[right, pos=0.25] {{\huge${y_k}$}} (controller.east);	
		\path [line] (actuator) -- (engine);
		\path [line] (controller.west) -| node[above, pos=0.3,] {\huge ${u_k}$}  (actuator);
		\path [line] (controller.west) -| node[left, pos=0.75,] {\huge ${u_k^a}$}  (actuator);
		\node[above=0.6cm of engine](uncertNode){};
		\path [line] (uncertNode) -- node[right, pos=0.3] {{\huge${w_k}$}} (engine);
		\node[above=0.6cm of sensor](noiseNode){};
		\path [line] (noiseNode) -- node[right, pos=0.3] {{\huge${e_k}$}} (sensor);
		\node[left=1.3 of controller] (mal1) {\includegraphics[,width=.2\textwidth]{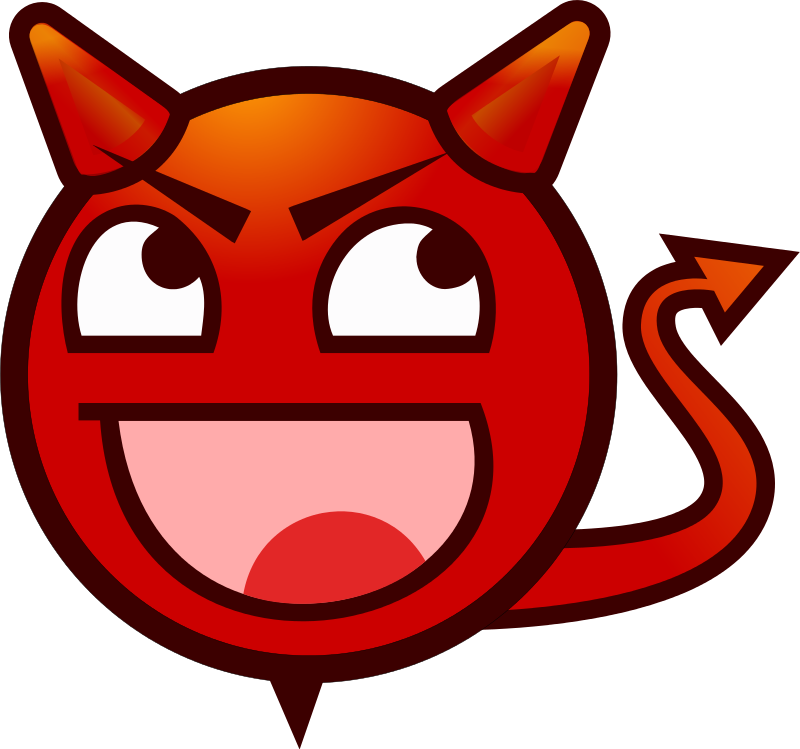}};
		\node[right=1.4 of controller] (mal2) {\includegraphics[,width=.2\textwidth]{./malware.png}};
		\end{tikzpicture}	
	\end{minipage}
	\caption{MITM attacks to sensor readings and control commands}
	\label{fig:threat-model}
\end{figure}

In order to make security assessments on our \CPS{s}, we adopt a
well-known approach called \emph{Generalized Non Deducibility on
	Composition (GNDC)}~\cite{FM99}. 
Thus, in our calculus \cname{}, we say that a \CPS{} $\mathit{Sys}$ 
\emph{tolerates} an attack $A$ if
\begin{displaymath}
\mathit{Sys} \: \parallel \: A \q \sqsubseteq \q \mathit{Sys} \,  .
\end{displaymath}
In this case, the presence of the attack $A$, does not change the (physical and logical) observable behaviour of the system $\mathit{Sys}$, and the attack can be considered harmless.

On the other hand, we say that a \CPS{} $\mathit{Sys}$
is \emph{vulnerable} to an attack $A$ of class $C \in [\I \rightarrow {\cal P}(m..n)]$ if there is a time interval $m'..n'$ in which the attack becomes observable (obviously, $m' \geq m$). Formally, we write: 
\begin{displaymath}
\mathit{Sys} \: \parallel \: A \q \sqsubseteq_{m'..n'} \q \mathit{Sys} \,  .
\end{displaymath}

We provide sufficient criteria to prove attack tolerance/vulnerability to attacks of an arbitrary class $C$. We define a notion of \emph{most powerful physics-based attack} of a given class $C$, $\mathit{Top}(C)$, and prove that if a \CPS{} tolerates $\mathit{Top}(C)$ then it tolerates all attacks $A$ of class $C$ (and of any weaker class\footnote{Intuitively, attacks of classes weaker than $C$ can do less with respect to attacks  of class $C$.}). Similarly, if a \CPS{} is vulnerable to $\mathit{Top}(C)$, in the time interval $m'..n'$, then no attacks of class $C$  (or weaker) can affect the system out of that time interval. This is very useful when checking for attack tolerance/vulnerability with respect to all attacks of a given class $C$.

As a \emph{third contribution}, we formalise how to estimate the \emph{impact of a successful attack} on a \CPS{}.  
As expected, \emph{risk assessment} in industrial \CPS{s} is a crucial phase preceding any defence strategy implementation~\cite{Swedish-guide}. The objective of this phase is to prioritise among vulnerabilities; this is done based on the likelihood that vulnerabilities are exploited, and the impact on the system under attack if exploitation occurs. In this manner, the resources can then be focused on preventing the most critical vulnerabilities~\cite{Milosevic2018a}. 
We provide a \emph{metric} to estimate the \emph{maximum perturbation} introduced in the system under attack with respect to its genuine behaviour, according to its evolution law and the uncertainty of the model. Then, we prove that the impact of the most powerful attack $\mathit{Top}(C)$ represents an upper bound for the impact of any attack $A$ of class $C$ (or weaker).

Finally, as a \emph{fourth contribution}, we formalise a \emph{running example} in \Uppaal{} SMC~\cite{David:2015:UST:2802769.2802840}, the statistical extension of the \Uppaal{} model checker~\cite{DBLP:conf/sfm/BehrmannDL04} supporting  the analysis of systems expressed as composition of \emph{timed and/or probabilistic automata}. 
Our goal is to test \Uppaal{} SMC as an automatic tool for the \emph{static security analysis} of a simple but significant \CPS{} exposed to a number of different physics-based attacks with different impacts on the system under attack. Here, we wish to remark that while we have kept our running example simple, it is actually non-trivial and designed to describe a wide number of attacks, as will become clear below.

This paper extends and supersedes a preliminary conference version that appeared in~\cite{CSF2017}. 

The \Uppaal{} SMC models of our system and the attacks that we have found are available at the repository {\small \texttt{https://bitbucket.org/AndreiMunteanu/cps\_smc/src/}}. 

\subsection{Organisation}
In Section~\ref{sec:calculus}, we give syntax and semantics of \cname{}.
In Section~\ref{sec:running_example}, we provide our running example and its formalisation in \Uppaal{} SMC.
In Section~\ref{sec:cyber-physical-attackers}, we first define our threat model for physics-based attacks, then we use \Uppaal{} SMC to carry out a security analysis of our running example when exposed to three different attacks, and, finally, we provide sufficient criteria for attack tolerance/vulnerability, 
based  on a proper notion of most powerful attack.
In Section~\ref{sec:impact}, we estimate the impact of attacks on \CPS{s} and prove that the most powerful attack of a given class has the maximum impact with respect to all attacks of the same class (or of a weaker one). 
In Section~\ref{sec:conclusions}, we draw conclusions and discuss related \nolinebreak and 
\nolinebreak future \nolinebreak work. 


\section{The Calculus}

\label{sec:calculus}

In this section, we introduce our \emph{Calculus of Cyber-Physical Systems and Attacks}, \cname{}, which extends the \emph{Calculus of Cyber-Physical Systems}, defined in our companion papers~\cite{LaMe17,LMT2018}, with specific features to formalise and study attacks to physical devices. 

Let us start with some preliminary notation. 

\subsection{Syntax of \cname{}}
\begin{notation}
	We use $x, x_k$ 
	for \emph{state variables} (associated to physical states of systems), 
	$c,d$ 
	for \emph{communication channels}, 
	$a, a_k$ 
	for \emph{actuator devices}, 
	and $s,s_k$
	or \emph{sensors devices}.
	
	\emph{Actuator names} are metavariables for actuator devices like
	$\mathit{valve}$, $\mathit{light}$, etc. Similarly, \emph{sensor names}
	are metavariables for sensor devices, e.g., a sensor $\mathit{thermometer}$ that measures a state variable called $\mathit{temperature}$, with a given precision. 
	
	\emph{Values}, ranged over by $v,v',w$, 
	are built from basic values, such as Booleans, integers and real numbers; they also include names.
	
	Given a generic set of names $\mathcal N $, we write $\mathbb{R}^{\mathcal N} $ to denote the set of functions assigning a real value to each name in $\mathcal N$. For $\xi \in \mathbb{R} ^{\mathcal N}$, $n \in \mathcal  N$ and $v \in \mathbb{R} $, we write $\xi [n \mapsto v]$ to denote the function $\psi \in \mathbb{R} ^{\mathcal N}$ such that $\psi(m)=\xi(m)$, for any $m \neq n$, and $\psi(n)=v$.
	Given two generic functions $\xi_1$ and $\xi_2$ with disjoint domains ${\mathcal N}_1$ and ${\mathcal N}_2$, respectively,   we denote with $\xi_1 \cup \xi_2$ the function 
	such that
	$(\xi_1 \cup \xi_2) (n)=\xi_1(n)$, if $n \in {{\mathcal N}_1} $, and 
	$(\xi_1 \cup \xi_2) (n)=\xi_2(n)$, if $n \in {{\mathcal N}_2} $.
\end{notation}

In general, a cyber-physical system consists of: 
\begin{inparaenum}[(i)]
	\item a \emph{physical component} (defining state variables, physical devices, physical evolution, etc.) and 
	\item a \emph{cyber (or logical) component} that interacts with the physical devices (sensors and actuators) and communicates with other cyber components of the same or of other \CPS{s}.
\end{inparaenum}

Physical components in \cname{} are given by two sub-components: 
\begin{inparaenum}[(i)] 
	\item the \emph{physical state},  which is supposed to change at runtime, and
	\item the \emph{physical environment}, which contains static information.\footnote{Actually, this information is periodically updated (say, every six months) to take into account possible drifts of the system.}
\end{inparaenum}
\begin{definition}[Physical state]
	\label{def:physical-state}
	Let $\mathcal X$ be a set of state variables, $\mathcal S$ be a set of sensors,
	and $ \mathcal A$ be a set of actuators. A \emph{physical state} $S$ is a triple
	$\stateCPS {\statefun{}} {\sensorfun{}} {\actuatorfun{}} $,
	where:
	\begin{itemize}
		\item $\statefun{} \in \mathbb{R}^{\mathcal X} $ is the
		\emph{state function},
		\item $\sensorfun{} \in \mathbb{R}^{\mathcal S}$ is the \emph{sensor
			function}, 
		\item $\actuatorfun{} \in \mathbb{R}^{\mathcal A} $ is the
		\emph{actuator function}.
	\end{itemize}
	All functions defining a physical state are \emph{total}. 
\end{definition}

The \emph{state function} $\statefun{}$ returns the current value  associated to each variable in $\mathcal X$, the \emph{sensor function} $\sensorfun{}$ returns the current value associated to each sensor in $\mathcal S$ and the \emph{actuator function} $\actuatorfun{}$ returns the current value associated to each actuator in $\mathcal A$.

\begin{definition}[Physical environment]
	\label{def:physical-env}
	Let $\mathcal X$ be a set of state variables, $\mathcal S$ be a set of sensors, and
	$\mathcal A$ be a set of actuators.
	A \emph{physical environment} $E$ is a 6-tuple 
	$\envCPS
	{\evolmap{}}
	{\measmap{}}
	{\invariantfun{}}
	{\safefun{}}
	{\uncertaintyfun{}}{\errorfun{}}
	$,
	where:
	\begin{itemize}
		\item $\evolmap{}: \mathbb{R}^{\mathcal X} \times
		\mathbb{R}^{\mathcal A} \times \mathbb{R}^{\mathcal X} \rightarrow
		2^{\mathbb{R}^{\mathcal X} }$ is the \emph{evolution map}, 
		
		\item $\measmap{}: \mathbb{R}^{\mathcal X} \times
		\mathbb{R}^{\mathcal S} \rightarrow 2^{\mathbb{R}^{\mathcal S} }$ is
		the \emph{measurement map}, 
		
		\item  $\invariantfun{} \in 2^{\mathbb{R}^{{\mathcal X} }}$  is the \emph{invariant set}, 
		
		\item $\safefun{} \in 2^{\mathbb{R}^{{\mathcal X} }}$ is the \emph{safety set}, 
		
		\item $\uncertaintyfun{} \in \mathbb{R}^{\mathcal X} $ is the \emph{uncertainty function},
		
		\item $\errorfun{} \in \mathbb{R}^{\mathcal S}$ is the \emph{sensor-error function}. 
		
	\end{itemize}
	All functions defining a physical environment are \emph{total functions}.
\end{definition}

The \emph{evolution map} $\evolmap{}$  models the \emph{evolution law} of the physical system, where changes made on actuators may reflect on state variables. Given a state function, an actuator function, and an uncertainty function, the \emph{evolution map} $\evolmap{}$ returns the set of next \emph{admissible state functions}. Since we assume an \emph{uncertainty} in our models, $\evolmap{}$ does not return a single state function but a set of possible state functions.

The \emph{measurement map} $\measmap{}$ returns the set of
next \emph{admissible sensor functions} based on the current state function. Since we assume error-prone sensors, $\measmap{}$ does not return a single sensor function but a set of possible sensor functions.

The \emph{invariant set} $\invariantfun{}$ represents the set of
state functions that satisfy the invariant of the system. A
\CPS{} that gets into a  physical state with a state function that does not satisfy the invariant is in \emph{deadlock}. Similarly, 
the \emph{safety set} $\safefun{}$ represents the set of state functions that satisfy the safety conditions of the system. Intuitively, if a \CPS{} gets into an unsafe state, then its functionality may get compromised.

The \emph{uncertainty function}
$\uncertaintyfun{}$ returns the uncertainty (or accuracy) associated to each state variable. Thus, given a state variable $x \in \mathcal X$,
$\uncertaintyfun{}(x)$ returns the maximum distance between the real value
of $x$, in an arbitrary moment in time, 
and its representation in the model. 
For $\uncertaintyfun{}, \uncertaintyfun'{} \in \mathbb{R}^{\mathcal X}$, we
will write $\uncertaintyfun{} \leq \uncertaintyfun'{}$ if
$\uncertaintyfun{}(x) \leq \uncertaintyfun'{}(x)$, for any $x \in
\mathcal X$. 
The evolution map  $\evolmap{}$ is obviously \emph{monotone} with respect to uncertainty: if
$\uncertaintyfun{} \leq \uncertaintyfun'{}$ then $\evolmap{}(\statefun{},
\actuatorfun{}, \uncertaintyfun{}) \subseteq \evolmap{}(\statefun{},
\actuatorfun{}, \uncertaintyfun'{})$. 

Finally, the \emph{sensor-error function} $\errorfun{}$ returns the maximum error associated to each sensor in $\mathcal  S$.

\medskip
Let us now define formally the cyber component of a \CPS{} in \cname{}.
Our (logical) processes build on Hennessy and Regan's \emph{Timed Process Language TPL}~\cite{HR95}, basically, CCS enriched with a discrete notion of time. We extend TPL with two main ingredients: 
\begin{itemize}[noitemsep]
	\item two constructs to read values detected at sensors and  write values on actuators, respectively; 
	\item special constructs to represent malicious activities on physical 
	devices. 
\end{itemize}
The remaining constructs are the same as those of TPL. 
\begin{definition}[Processes]
	\emph{Processes} are defined as follows: 
	\begin{displaymath}
	\begin{array}{rl}	
	P,Q \Bdf & \nil \q \big| \q \tick.P \q \big| \q P \parallel Q \q 
	\big| \q \pi.P \q \big| \q \phi.P \q \big| \q  \timeout{\mu.P}{Q} 
	\q \big|  \q \ifelse b P Q \q \big| \q P{\setminus} c \q \big| \q  H \langle \tilde{w} \rangle \\[1pt]
	\pi \Bdf & \LIN{c}{x} \q \big| \q \OUT{c}{v}     \\[1pt]
	\phi \Bdf & \rsens x s \q \big| \q   \wact v a \\[1pt]
	\mu \Bdf & \sniff x s \q \big| \q  \drop x a \q \big| \q  \forge v p \, . 
	\end{array}
	\end{displaymath}
\end{definition}

We write $\nil$ for the \emph{terminated process}. The process $\tick.P$
sleeps for one time unit and then continues as $P$. We write $P \parallel Q$ to denote the \emph{parallel composition} of concurrent \emph{threads} $P$ and $Q$.
The process $ \pi.P$ denotes \emph{channel transmission}. 
The construct $\phi.P$ denotes  activities on \emph{physical devices}, i.e., \emph{sensor reading} and \emph{actuator writing}. 
The process $\timeout{\mu.P}{Q}$ denotes MITM malicious activities under timeout targeting physical devices (sensors and actuators). More precisely, we support \emph{sensor sniffing}, \emph{drop of actuator commands}, and \emph{integrity attacks on data coming from sensors and addressed to actuators}. 
Thus, for instance,  $\timeout{\drop x a . P}Q$ drops  a command on the actuator $a$ supplied by the controller in the current time slot;  otherwise, if there are no commands on $a$, it moves to the next time slot and evolves into $Q$.

The process $P{\setminus}c$ is the channel restriction operator of CCS. 
We sometimes write $P{\setminus}\{c_1, c_2, \ldots, c_n \}$ to mean
$P{\setminus}{c_1}{\setminus}{c_2}\cdots{\setminus}{c_n}$. 
The process $\ifelse b P Q$ is the standard conditional, where $b$ is a decidable guard. 
In processes of the form $\tick.Q$ and $\timeout {\mu.P} Q$, the occurrence of $Q$ is said to be \emph{time-guarded}. 
The process $H \langle \tilde{w} \rangle$ denotes (guarded) recursion. 

We  assume a set of \emph{process identifiers} ranged over by $H,H_1,H_2$.
We write $H \langle w_1,\ldots, w_k \rangle$ to denote a recursive process $H$ defined via an equation $H(x_1,\ldots, x_k) = P$, where (i) the tuple $x_1,\ldots, x_k$ contains all the variables that appear free in $P$, and (ii) $P$ contains only  guarded occurrences of the process identifiers, such as $H$ itself. We say that recursion is \emph{time-guarded} if $P$ contains only time-guarded occurrences of the process identifiers. Unless explicitly stated our recursive processes are always time-guarded.

In the  constructs ${\LIN c x. P}$, $\rsens x s . P$, $\timeout{\sniff x s . P}Q$ and $\timeout{ \drop x a. P}Q$  the variable $x$ is said to be \emph{bound}. 
This gives rise to the standard notions of \emph{free/bound (process) variables} and \emph{$\alpha$-conversion}. 
A term is \emph{closed} if it does not contain free variables, and 
we assume to always work with closed processes: the absence of free variables is preserved at run-time. As further notation, we write $T{\subst v x}$ for the substitution of all occurrences of the free variable $x$ in $T$ with the value $v$.

Everything is in place to provide the definition of cyber-physical systems expressed in \cname{}. 

\begin{definition}[Cyber-physical system]
	Fixed a set of state variables $\mathcal X$, a set of sensors $\mathcal S$, and a set of actuators $\mathcal A$, a \emph{cyber-physical system} in \cname{} is given by two main components:
	
	\begin{itemize}
		\item a \emph{physical component} consisting of 
		\begin{itemize}
			\item a \emph{physical environment} $E$ defined on $\mathcal X$, $\mathcal S$, and $\mathcal A$, and 
			\item a \emph{physical state} $S$ recording the current values associated to the state variables in $\mathcal X$, the sensors in $\mathcal S$, and the actuators in $\mathcal A$;  
		\end{itemize}
		\item a \emph{cyber component} $P$ that interacts with the sensors in $\mathcal S$ and the actuators $\mathcal A$, and can communicate, via channels, with other cyber components of the same or of other \CPS{s}.
	\end{itemize}
	We write $\confCPS {E;S} P$ to denote the resulting \CPS{}, and use $M$ and $N$ to range over \CPS{s}. Sometimes, when the physical environment $E$ is clearly identified, we write $\confCPS S P$ instead of $\confCPS {E;S} P$. \CPS{s} of the form $\confCPS S P$ are called environment-free \CPS{s}. 
\end{definition}

The syntax of our \CPS{s} is slightly too permissive as a process 
might use sensors and/or actuators that are not defined in the physical state. To rule out ill-formed \CPS{s}, we use the following definition. 
\begin{definition}[Well-formedness]
	\label{def:well-formedness}
	Let $E = \envCPS
	{\evolmap{}}
	{\measmap{}}
	{\invariantfun{}}
	{\safefun{}}
	{\uncertaintyfun{}}
	{\errorfun{}}
	$ 
	be a physical environment, let $S = \stateCPS
	{\statefun{}}
	{\sensorfun{}}
	{\actuatorfun{}}
	$ be a physical state defined on a set of physical variables $\mathcal X$, a set of sensors ${\mathcal S}$, and a set of actuators ${\mathcal A}$, and let $P$ be a process. 
	The \CPS{} $\confCPS {E;S} P$ is said to be \emph{well-formed} if: 
	\begin{inparaenum}[(i)]
		\item any sensor mentioned in $P$ is in the domain of the function $\sensorfun{}$; 
		\item any actuator mentioned in $P$ is in the domain of the function $\actuatorfun{}$. 
	\end{inparaenum}
\end{definition}
In the rest of the paper, we will always work with well-formed \CPS{s} and use the following abbreviations.

\begin{notation}
	We write $\mu.P$ for 
	the process defined via the  equation $Q = 
	\timeout{\mu.P}Q$, where $Q$ does not occur in $P$. 
	Further, we write
	\begin{itemize}[noitemsep]
		\item $\timeout{\mu}Q$ as an abbreviation for $\timeout{\mu.\nil}{Q}$, 
		\item $\timeout{\mu.P}{}$ as an abbreviation for $\timeout{\mu.P}{\nil}$,
		\item $\OUTCCS c$ and $\LINCCS c$, when channel $c$ is used for pure synchronisation, 
		\item $\tick^{k}.P$ as a shorthand for $\tick \ldots \tick.P$, where the prefix $\tick$ appears $k \geq 0$ consecutive times. 
	\end{itemize} 
	Finally, let $M = \confCPS {E;S}  P$, we write 
	$M\parallel Q$ for $\confCPS {E;S}  { (P\parallel Q) }$, and $M{\setminus}c$ for $\confCPS {E;S} {(P {\setminus}c)}$. 
\end{notation}

\subsection{Labelled transition semantics}
\label{lab_sem}

\begin{table*}[t]
	{\small 
		\caption{LTS for processes}
		\label{tab:lts_processes} 
		\begin{displaymath}
		\begin{array}{l@{\hspace*{10mm}}l}
		\Txiom{Inpp}
		{-}
		{ { {\LIN c x .P} } \trans{\inp c v}    {P{\subst v x}}  }
		& 
		\Txiom{Outp}
		{-}
		{  {\OUT c v .P}  \trans{\out c v}   P} 
		\\[13pt]
		\Txiom{Com}
		{ P \trans{\out c v}  { P'}  \Q  Q \trans{\inp c v}  { Q'} }
		{ P \parallel  Q \trans{\tau}  {P'\parallel Q'}}
		&
		\Txiom{Par}
		{ P \trans{\lambda}  P' \Q \lambda \neq  \tick }
		{ {P\parallel Q} \trans{\lambda} {P'\parallel Q}}
		\\[13pt]
		\Txiom{Read}
		{ - }
		{ { {\rsens x s .P} } \trans{\rcva s v}    {P{\subst v x}}  }
		&
		\Txiom{Write}
		{ - }  
		{\wact v a .P  \trans{\snda a v}   P}
		\\[13pt]
		\Txiom{{\Lightning}Sniff{\Lightning}}
		{ - }
		{ { \timeout{\sniff x s .P}Q } \trans{\rcva {{\mbox{\Lightning}}s} v}    {P{\subst v x}}  }
		&
		\Txiom{{\Lightning}Drop{\Lightning}}
		{ - } 
		{ { \timeout{\drop x a .P}Q } \trans{\rcva {\mbox{\Lightning}a} v}    {P{\subst v x}}  }
		\\[13pt]
		\multicolumn{2}{c}{\Txiom{{\Lightning}Forge{\Lightning}}
			{ p \in \{ s, a \}} 
			{\timeout{\forge v p .P}Q  \trans{\snda {\mbox{\Lightning}p} v}   P}}
		\\[13pt]
		\Txiom{{\Lightning}ActDrop{\Lightning}}
		{ P \trans{\snda {a} v}  {P'}  \Q  Q \trans{\rcva {\mbox{\Lightning}a} v}  { Q'}
		}
		{ P \parallel  Q \trans{\tau}  {P'\parallel Q'}}
		&
		\Txiom{{\Lightning}SensIntegr{\Lightning}}
		{ P \trans{\snda {\mbox{\Lightning}s} v}  {P'}  \Q  Q \trans{\rcva s v}  { Q'} 
		}
		{ P \parallel  Q \trans{\tau}  {P'\parallel Q'}} 
		\\[13pt]
		\Txiom{Res}{P \trans{\lambda} P' \Q \lambda \not\in \{ {\inp c v}, {\out c v} \}}{P {\setminus}c \trans{\lambda} {P'}{\setminus}c}
		&
		\Txiom{Rec}
		{  P{\subst {\tilde{w}} {\tilde{x}}} \trans{\lambda}  Q \Q H(\tilde{x})=P}
		{ H \langle \tilde{w} \rangle  \trans{\lambda}  Q}
		\\[13pt] 
		\Txiom{Then}{\bool{b}=\true \Q P \trans{\lambda} P'}
		{\ifelse b P Q \trans{\lambda} P'}
		&
		\Txiom{Else}{\bool{b}=\false \Q Q \trans{\lambda} Q'}
		{\ifelse b P Q \trans{\lambda} Q'}
		\\[13pt]
		\Txiom{TimeNil}{-}
		{ \nil \trans{\tick}  \nil}
		& 
		\Txiom{Sleep}
		{-}
		{  { \tick.P} \trans{\tick}  P}
		\\[13pt]
		\Txiom{Timeout}
		{-}
		{  {\timeout{\mu.P}{Q} }   \trans{\tick}  Q}
		&
		\Txiom{TimePar}
		{
			P \trans{\tick}  {P'}  \Q 
			Q \trans{\tick} {Q'} 
		}
		{
			{P \parallel Q}   \trans{\tick}  { P' \parallel Q'}
		}
		\end{array}
		\end{displaymath}
	}
\end{table*}

In this subsection, we provide the dynamics of \cname{} in terms of a \emph{labelled transition system (LTS)} in the SOS style of Plotkin. First, we give in Table~\ref{tab:lts_processes} an LTS for logical processes, then in Table~\ref{tab:lts_systems} we lift transition rules from processes to environment-free \CPS{s}. 

\begin{table*}[t]
	{\small 
		\caption{LTS for \CPS{s} $\confCPS S  P$ parametric on an environment 
			$E = \envCPS
			{\evolmap{}}
			{\measmap{}}
			{\invariantfun{}}
			{\safefun{}}
			{\uncertaintyfun{}}
			{\errorfun{}}
			$}
		\label{tab:lts_systems} 
		\begin{displaymath}
		\begin{array}{c}
		\Txiom{Out}
		{S=\langle \statefun{} ,  \sensorfun{} , \actuatorfun{} \rangle \Q P \trans{\out c v}  P' \Q \statefun{} \in \invariantfun{}}
		{\confCPS {S}  P   \trans{\out c v}   \confCPS {S} {P' }}
		\Q\Q\Q
		\Txiom{Inp}
		{S=\langle \statefun{} ,  \sensorfun{} , \actuatorfun{} \rangle \Q P  \trans{\inp c v}  P' \Q \statefun{} \in \invariantfun{}}
		{\confCPS {S}  P    \trans{\inp c v}  \confCPS {S}  {P' }}
		\\[14pt]
		\Txiom{SensRead}{P \trans{\rcva s v} P'  \Q \sensorfun{}(s)=v \Q P 
			\ntrans{\snda{\mbox{\Lightning}s}{v}}\Q  \statefun{}  \in \invariantfun{} 
		}
		{\confCPS {\langle \statefun{} ,  \sensorfun{} , \actuatorfun{} \rangle}
			P \trans{\tau} \confCPS {\langle \statefun{} ,  \sensorfun{} , \actuatorfun{} \rangle}  {P'}}
		\\[14pt]
		\Txiom{{\Lightning}SensSniff{\Lightning}}{P \trans{\rcva {\mbox{\Lightning}s} v} P'   \Q \sensorfun{}(s)=v \Q 
			\statefun{} \in \invariantfun{} 
		}
		{\confCPS {\langle \statefun{} ,  \sensorfun{} , \actuatorfun{} \rangle}  P \trans{\tau} \confCPS {\langle \statefun{} ,  \sensorfun{} , \actuatorfun{} \rangle}  {P'}}
		\\[14pt]
		\Txiom{ActWrite}{P \trans{\snda a v} {P'}  \Q  \actuatorfun'{}=\actuatorfun{}[a \mapsto v] \Q  P \ntrans{\rcva{\mbox{\Lightning}a}{v}} \Q  \statefun{}  \in \invariantfun{}
		}
		{\confCPS {\langle \statefun{} ,  \sensorfun{} , \actuatorfun{} \rangle}  P \trans{\tau} \confCPS {\langle \statefun{} ,  \sensorfun{} , \actuatorfun'{} \rangle} {P'}}
		\\[14pt]
		\Txiom{{\Lightning}ActIntegr{\Lightning}}{P \trans{\snda {\mbox{\Lightning}a} v} {P'}   \Q \actuatorfun'{}=\actuatorfun{}[a \mapsto v]\Q 
			\statefun{}  \in \invariantfun{} }
		{\confCPS {\langle \statefun{} ,  \sensorfun{} , \actuatorfun{} \rangle}  P \trans{\tau} \confCPS {\langle \statefun{} ,  \sensorfun{} , \actuatorfun'{} \rangle} {P'}}
		\\[15pt]
		\Txiom{Tau}{P \trans{\tau} P' \Q  \statefun{} \in \invariantfun{}}
		{ \confCPS {\langle \statefun{} ,  \sensorfun{} , \actuatorfun{} \rangle}  P \trans{\tau} \confCPS {\langle \statefun{} ,  \sensorfun{} , \actuatorfun{} \rangle}  {P'}} 
		\Q\Q
		\Txiom{Deadlock}
		{S=\langle \statefun{} ,  \sensorfun{} , \actuatorfun{} \rangle \Q \statefun{} \not \in \invariantfun{}}
		{ \confCPS {S}  P \trans{\dead} \confCPS {S}  {P}}
		\\[14pt]
		\Txiom{Time}{ P \trans{\tick} {P'} \Q
			S=\langle \statefun{} ,  \sensorfun{} , \actuatorfun{} \rangle \Q
			S' \in \operatorname{next}(E;S) \Q \statefun{}\in \invariantfun{} }
		{\confCPS {S}  P \trans{\tick} \confCPS {S'}  {P'}}
		\\[14pt]
		\Txiom{Safety}
		{S=\langle \statefun{} ,  \sensorfun{} , \actuatorfun{} \rangle \Q \statefun{} \not \in \safefun{} \Q \statefun{} \in \invariantfun{}}
		{ \confCPS {S}  P \trans{\unsafe} \confCPS {S} {P}}
		\end{array}
		\end{displaymath}
	}
\end{table*}

In Table~\ref{tab:lts_processes}, the meta-variable $\lambda$ ranges over labels in the set 
$\{\tick, \tau, {\out c v}, {\inp c v}, \allowbreak \snda a v,\rcva s v, \snda {\mbox{\Lightning}p} v, \rcva {\mbox{\Lightning}p} v\}$. Rules \rulename{Outp}, \rulename{Inpp} and \rulename{Com} serve to model channel communication, on some channel $c$. Rules \rulename{Read} and \rulename{Write} 
denote sensor reading and actuator writing, respectively. The following 
three rules model three different MITM malicious activities: sensor sniffing, dropping of actuator commands, and integrity attacks on data coming from sensors or addressed to actuators. 
In particular, rule \rulename{$\mbox{\Lightning}$ActDrop$\mbox{\,\Lightning}$} models a \emph{DoS attack to the actuator $a$}, where the update request of the controller is dropped by the attacker and it never reaches the actuator, whereas 
rule \rulename{$\mbox{\Lightning}$SensIntegr$\mbox{\,\Lightning}$} models an
\emph{integrity attack on sensor $s$}, as the controller of $s$ is
supplied with a fake value $v$ forged by the attack.  
Rule \rulename{Par} propagates untimed actions over parallel components.
Rules \rulename{Res}, \rulename{Rec}, \rulename{Then} and \rulename{Else} are standard. The following four rules 
\rulename{TimeNil},
\rulename{Sleep},
\rulename{TimeOut} and 
\rulename{TimePar} 
model the passage of time. For simplicity, we omit the symmetric counterparts of the rules \rulename{Com}, \rulename{$\mbox{\Lightning}$ActDrop$\mbox{\,\Lightning}$}, \rulename{$\mbox{\Lightning}$SensIntegr$\mbox{\,\Lightning}$}, 
and \rulename{Par}.

In Table~\ref{tab:lts_systems}, we lift the transition rules from processes to environment-free \CPS{s} of the form $\confCPS S P$ for $S = \langle \statefun{}, \sensorfun{}, \actuatorfun{} \rangle$. The transition rules are parametric on a physical environment $E$. Except for rule \rulename{Deadlock}, all rules have a common
premise $\statefun{} \in \invariantfun{}$: a system can evolve only if the invariant is satisfied by the current physical state. 
Here, actions, ranged over by $\alpha$, are in the set $\{\tau,
{\out c v}, {\inp c v}, \tick , \dead , \unsafe\}$. These actions denote: internal activities ($\tau$); channel transmission (${\out c v}$ and ${\inp c v}$); the passage of time ($\tick$); and two specific  physical events: system deadlock ($\dead$) and the violation of the safety conditions ($\unsafe$).  
Rules \rulename{Out} and \rulename{Inp} model transmission and reception, with an external system, on a channel $c$. Rule \rulename{SensRead} models the reading of the current data detected at a sensor $s$; 
here, the presence of a malicious action $\snda {\mbox{\Lightning}s} w$ would prevent the reading of the sensor. We already said that rule
\rulename{$\mbox{\Lightning}$SensIntegr$\mbox{\,\Lightning}$} of
Table~\ref{tab:lts_processes} models integrity attacks on a sensor
$s$. However, together with rule \rulename{SensRead}, it also serves
to implicitly model \emph{DoS attacks on a sensor $s$}, as the controller of $s$ cannot read its correct value if the attacker is currently supplying a fake value for it.
Rule \rulename{{\Lightning}SensSniff{\Lightning}} allows the
attacker to read the confidential value detected at a sensor $s$.
Rule~\rulename{ActWrite} models the writing of a value $v$ on an 
actuator $a$; here, the presence of an attack capable of performing a drop action $\rcva{\mbox{\Lightning}a}{v}$ prevents the access to the actuator by the controller. Rule
\rulename{{\Lightning}ActIntegr{\Lightning}} models a
\emph{MITM integrity attack to an actuator $a$}, as the actuator is provided with a value forged by the attack. 
Rule \rulename{Tau} lifts non-observable actions from processes to
systems. This includes communications channels and attacks' accesses to
physical devices. A similar lifting occurs in rule \rulename{Time} for timed actions, where $\operatorname{next}(E;S)$ returns the set of possible physical states for the next time slot. Formally, for $E = \envCPS
{\evolmap{}}
{\measmap{}}
{\invariantfun{}}
{\safefun{}}
{\uncertaintyfun{}}
{\errorfun{}}$ and $S = \langle \statefun{} ,  \sensorfun{} , \actuatorfun{} \rangle$, we
define:
\begin{displaymath}
\operatorname{next}(E;S) \deff \big \{ \langle \statefun'{} ,  \sensorfun'{} , \actuatorfun'{} \rangle \, : \: \statefun'{} \in   \evolmap{}(\statefun{}, \actuatorfun{}, \uncertaintyfun{}) \: \wedge \; \sensorfun'{} \in \measmap{}(\statefun'{},\errorfun{}) \: \wedge \: \actuatorfun'{}=\actuatorfun{} \big \} \, . 
\end{displaymath}
Thus, by an application of rule \rulename{Time} a \CPS{} moves to the next physical state, in the next time slot. Rule~\rulename{Deadlock} is introduced  to signal the violation of the invariant. When the invariant
is violated, a system deadlock occurs and then, in \cname{}, the system
emits a special action $\dead$, forever. Similarly, rule~\rulename{Safety} is introduced to detect the violation of safety conditions. In this case, the system may emit a special action $\unsafe$ and then continue its evolution.

Summarising, in the LTS of Table~\ref{tab:lts_systems} we define transitions rules of the form $\confCPS {S}{P} \trans{\alpha} \confCPS {S'} {P'}$, parametric on some physical environment $E$. As physical environments do not change at runtime, $\confCPS {S}{P} \trans{\alpha} \confCPS {S'} {P'}$ entails $\confCPS {E;S}{P} \trans{\alpha} \confCPS {E;S'} {P'}$, thus providing the  LTS for all \CPS{s} in \cname{}.

\begin{remark}
	Note that our operational semantics ensures that malicious actions of the form $\snda {\mbox{\Lightning}s} v$ (integrity/DoS attack on sensor $s$) or $\rcva {\mbox{\Lightning}a} v$ (DoS attack on actuator $a$) have a pre-emptive power. 
	These attacks can always prevent the regular access to a physical device by its controller. 
\end{remark}

\subsection{Behavioural semantics} 
\label{sec:trace-semantics}
Having defined the actions that can be performed by a \CPS{} of the form $\confCPS{E;S}{P}$, we can easily concatenate these actions to define the possible \emph{execution traces} of the system. Formally, given a trace  $t = \alpha_1 \ldots \alpha_n$, we will write $\trans{t}$ as an abbreviation for $\trans{\alpha_1}\ldots \trans{\alpha_n}$, and we will use the function $\#\tick(t)$ to get the number of occurrences of the  action $\tick$  in $t$.

The notion of trace allows us to provide a formal definition of system soundness: a \CPS{} is said to be \emph{sound} if it never deadlocks and never violates the safety conditions.
\begin{definition}[System soundness]
	Let $M$ be a well-formed \CPS{}. We say that $M$ is \emph{sound} if whenever $M \trans{t} M'$, for some $t$, the actions $\dead$ and $\unsafe$ never occur in $t$. 
\end{definition}
In our security analysis, we will always focus on sound \CPS{s}.

We recall that the \emph{observable activities} in \cname{} are: time
passing, system deadlock, violation of safety conditions, and channel communication.
Having defined a labelled transition semantics, we are ready to formalise our behavioural semantics, based on execution traces. 

We adopt a standard notation for weak transitions: we write $\Trans{}$ 
for $(\trans{\tau})^*$, whereas $\Trans{\alpha}$ means $\ttranst{\alpha}$, and finally $\ttrans{\hat{\alpha}}$ denotes $\Trans{}$ if $\alpha=\tau$ and $\ttrans{\alpha}$ otherwise. Given a trace $t = \alpha_1 {\ldots} \alpha_n$, we write $\Trans{\hat{t}}$ as an abbreviation for $\Trans{\widehat{\alpha_1}} {\ldots} \Trans{\widehat{\alpha_n}}$. 

\begin{definition}[Trace preorder]
	\label{Trace-equivalence}
	We write $M \sqsubseteq N$ if whenever $M \trans{t}
	M'$, for some $t$, there is $N'$ such that $N\Trans{\hat{t}}N'$. 
\end{definition}
\begin{remark}
	\label{rem:deadlock}
	Unlike other process calculi, in \cname{} our trace preorder is able to observe (physical) deadlock due to the presence of the rule \rulename{Deadlock} and the special action $\dead$: whenever $M \sqsubseteq N$ then $M$ eventually 
	deadlocks if and only if $N$ eventually deadlocks (see Lemma~\ref{lem:trace-deadlock} in the appendix).
\end{remark}

Our trace preorder can be used for \emph{compositional reasoning} in those contexts that don't interfere on physical devices (sensors and actuators) while they may interfere on logical components (via channel communication). 
In particular, trace preorder is preserved by parallel composition of \emph{physically-disjoint} \CPS{s}, by parallel composition of \emph{pure-logical} processes, and by channel restriction. Intuitively, two \CPS{s} are physically-disjoint if they have different plants but they may share logical channels for communication purposes. More precisely, physically-disjoint \CPS{s} have disjoint state variables and disjoint physical devices (sensors and actuators).  
As we consider only well-formed \CPS{s} (Definition~\ref{def:well-formedness}), this ensures that a \CPS{} cannot physically interfere with a parallel \CPS{} by acting on its physical devices. 

Formally, let   
$S_i = \stateCPS {\statefun^i{}} {\sensorfun^i{} } {\actuatorfun^i{} }$ and 
$E_i = \envCPS
{\evolmap^i{} }
{\measmap^i{} }   
{\invariantfun^i{}}
{\safefun^i{}}
{\uncertaintyfun^i{}}
{\errorfun^i{}} 
$ be physical states and physical environments, respectively, 
associated to sets of state variables ${\mathcal{X}}_i$, sets of sensors ${\mathcal{S}}_i$, and sets of actuators ${\mathcal{A}}_i$,  for $i \in \{ 1,2\}$.
For ${\mathcal{X}}_1 \cap {\mathcal{X}}_2=\emptyset$, 
${\mathcal{S}}_1 \cap {\mathcal{S}}_2=\emptyset$  and 
${\mathcal{A}}_1 \cap {\mathcal{A}}_2=\emptyset$, we define: 

\begin{itemize}
	\item the \emph{disjoint union} of the physical states $S_1$ and $S_2$,
	written $S_1 \uplus S_2$, to be the physical state 
	$\stateCPS
	{\statefun{}} 
	{\sensorfun{}}
	{\actuatorfun{} }
	$
	such that: 
	${\statefun{} }=  \statefun^1{} \cup \statefun^2{}$, 
	${\sensorfun{} }=\sensorfun ^1{} \cup \sensorfun^2{}  $, and 
	${\actuatorfun{} }=\actuatorfun ^1{} \cup \actuatorfun^2{}  $; 
	\item the \emph{disjoint union} of the physical environments $E_1$ and $E_2$,
	written $E_1 \uplus E_2$,  to be the physical environment 
	$   \envCPS 
	{\evolmap{} } 
	{\measmap{} }   
	{\invariantfun{} }
	{\safefun{}}
	{\uncertaintyfun{}}
	{\errorfun{}}
	$ such that:
	\begin{enumerate}
		\item ${\evolmap{}} = {\evolmap^1{}} \cup {\evolmap^2{}}$
		\item ${\measmap{}}  =  {\measmap^1{}} \cup {\measmap^2{}}$
		\item $S_1 \uplus S_2 \in \invariantfun{}$ iff $S_1 \in \invariantfun^1{}$ and $S_2 \in \invariantfun^2{}$ 
		\item $S_1 \uplus S_2 \in \safefun{}$ iff $S_1 \in \safefun^1{}$ and $S_2 \in \safefun^2{}$ 
		\item ${\uncertaintyfun{}} = {\uncertaintyfun^1{}} \cup {\uncertaintyfun^2{}}$
		\item ${\errorfun{}} = {\errorfun^1{}} \cup {\errorfun^2{}}$. 
	\end{enumerate}
\end{itemize}

\begin{definition}[Physically-disjoint \CPS{s}]
	Let $M_i = \confCPS {E_i;S_i}{P_i}$, for $i \in \{ 1 , 2 \}$. We say that $M_1$ and $M_2$ are \emph{physically-disjoint} if
	$S_1$ and $S_2$ have disjoint sets of state variables, sensors and actuators. In this case,  we write $M_1 \uplus M_2$ to denote the \CPS{}  defined as $\confCPS {(E_1\uplus E_2); (S_1 \uplus S_2)}{(P_1\parallel P_2)}$. 
\end{definition}

A \emph{pure-logical process} is a process that may interfere on communication channels but it never interferes on physical devices as it never accesses sensors and/or actuators. Basically, a pure-logical process is a TPL process~\cite{HR95}. Thus, in a system $M \parallel Q$, where $M$ is an arbitrary \CPS{}, a pure-logical process $Q$ cannot interfere with the physical evolution of $M$. A process $Q$ can, however, definitely interact with $M$ via communication channels, and hence affect its observable behaviour. 
\begin{definition}[Pure-logical processes]
	A process $P$ is called \emph{pure-logical} if it never acts on 
	sensors and/or actuators.
\end{definition}

Now, we can finally state the compositionality of our trace preorder $\sqsubseteq$ (the proof can be found in the appendix).
\begin{theorem}[Compositionality of $\sqsubseteq$]
	\label{thm:congruence-trace}
	Let $M$ and $N$ be two arbitrary \CPS{s} in \cname{}. 
	\begin{enumerate}
		\item
		\label{thm:congruence1-trace}
		$M \sqsubseteq N$ implies $M \uplus O \sqsubseteq N \uplus O$, for any 
		physically-disjoint \CPS{} $O$; 
		\item
		\label{thm:congruence2-trace}
		$M \sqsubseteq N$ implies $M \parallel P \sqsubseteq N \parallel P$, for any 
		pure-logical process $P$; 
		\item 
		\label{thm:congruence3-trace}
		$M \sqsubseteq N$ implies $M \backslash c \; \sqsubseteq \; N \backslash c$, for 
		any channel $c$.
	\end{enumerate} 
\end{theorem}

The reader may wonder whether our trace preorder $\sqsubseteq$ is preserved by more permissive contexts. The answer is no. 
Suppose that in the second item of Theorem~\ref{thm:congruence-trace} we allowed a process $P$ that can also read on sensors. In this case, even if $M \sqsubseteq N$, the parallel process $P$ might read a different value in the two systems  at the very same sensor $s$ (due to the sensor error) and transmit these different values on a free channel, breaking the congruence. Activities on actuators  may also lead to different behaviours of the compound systems:  $M$ and $N$ may have physical components that are not exactly aligned. 
A similar reasoning applies when composing \CPS{s} with non physically-disjoint ones: noise on physical devices may break the compositionality result.  

As we are interested in formalising timing aspects of attacks, such as beginning and duration, we propose a timed variant of $\sqsubseteq$ up to (a possibly infinite) \emph{discrete time interval} $m..n$, with $m \in \mathbb{N}^{+}$ and $n \in \mathbb{N}^{+} \cup \infty $. Intuitively, we write $M \sqsubseteq_{m..n} N$ if the \CPS{} $N$ simulates the execution traces of $M$ in all time slots, except for
those contained in the discrete time interval $m..n$.
\begin{definition}[Trace preorder up to a time interval]
	\label{Time-bounded-trace-equivalence}
	We write $M \sqsubseteq_{m..n} N$, for $m \in \mathbb{N}^+$ and $n \in \mathbb{N}^+ \cup \{\infty\}$, with $m \leq n$, if the following conditions hold: 
	\begin{itemize}
		\item $m$ is the minimum integer for which there is a trace $t$, with $\#\tick(t){=}m{-}1$, s.t.\ $M \trans{t}$ and $N \not\!\Trans{\hat{t}}$;
		\item $n$ is the infimum element of $\mathbb{N}^+ \cup \{ \infty \}$, $n \geq m$, such that whenever $M \trans{t_1}M'$, with $\#\tick(t_1)=n-1$, there is $t_2$, with 
		$\#\tick(t_1)=\#\tick(t_2)$, such that $N \trans{t_2}N'$, for some $N'$, and $M' \sqsubseteq N'$.
	\end{itemize}
\end{definition}

In Definition~\ref{Time-bounded-trace-equivalence}, the first item says that $N$ can simulate the traces of $M$ for at most $m{-}1$ time slots; whereas the second item says two things: (i)~in time interval $m..n$ the simulation does not hold; (ii)~starting from the time slot $n{+}1$ the \CPS{} $N$ can simulate again the traces of $M$. Note that $\mathrm{inf}(\emptyset)=\infty$. Thus, if $M \sqsubseteq_{m..\infty} N$, then $N$ simulates $M$ only in the first $m-1$ time slots. 

\begin{theorem}[Compositionality of $\sqsubseteq_{m..n}$]
	\label{thm:congruence-traceupto}
	Let $M$ and $N$ be two arbitrary \CPS{s} in \cname{}. 
	\begin{enumerate}
		\item
		\label{thm:congruence1-traceupto}
		$M \sqsubseteq_{m..n} N$ implies that for any physically-disjoint \CPS{} there are $m',n' \in \mathbb{N}^+ \cup \infty  $, with $ m'..n' \subseteq  m..n $ such that  $M \uplus O \sqsubseteq_{m'..{n'}} N \uplus O$; 
		\item
		\label{thm:congruence2-traceupto}
		$M \sqsubseteq_{m..n} N$ implies that for any pure-logical process $P$  there are $m',n' \in \mathbb{N}^+ \cup \infty  $, with $ m'..n' \subseteq  m..n $ such that 
		$M \parallel P \sqsubseteq_{m'..{n'}} N \parallel P$; 
		\item 
		\label{thm:congruence3-traceupto}
		$M \sqsubseteq_{m..n} N$ implies that for 
		any channel $c$ there are $m',n' \in \mathbb{N}^+ \cup \infty  $, with $ m'..n' \subseteq  m..n $ such that $M \backslash c \; \sqsubseteq_{ {m'}..{n'} } \; N \backslash c$.  
	\end{enumerate} 
\end{theorem}
The proof can be found in the appendix. 

\section{A Running Example}
\label{sec:running_example}

In this section, we introduce a running example to illustrate how we can precisely represent \CPS{s} and a variety of 
different physics-based attacks. In practice, we formalise a relatively simple \CPS{} $\mathit{Sys}$ in which the temperature of an \emph{engine} is maintained within a specific range by means of a cooling system. 
We wish to remark here that while we have kept the example simple, it is actually far from trivial and designed to describe a wide number of attacks. 
The main structure of the \CPS{} $\mathit{Sys}$ is shown in Figure~\ref{f:Sys-structure}.

\begin{figure}[t]
	\centering
	\begin{minipage}[c]{0.4\textwidth}
		\begin{tikzpicture}[scale=0.4, every node/.style={scale=0.5}]
		\node[punkt] (engine) {\huge Engine};
		
		\node[punkt, below=0.7cm of engine]
		(controller) {\huge Ctrl};
		\node[punkt, below=0.7cm of controller]
		(ids) {\huge IDS};
		
		\node[punkt, right=of engine] (sensor) {\huge Sensor};
		\node[punkt, left=of engine] (actuator) {\huge Actuator}
		edge[pil, right=45] (engine);
		
		\path[->] (engine) edge[line] (sensor);
		\path [line] (sensor) |- node[left, pos=0.3] {{\huge${s_t}$}} (controller.east);	\path [line] (sensor) |- (ids.east);
		\path [line] (actuator) -- (engine);
		\path [line] (controller.west) -| node[above, pos=0.3,] {\LARGE ${cool}$}  (actuator);
		\path (ids) edge[pil,<->] node[left, pos=0.46] {\LARGE ${sync}$}  (controller); 
		\end{tikzpicture}	
	\end{minipage}
	\caption{The main structure of the \CPS{} $\mathit{Sys}$}
	\label{f:Sys-structure}
	
\end{figure}

\subsection{The \CPS{} $\mathit{Sys}$}

The physical state $\stateSys$ of the engine  is characterised by:\
\begin{inparaenum}[(i)]
	\item a state variable $\mathit{temp}$ containing the current temperature of the engine, and an integer state variable $\mathit{stress}$ keeping track of the level of stress of the mechanical parts of the engine due to high temperatures (exceeding $9.9$ degrees); this integer 
	variable ranges from $0$, meaning no stress, to $5$, for high stress; 
	\item a sensor $s_{\mathrm{t}}$ (such as a thermometer or a thermocouple) measuring the temperature of the engine,
	\item an actuator $\mathit{cool}$ to turn on/off the cooling system. 
\end{inparaenum}

The physical environment of the engine, $\env$, is constituted by:\ 
\begin{inparaenum}[(i)]
	\item a simple evolution law $\evolmap$ that increases (respectively, decreases) the value of $\mathit{temp}$ by one degree per time unit, when the cooling system is inactive (respectively, active), up to the uncertainty of the system; the variable $\mathit{stress}$ is increased each time the current temperature is above $9.9$ degrees, and dropped to $0$ otherwise;
	\item a measurement map $\measmap{}$ returning the value detected by the sensor $s_t$, up to the error associated to the sensor;
	\item an invariant set saying that the system gets faulty when the temperature of the engine gets out of the range $[0, 50]$, 
	\item a safety set to express that the system moves to an unsafe state when the level of stress reaches the threshold $5$, 
	\item an uncertainty function in which each state variable may evolve with an uncertainty  $\delta=0.4$ degrees, 
	\item a sensor-error function saying that the sensor $s_{\mathrm{t}}$ has an accuracy $\epsilon = 0.1$ degrees. 
\end{inparaenum}

Formally, $\stateSys = \langle {\statefun{}}, {\sensorfun{}}, {{\actuatorfun{}}} \rangle $ where:
\begin{itemize}
	\item $\statefun{} \in \mathbb{R} ^{\{\mathit{temp},\mathit{stress}\} }$ and 
	$\statefun{}(\mathit{temp})=0$ and $\statefun{}(\mathit{stress}) = 0$;
	\item $\sensorfun{} \in \mathbb{R} ^{\{s_{\mathrm t}\} }$ and 
	$\sensorfun{}(s_{\mathrm t})=0$;  
	\item $\actuatorfun{} \in \mathbb{R} ^{\{\mathit{cool}\} }$ and
	$\actuatorfun{}(\mathit{cool})=\off$; for the sake of simplicity, we can
	assume $\actuatorfun{}$ to be a mapping $\{ \mathit{cool} \} \rightarrow
	\{ \on , \off\}$ such that $\actuatorfun{}(\mathit{cool})= \off$ if
	$\actuatorfun{}(\mathit{cool}) \geq 0$, and $\actuatorfun{}(\mathit{cool})= \on$ if
	$\actuatorfun{}(\mathit{cool}) < 0$;  
\end{itemize}
and 
$\env = \envCPS 
{\evolmap{}}
{\measmap{}} 
{\invariantfun{}}
{\safefun{}} 
{\uncertaintyfun{}}  
{\errorfun{}}   
$ with:
\begin{itemize}
	
	\item $\evolmap{}( \statefun^i{}, \actuatorfun^i{}, \uncertaintyfun{}) $
	is the set of functions $\xi \in \mathbb{R} ^{\{\mathit{temp},\mathit{stress}\} }$
	such that:
	\begin{itemize}
		\item
		$\xi(\mathit{temp}) = \statefun^i{}(\mathit{temp}) + 
		\mathit{\mathit{heat}}(\actuatorfun^i{},\allowbreak \mathit{cool}) + \gamma $, 
		with $ \gamma \in [- \delta, + \delta] $ and $\mathit{heat}(\actuatorfun^i{},\mathit{cool})=-1$ if
		$\actuatorfun^i{}(\mathit{cool}) = \on$ (active cooling), and
		$\mathit{heat}(\actuatorfun^i{},\mathit{cool})=+1$ if
		$\actuatorfun^i{}(\mathit{cool}) = \off$ (inactive cooling);
		
		\item $\xi(\mathit{stress}) = \min (5 \, , \,  \statefun^i{}(\mathit{stress}){+}1)$ if   $\statefun^i{}(\mathit{temp})>9.9$;  $\xi(\mathit{stress}) = 0$, otherwise;
	\end{itemize}

	\item $\measmap{}(\statefun^i{}, \errorfun{})  = \big \{ \xi :
	\xi(s_{\mathrm{t}}) \in [ \statefun^i{}(\mathit{temp})-\epsilon \, , \, \allowbreak \statefun^i{}(\mathit{temp}) + \epsilon ]  \big \}$;
	
	\item $\invariantfun{}=\{ \statefun^i{} \,: \,  0 \leq \statefun^i{}(\mathit{temp})\leq 50\}$; 
	
	\item $\safefun{}=\{\statefun^i{} \,: \, \statefun^i{}(\mathit{stress}) < 5\}$ (we recall that the stress threshold is $5$);
	
	\item $\uncertaintyfun{} \in \mathbb{R}^{\{\mathit{temp},\mathit{stress}\} }$, $\uncertaintyfun{}(\mathit{temp}) =0.4 = \delta$ and $\uncertaintyfun{}(\mathit{stress})= 0$;
	
	\item $\errorfun{} \in \mathbb{R}^{\{s_{\mathrm{t}} \}}$ and
	$\errorfun(s_{\mathrm{t}})= 0.1=\epsilon$. 
\end{itemize}

For the cyber component of the \CPS{} $\mathit{Sys}$,  
we define two parallel processes: $\mathit{Ctrl}$ and $\mathit{IDS}$. 
The former models the \emph{controller} activity, consisting in reading the temperature sensor and in governing the cooling system via its actuator, whereas the latter models a simple \emph{intrusion detection system} that attempts to detect and signal \emph{anomalies} in the behaviour of the system~\cite{ACM-survey2018}. Intuitively, $\mathit{Ctrl}$ senses the temperature of the engine at each time slot. When the 
\emph{sensed temperature} is above $10$ degrees, the controller activates the coolant. The cooling activity is maintained for $5$ consecutive time units. After that time, the controller synchronises with the $\mathit{IDS}$ component via a private channel $\mathit{sync}$, and then waits for \emph{instructions}, via a channel $\mathit{ins}$. The $\mathit{IDS}$ component checks whether the 
\emph{sensed temperature} is still 
above $10$. If this is the case, it sends an \emph{alarm} of ``high 
temperature'', via a specific channel, and then tells $\mathit{Ctrl}$ to keep cooling for $5$ more time units; otherwise, if the temperature is not above $10$, the $\mathit{IDS}$ component requires $\mathit{Ctrl}$ to stop the cooling activity.

\begin{displaymath}
{\normalsize  
	\begin{array}{rcl}
	\mathit{Ctrl} &  =  & \rsens x {s_{\operatorname{t}}} . \ifelse {x>10}
	{ \mathit{Cooling} } { \tick.\mathit{Ctrl} } \\[1pt]
	\mathit{Cooling}  & =  &   \wact{\on}{\emph{cool}}.\tick^5 . \mathit{Check}
	\\[1pt]
	\mathit{Check} & = & 
	\OUTCCS{\mathit{sync}}. 
	\LIN{\mathit{ins}}{y}.\mathsf{if} \, 
	(y=\mathsf{keep\_cooling}) \, \{ \tick^5.\mathit{Check} \} \:
	\mathsf{else} \;
	\{  \wact{\off}{\mathit{cool}}.\tick .\mathit{Ctrl} \}\\[1pt]
	\mathit{IDS} & = &  \LINCCS {\mathit{sync}}  .  \rsens x {s_{\operatorname{t}}} .  \mathsf{if} 
	\, (x>10)  \,  \{ \OUT{\mathit{alarm}}{\mathsf{high\_temp}}. 
	\OUT{\mathit{ins}}{\mathsf{keep\_cooling}}. \tick.\mathit{IDS} \} \\
	&& \Q \mathsf{else} \; \{ \OUT{\mathit{ins}}{\mathsf{stop}}.\tick. \mathit{IDS} \}
	\, . \end{array}
}
\end{displaymath}
Thus, the whole \CPS{} is defined as: 
\begin{figure}[t]
	
\end{figure}

\begin{displaymath}
\mathit{Sys} \; = \; \confCPS {\env;\stateSys} {(\mathit{Ctrl} \parallel \mathit{IDS})
	{\setminus}\{ \mathit{sync}, \mathit{ins}\}} \
\end{displaymath}%

For the sake of simplicity, our $\mathit{IDS}$ component is quite basic: for instance, it does not check whether the temperature is too low. However, it is straightforward to replace it with a more sophisticated one, containing more informative tests on sensor values and/or on actuators commands.

\begin{figure}[t]
	\centering
	\includegraphics[width=7.25cm,keepaspectratio=true,angle=0]{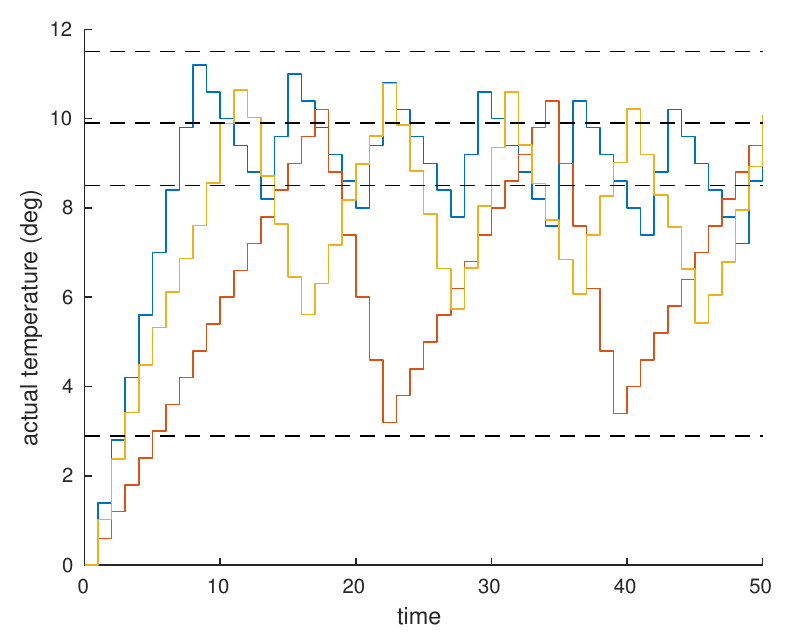} 
	\caption{Three  possible evolutions of the \CPS{} $\mathit{Sys}$}
	\label{f:HS-traj}
\end{figure}

Figure~\ref{f:HS-traj} shows three possible evolutions in time of
the state variable $\mathit{temp}$ of $\mathit{Sys}$:
\begin{inparaenum}[(i)] 
	\item  the first one (in red), in which the temperature of the engine always grows of $1-\delta = 0.6$ degrees per time unit, when the cooling is off, and always decrease of $1+\delta=1.4$ degrees per time unit, when the cooling is on; 
	\item the second one (in blue), in which the temperature always grows of $1+\delta=1.4$ degrees per time unit, when the cooling is off, and  always decreases of $1-\delta=0.6$ degrees per time unit, when the cooling is on; 
	\item and a third one (in yellow), in which, depending on whether the cooling is off or on, at each time step the temperature grows or decreases of an arbitrary offset lying in the interval $[1-\delta , 1+\delta]$. 
\end{inparaenum}

Our operational semantics allows us to formally prove a number of properties of our running example. For instance, 
Proposition~\ref{prop:sys} says that the $\mathit{Sys}$ is 
sound and it never fires the $\mathit{alarm}$. 
\begin{proposition} 
	\label{prop:sys}
	\label{prop:sys1}
	\label{prop:sys2}
	If $\mathit{Sys} \trans{t}$ for some trace $t =\alpha_1 \ldots \alpha_n$, then $\alpha_i \in \{ \tau , \tick \}$, for any $i \in \{1, \ldots, n\}$. 
\end{proposition}

Actually, we can be quite precise on the temperature reached by $\mathit{Sys}$ before and after the cooling: in each of the $5$ rounds of cooling, the temperature will drop of a value lying in the real interval $[1{-}\delta, 1 {+} \delta]$, where $\delta$ is the uncertainty. 

\begin{proposition}
	\label{prop:X}
	For any execution trace of $\mathit{Sys}$, we have:
	\begin{itemize}[noitemsep]
		\item when $\mathit{Sys}$ \emph{turns on} the cooling, the value of
		the state variable $\mathit{temp}$ ranges over $(9.9  ,   11.5]$; 
		\item when $\mathit{Sys}$ \emph{turns off} the cooling, the value of
		the variable $\mathit{temp}$ ranges over $(2.9, 8.5]$. 
	\end{itemize}
\end{proposition}

The proofs of the Propositions~\ref{prop:sys} and \ref{prop:X} can 
be found in the appendix.  
In the following section, we will verify the safety properties stated in these two propositions relying on the statistical model checker \Uppaal{} SMC~\cite{David:2015:UST:2802769.2802840}.

\subsection{A formalisation of $\mathit{Sys}$ in \Uppaal{} SMC}

\begin{figure}[!t]
	\begin{displaymath}
	\begin{array}{c@{\hspace*{.4cm}}l}
	\includegraphics[scale=0.4]{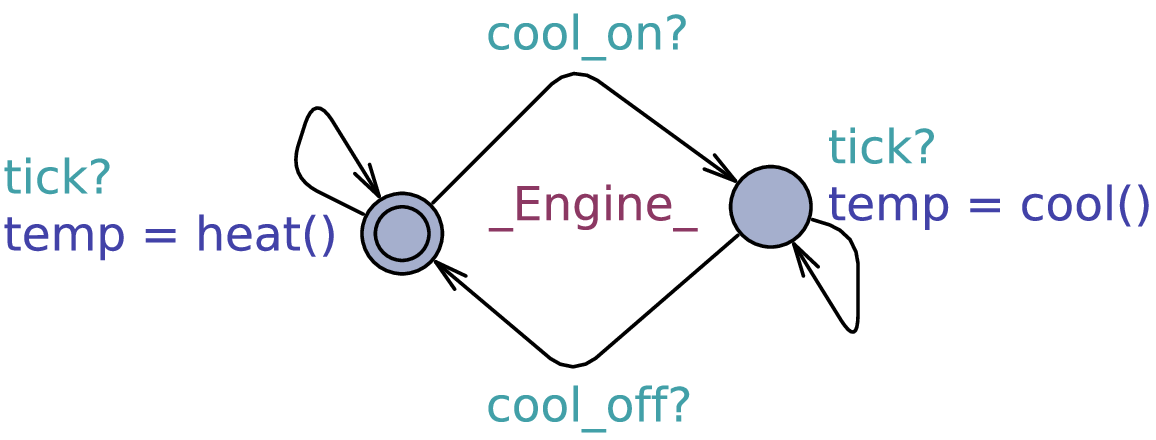}  &
	\includegraphics[scale=0.4]{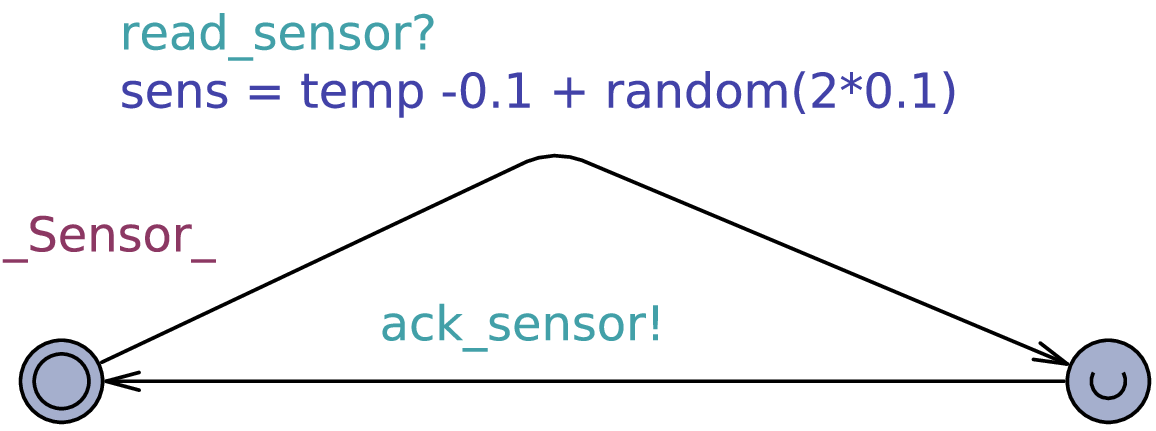}\\[5pt]
	\includegraphics[scale=0.4]{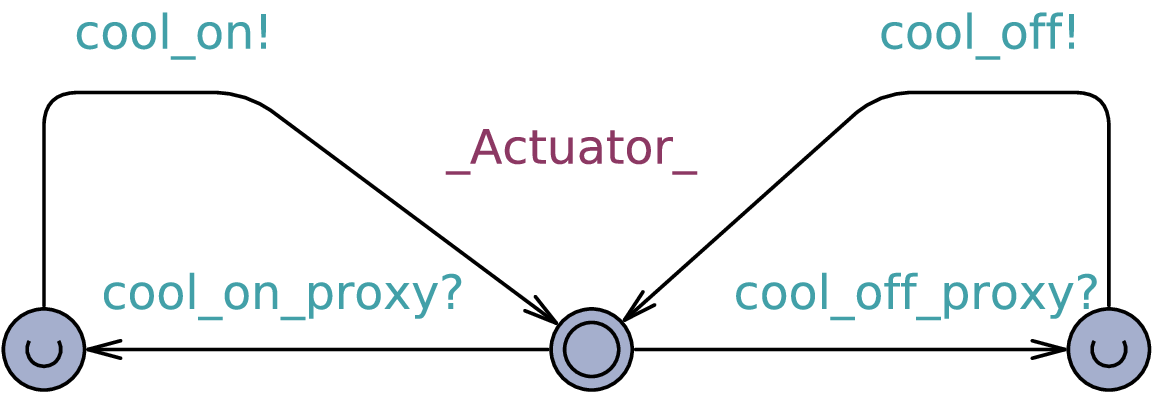}
	&
	\includegraphics[scale=0.4]{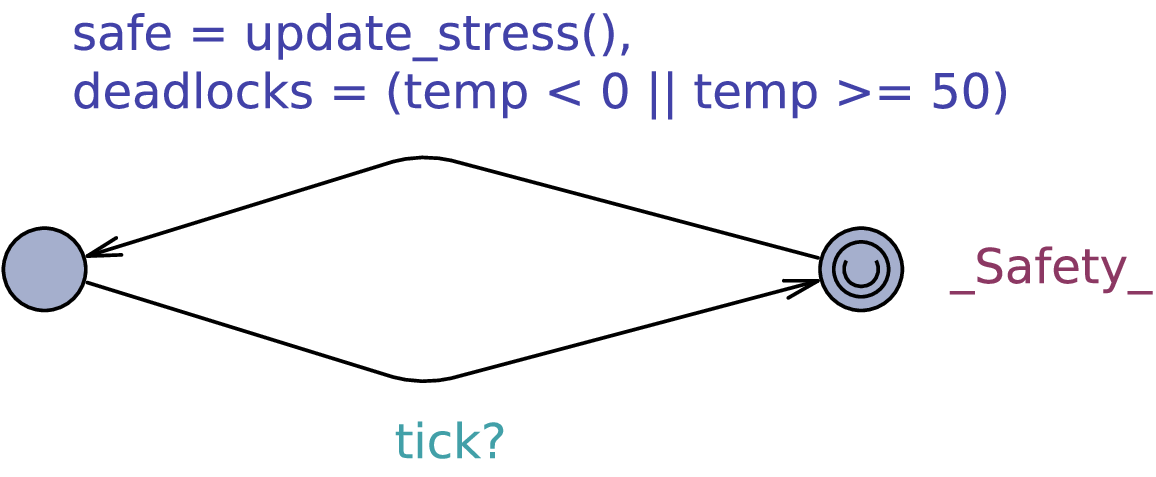} 
	\end{array}
	\end{displaymath}
	\caption{\Uppaal{} SMC model for the physical component of $\mathit{Sys}$}
	\label{fig:physical}	
	\label{fig:physicalBIS}
\end{figure}

In this section, we formalise our running example in \Uppaal{} SMC~\cite{David:2015:UST:2802769.2802840}, the \emph{statistical} extension of the \Uppaal{} model checker~\cite{DBLP:conf/sfm/BehrmannDL04} supporting  the analysis of systems expressed as composition of \emph{timed and/or probabilistic automata}. 
In \Uppaal{} SMC, the user must specify two main statistical parameters $\alpha$ and $\epsilon$,  ranging in the interval $[0, 1]$, and representing the probability of \emph{false negatives} and probabilistic \emph{uncertainty\/}, respectively. Thus, given a CTL property of the system under investigation,  the tool returns a probability 
estimate for that property, lying in a  confidence interval [$p -\epsilon$, 
$p + \epsilon$],  for some probability $p \in [0,1]$, with an accuracy $1-\alpha$.  The number of necessary runs to ensure the required accuracy is then computed by the tool relying on the Chernoff-Hoeffding theory~\cite{Chernoff1952}. 

\subsubsection{Model}
The \Uppaal{} SMC model of our use case $\mathit{Sys}$ is given by three main components represented in terms of \emph{parallel timed automata}: the \emph{physical component}, the \emph{network}, and the 
\emph{logical component}. 

The physical component, whose model is shown in Figure~\ref{fig:physical}, consists of four automata: 
\begin{inparaenum}[(i)] 
	\item the \emph{\_Engine\_} automaton that governs the evolution of the variable \emph{temp} by means of the \emph{heat} and \emph{cool} functions;
	\item the \emph{\_Sensor\_} automaton that updates the global variable \emph{sens} at each measurement request;  
	\item the \emph{\_Actuator\_} automaton that activates/deactivates the cooling system;
	\item the \emph{\_Safety\_} automaton that handles the integer variable \emph{stress}, via the \emph{update{\_}stress} function, and the Boolean variables \emph{safe} and \emph{deadlocks}, associated to the safety set $\safefun{}$ and the invariant set $\invariantfun{}$ of $\mathit{Sys}$, respectively.\footnote{In Section~\ref{sec:related}, we explain why we need to implement an automaton to check  for safety conditions rather than  verifying a safety property.}
\end{inparaenum}
We also have a small automaton to model a discrete notion of time (via a synchronisation channel $\textsf{tick}$) as the evolution of state variables is represented via difference \nolinebreak equations. 

The \emph{network}, whose model is given in Figure~\ref{fig:network}, consists of \emph{two proxies}: a proxy to relay actuator commands between the actuator device and the controller, a second proxy to relay measurement requests between the sensor device and the logical components (controller and IDS). 

The \emph{logical component}, whose model is given in Figure~\ref{fig:logical},
consists of two automata: \emph{\_Ctrl\_} and \emph{\_IDS\_} to model the controller and the Intrusion Detection System, respectively; both of them synchronise with their associated proxy copying a fresh value of \emph{sens} into their local variables (\emph{sens\_ctrl} and \emph{sens\_ids}, respectively). 
Under proper conditions, the $\mathit{\_IDS\_}$ automaton fires alarms by setting a Boolean variable \emph{alarm}.

\begin{figure}[t]
	\centering
	\includegraphics[scale=0.35]{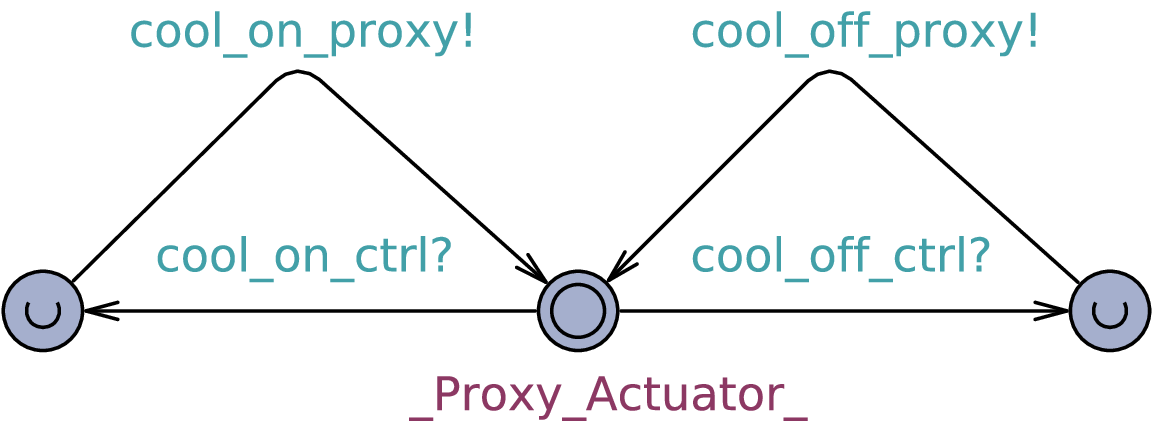}	\Q
	\includegraphics[scale=0.35]{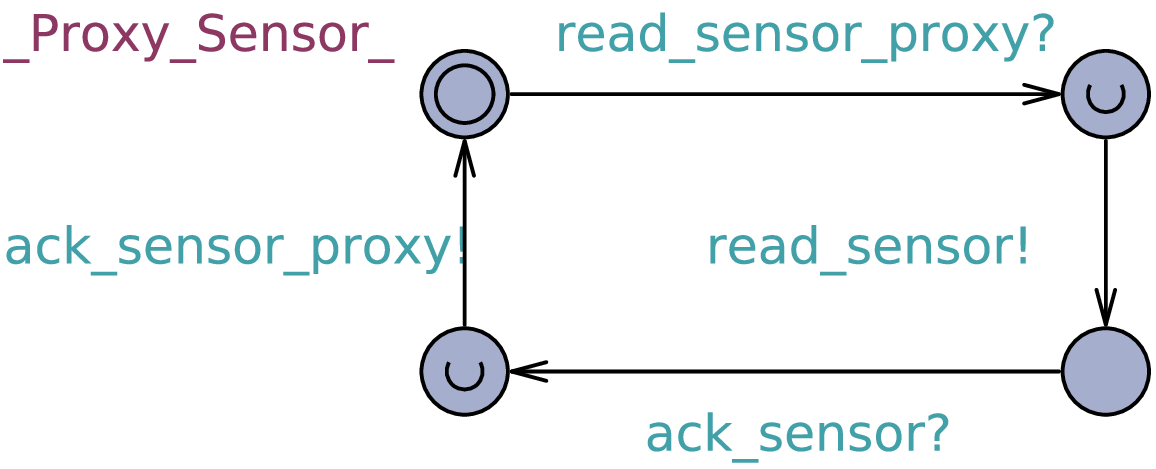} 
	\caption{\Uppaal{} SMC model for the network component of $\mathit{Sys}$}
	\label{fig:network}
\end{figure}

\begin{figure}[t]
	\centering
	\includegraphics[scale=0.4]{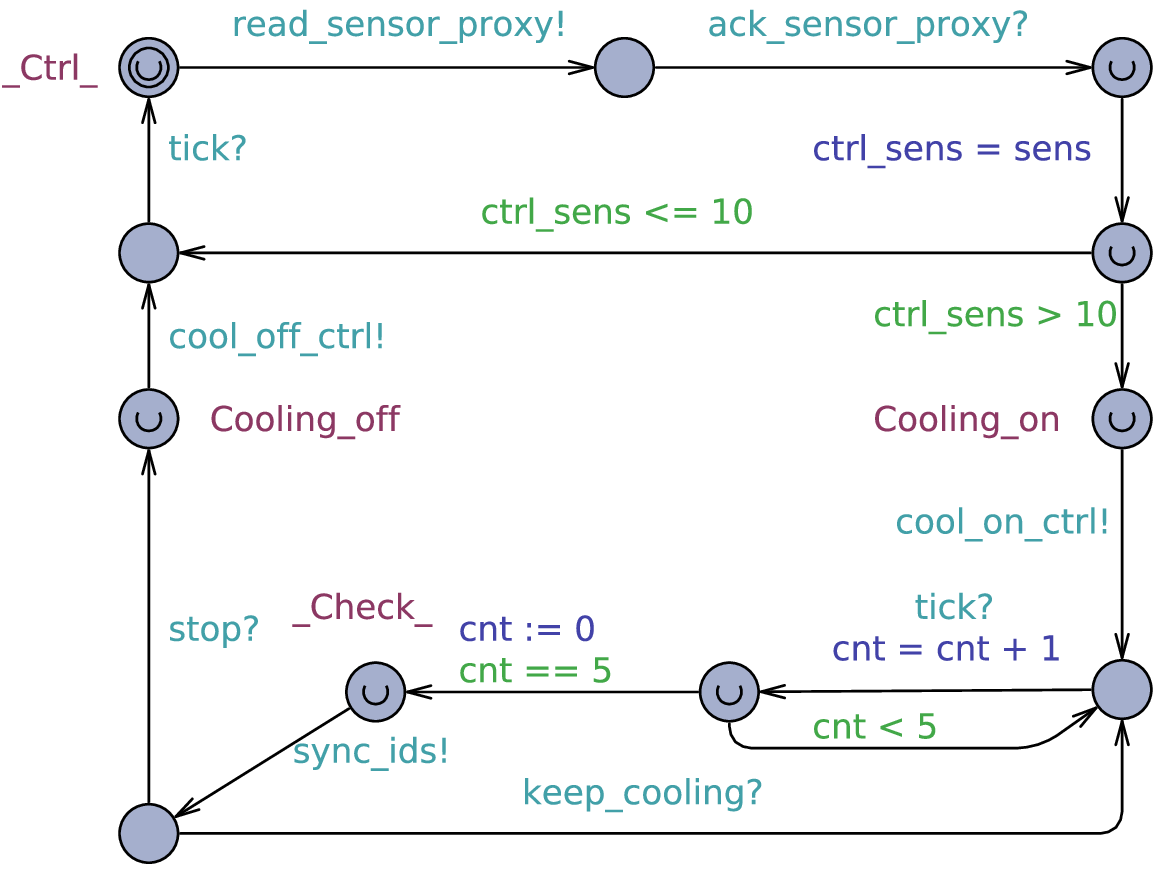}\Q\q
	\includegraphics[scale=0.4]{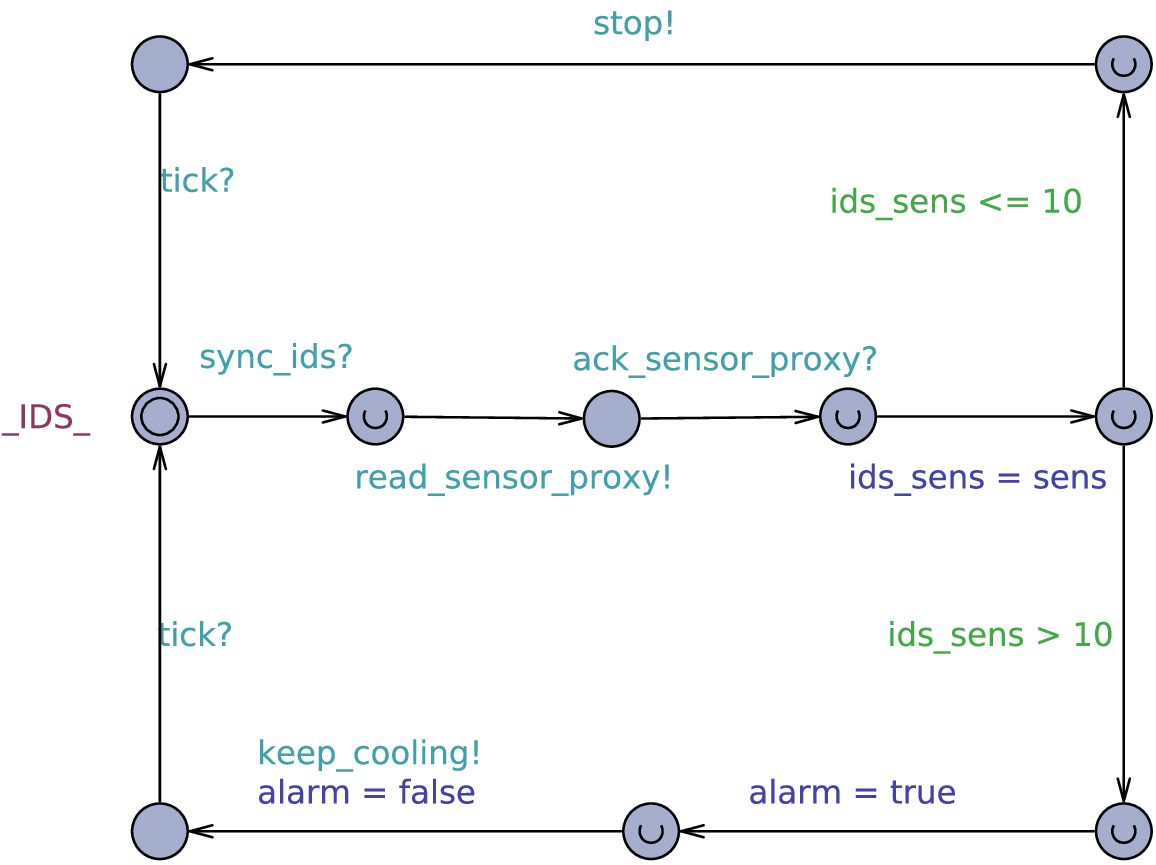} 
	\caption{\Uppaal{} SMC model for the logical component of $\mathit{Sys}$}
	\label{fig:logical}
\end{figure}

\subsubsection{Verification}
We conduct our safety verification using a notebook with the following set-up: 
\begin{inparaenum}[(i)]
	\item 2.8 GHz Intel i7 7700 HQ, with 16 GB memory, and Linux Ubuntu 16.04 operating system;  
	\item \Uppaal{} SMC model-checker 64-bit, version 4.1.19. 
\end{inparaenum} 
The statistical parameters of false negatives ($\alpha$) and probabilistic uncertainty ($\epsilon$) are both set to 0.01, leading to a confidence level of 99\%. As a consequence, having fixed these parameters, for each of our experiments, \Uppaal{} SMC run a number of runs that may vary from a few hundreds to 26492  (\emph{cf.} Chernoff-Hoeffding bounds).

We basically use \Uppaal{} SMC to verify properties expressed in terms of \emph{time bounded CTL} formulae of the form $\square_{[t_1,t_2]} e_{prop}$ and $\lozenge_{[0,t_2]} e_{prop}$\footnote{The $0$ in the left-hand side of the time interval is imposed by the syntax of \Uppaal{} SMC.}, where $t_1$ and $t_2$ are \emph{time instants} according to the discrete representation of time in \Uppaal{} SMC. In practice, we use formulae of the form $\square_{[t_1,t_2]} e_{prop}$ to compute the probability that a property $e_{prop}$\footnote{$e_\mathit{prop}$ is a side-effect free expression over variables (e.g.,  clock variables, location names and primitive variables)~\cite{DBLP:conf/sfm/BehrmannDL04}.} holds in all \emph{time slots} of the time interval $t_1..{t_2}$, whereas  with formulae of the form $\lozenge_{[0,t_2]} e_{prop}$ we calculate the probability that a property $e_{prop}$ holds in a least one \emph{time slot} of the time interval $0..{t_2}$.

Thus, instead of proving Proposition~\ref{prop:sys}, we verify, with a  99\% accuracy, that in all possible executions that are at most $1000$ time slots long, the system $\mathit{Sys}$ results to be sound and alarm free, with probability $0.99$. Formally, we verify the following three properties: 
\begin{itemize}
	\item $\square_{[1,1000]}(\neg \mathit{deadlocks})$, expressing that the system does not deadlock; 
	\item $\square_{[1,1000]}(\mathit{safe})$, expressing that the system does not violate the safety conditions; 
	\item $\square_{[1,1000]} (\neg \mathit{alarm})$, expressing that the IDS does not fire any alarm. 
\end{itemize} 

Furthermore, instead of Proposition~\ref{prop:X}, we verify, with the same accuracy and for runs of the same length (up to a short initial transitory phase  
lasting $5$ time instants) that if the cooling system is off, then the temperature of the engine lies in the real interval $(2.9, 8.5]$, otherwise it ranges over the interval $(9.9, 11.5]$. Formally, we verify the following two properties: 
\begin{itemize}
	\item $\square_{[5,1000]}(\mathit{Cooling\_off} \implies (\mathit{temp} > 2.9 \wedge \mathit{temp} \leq 8.5))$ 
	\item $\square_{[5,1000]}(\mathit{Cooling\_on} \implies (\mathit{temp} > 9.9 \wedge \mathit{temp} \leq 11.5))$. 
\end{itemize}
The verification of each of the five properties above requires around $15$ minutes. The \Uppaal{} SMC models of our system and the attacks discussed in the next section are available at the repository 
\begin{center}
{\small \texttt{https://bitbucket.org/AndreiMunteanu/cps\_smc/src/}}. 
\end{center}

\begin{remark}
	In our \Uppaal{} SMC model we decided to represent both  uncertainty of physical evolution (in the functions \emph{heat} and \emph{cool} of \emph{\_Engine\_}) and measurement noise (in \emph{\_Sensor\_}) in a probabilistic manner via random extractions.
	Here, the reader may wonder whether it would have been enough to restrict our SMC analysis by considering only upper and lower bounds on these two quantities. Actually, this is not the case because 
	such a restricted analysis might miss admissible execution traces. 
	To see this, suppose to work with 
	a physical uncertainty that is always either $0.4$ or $-0.4$. Then, the temperature reached by the system would always be of the form $n.k$, for $n,k \in \mathbb{N}$ and $k$ even. As a consequence, our analysis would miss all  execution  traces in which the system reaches the maximum admissible temperature of $11.5$ degrees. 
\end{remark}


\section{Physics-based Attacks}
\label{sec:cyber-physical-attackers}

In this section, we use \cname{} to formalise a \emph{threat model} of 
physics-based attacks, i.e.,  attacks that can manipulate sensor and/or actuator signals in order to drive a
\emph{sound} \CPS{} into an undesired state~\cite{TeShSaJo2015}.  
An attack may have different levels of access to physical devices;  
for example, it might be able to get read access to the sensors
but not write access; or it might get write-only access to the actuators but not read-access. 
This level of granularity is very important to model precisely how physics-based attacks can affect a
CPS~\cite{Cardenas2015}.  
In \cname{}, we have a syntactic way to distinguish malicious processes from honest ones. 
\begin{definition}[Honest system]
	A \CPS{} $\confCPS {E;S}  P$ is \emph{honest} if $P$ is honest, where 
	$P$ is honest if it does not contain constructs of the form 
	$\timeout{\mu.{P_1}}{P_2}$.
\end{definition}

We group physics-based attacks in classes that describe both the malicious activities and the timing aspects of the attack. Intuitively, a class of attacks provides information about which physical devices are accessed by the attacks of that class, how they are accessed (read and/or write), when the attack begins and when the attack ends. Thus, let $\I$ be the set of all possible malicious activities on the physical devices of a system, $m \in \mathbb{N}^+$ be the time slot when an attack starts, and $n \in \mathbb{N}^+ \cup \{\infty\}$ be the time slot when the attack ends. We then say that an \emph{attack $A$ is of class $C \in [ \I \rightarrow {\cal P}(m..n) ]$} \nolinebreak if: 
\begin{enumerate}
	\item all possible malicious activities of $A$ coincide with
	those contained in $\I$; 
	\item the first of those activities may occur in the
	$m^\mathrm{th}$ time slot but not  before; 
	\item the last of those activities may occur in the $n^\mathrm{th}$ time slot but not after; 
	\item for $\iota \in \I$, $C(\iota)$ returns a (possibly
	empty) set of time slots when $A$ may read/tamper with the device $\iota$ (this set is contained in $m..n$); 
	\item $C$ is a total function, i.e., 
	if no attacks of class $C$ can achieve the malicious activity $\iota \in \I$, then $C(\iota) = \emptyset$. 
\end{enumerate} 

\begin{definition}[Class of attacks]
	\label{def:attacker-class}	
	Let ${\cal I} = \{ \mbox{\Lightning}p ? : p \in {\cal S} \cup
	{\cal A} \} \cup \{ \mbox{\Lightning}p ! \, : p \in {\cal S} \cup {\cal A}
	\}$ be the set of all possible \emph{malicious activities} on 
	physical devices. Let 
	$m \in \mathbb{N}^{+}$,    $n \in \mathbb{N}^+ \cup
	\{\infty\}$, with $m \leq n$. 
	A \emph{class of attacks}
	$C \in [ \I \rightarrow {\cal P}(m..n) ]$ is a total function such that
	for any attack $A$ of class $C$ we have:
	\begin{enumerate}
		\item[(i)]
		\begin{math}
		C(\iota)=
		\{  k  :  A \trans{t}\trans{\iota v} A' \,   \wedge \, 
		k = \#\tick(t)+1  
		\}
		\end{math},
		for  $\iota \in \I$,
		\item[(ii)] 
		\begin{math}
		m = \inf \{ \, k \, : \, k \in C(\iota) 
		\, \wedge \, \iota \in \I \,  \}
		\end{math},
		\item[(iii)]
		\begin{math}
		n = \sup \{ \, k \, : \, k \in C(\iota) 
		\, \wedge \, \iota \in \I \,  \}. 
		\end{math} 
	\end{enumerate}
\end{definition}

Along the lines of~\cite{FM99}, 
we can say that an attack $A$ affects a \emph{sound} \CPS{}  $M$ if the execution of the compound system $M \parallel A$ differs from that of the original system $M$, in an observable manner. Basically, a physics-based attack can influence the system under attack in at least two different \nolinebreak ways:
\begin{itemize}
	\item The system $M \parallel  A$ might deadlock when $M$ may not; this means that the attack $A$ affects the \emph{availability} of the system. We recall that in the context of \CPS{s}, deadlock is a particular severe physical event. 
	\item The system $M \parallel A$ might have non-genuine execution traces containing observables (violations of safety conditions or 
	communications on channels) that can't be reproduced by $M$; here the attack affects the \emph{integrity} of the system behaviour. 
\end{itemize}

\begin{definition}[Attack tolerance/vulnerability]
	\label{def:attack-tolerance}
	Let $M$ be an honest and sound  \CPS{}. We say that $M$ is \emph{tolerant to an attack $A$} if 
	$M \parallel A \, \sqsubseteq \,  M$. 
	We say that $M$ is \emph{vulnerable to an attack $A$} if there is a time interval $m..n$, with $m\in \mathbb{N}^+$ and $n \in \mathbb{N}^+ \cup \{\infty\}$, such that $M \parallel A \: \sqsubseteq_{m..n} \,  M$. 
\end{definition}

Thus, if a system $M$ is vulnerable to an attack $A$ of class $C \in [\I
\rightarrow {\cal P}(m..n)]$, during the time interval $m'..n'$, then the attack operates during the interval $m..n$ but it influences the system under attack in the time interval $m'..n'$ (obviously, $m' \geq m$). If $n'$ is finite, then we have a \emph{temporary attack}, otherwise we have a \emph{permanent attack}. Furthermore, if $m'-n$ is big enough and $n-m$ is small, then we have a quick nasty attack that affects the system late enough to allow \emph{attack camouflages}~\cite{GGIKLW2015}. On the other hand, if $m'$ is significantly smaller than $n$, then the attack affects the observable behaviour of the system well before its termination and the \CPS{} has good chances of undertaking countermeasures to stop the attack. Finally, if $M \parallel A \trans{t}\trans{\dead}$, for some trace $t$, then we say that the attack $A$ is \emph{lethal}, as it is capable to halt (deadlock) the \CPS{} $M$. This is obviously a permanent attack.

Note that, according to Definition~\ref{def:attack-tolerance}, the tolerance (or vulnerability) of a \CPS{} also depends on the capability of the $\mathit{IDS}$ component to detect and signal undesired physical behaviours. In fact, the $\mathit{IDS}$ component might be designed to detect abnormal physical behaviours going well further than deadlocks and violations of safety conditions.

According to the literature, we say that an attack is \emph{stealthy} if it is able to drive the \CPS{} under attack into an incorrect physical state (either deadlock or violation of the safety conditions) without being noticed by the $\mathit{IDS}$ component.

\subsection{Three different attacks on the physical devices of the \CPS{} $\mathit{Sys}$}
In this subsection, we present three different attacks to the \CPS{} $\mathit{Sys}$ described in Section~\ref{sec:running_example}.

Here, we use \Uppaal{} SMC to verify the models associated to the system under attack in order to detect deadlocks, violations of safety conditions, and IDS failures.

\begin{example}
	\label{exa:att:DoS}
	Consider the following \emph{DoS/Integrity attack} on the
	the actuator $\mathit{cool}$, of class $C \in [\I
	\rightarrow {\cal P}(m..m)]$ with $C(\mbox{\Lightning}cool?)=C(\mbox{\Lightning}cool!)=\{ m \} $ and $C(\iota) = \emptyset$, for $\iota \not \in \{
	\mbox{\Lightning}cool? , \mbox{\Lightning}cool! \}$: 
	\begin{displaymath}
	A_m  \: = \:  
	\tick^{m{-}1}.  \lfloor {\drop x {cool}}.\mathsf{if}  \:  (x{=}{\off}) \:   \{ {\forge {\off}{cool}} \} \: \mathsf{else} \: \{ \nil \} \rfloor \,  . 
	\end{displaymath} 
	Here, the attack $A_m$ \emph{operates exclusively in the $m^\mathrm{th}$ time slot}, when it tries to drop an eventual cooling command (on or off) coming from the controller, and fabricates a fake command to turn off the cooling system. Thus, if the controller sends in the $m^\mathrm{th}$ time slot a command to turn off the coolant, then nothing bad happens as the attack will put the same message back. On the hand, if the controller sends a command to turn the cooling on, then the attack will drop the command. We recall that the controller will turn on the cooling only if the sensed temperature is greater than $10$ (and hence $\mathit{temp}> 9.9$); this may happen only if $m > 8$. 
	Since the command to turn the cooling on is never re-sent by $\mathit{Ctrl}$, the temperature will continue to rise, and after only $4$ time units the system may violate the safety conditions emitting an action $\unsafe$, while the $\mathit{IDS}$ component will start sending alarms every $5$ time units, until the whole system deadlocks because the temperature reaches the threshold of $50$ degrees. Here, the $\mathit{IDS}$ component of $\mathit{Sys}$ is able to detect the attack with only one time unit delay. 
	
	\begin{figure}[t]
		\centering	
		\includegraphics[scale=0.45]{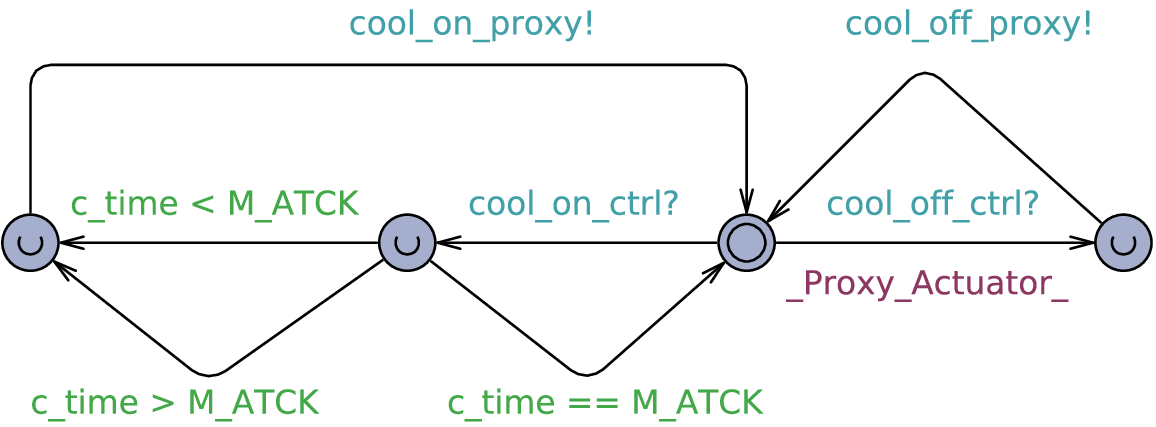}  
		\caption{\Uppaal{} SMC model for the attacker $A_m$ of Example~\ref{exa:att:DoS}}
		\label{fig:attack1_auto}
	\end{figure}
\end{example}

\begin{proposition}
	\label{prop:att:DoS}
	Let $\mathit{Sys}$ be our running example and $A_m$ be the attack defined in Example~\ref{exa:att:DoS}. Then, 
	\begin{itemize}
		\item $\mathit{Sys} \parallel  A_m \; \sqsubseteq \; \mathit{Sys}$, for $1 \leq m \leq 8$, 
		\item
		$\mathit{Sys} \parallel  A_m \;\: 
		\sqsubseteq_{m{+}4..\infty} \; \mathit{Sys}$, for $m > 8$. 
	\end{itemize}
\end{proposition}

In order to support the statement of Proposition~\ref{prop:att:DoS} (see also proof in appendix) we verify our \Uppaal{} SMC model of $\mathit{Sys}$ in which the communication network used by the controller to access the actuator is compromised. More precisely, we replace the \emph{\_Proxy\_Actuator\_} automaton of Figure~\ref{fig:network} with a compromised \nolinebreak  one, provided in Figure~\ref{fig:attack1_auto}, that implements the malicious activities of the 
MITM attacker $A_m$ of Example~\ref{exa:att:DoS}.

We have done our analysis, with a $99\%$ accuracy,  for execution traces that are at most $1000$ \nolinebreak  time \nolinebreak  units long and restricting the \emph{attack time} $m$   in the time interval $1..300$. The results of our analysis  \nolinebreak are: 
\begin{itemize}
	\item when $m \in 1..8$,
	the attack is harmless as the system results to be safe, deadlock free and alarm free, with probability $0.99$;  
	\item when $m \in  9..300$, we have the following situation: 
	\begin{itemize}
		\item the probability that at the attack time $m$ the controller sends a command to activate the cooling system (thus, triggering the attacker that will drop the command) can be obtained by verifying the property $\lozenge_{[0,m]}(\mathit{Cooling\_on} \, \wedge \, global\_clock \ge m)$; as shown in Figure~\ref{fig:attack1}, when $m$ grows in the time interval $1..300$, the resulting probability stabilises around the value $0.096$;
		\item up to the $m{+}3^\mathrm{th}$ time slot the system under attack remains safe, i.e., both properties $\square_{[1,m+3]}(\mathit{safe})$ and  $\square_{[1,m+3]}(\neg\mathit{deadlock})$ hold with probability $0.99$;  
		\item 
		up to the $m{+}4^\mathrm{th}$ time slot no alarms are fired, i.e., the  property $\square_{[1,m+4]}(\neg \mathit{alarm})$ holds with  probability $0.99$ (no false positives);   
		\item in the $m{+}4^\mathrm{th}$ time slot the system under attack might become unsafe as the probability, for $m \in 9..300$, that the property $\lozenge_{[0,m+4]}(\neg \mathit{safe})$ is satisfied  stabilises around the value $0.095$;\footnote{Since this probability coincides with that of $\lozenge_{[0,m]}(\mathit{Cooling\_on} \, \wedge \, global\_clock \ge m)$, it appears very likely that the activation of the cooling system in the $m^\mathrm{th}$ time slot triggers the attacker whose activity drags the system into an unsafe state with a delay of $4$ time slots. }
		\item in the  $m{+}5^\mathrm{th}$ time slot the IDS may fire an alarm as the probability, for $m \in 9..300$, that the property $\lozenge_{[0,m+5]}(\mathit{alarm})$ is satisfied stabilises around the value $0.094$;\footnote{As the two probabilities are pretty much the same, and  $\square_{[1,m+3]}(\mathit{safe})$ and $\square_{[1,m+4]}(\neg\mathit{alarm})$ hold, the IDS seems to be quite effective in detecting the violations of the safety conditions in the  $m{+}4^\mathrm{th}$ time slot, with only one time slot delay.} 
		\item the system under attack may deadlock as the property $\lozenge_{[0,1000]}(deadlocks)$ is satisfied with probability $0.096$.\footnote{Since the probabilities are still the same, we argue that when the system reaches an unsafe state then it is not able to recover and it is doomed to deadlock.}
	\end{itemize}
\end{itemize}

\begin{figure}[t]	
	\centering
	\includegraphics[scale=0.45]{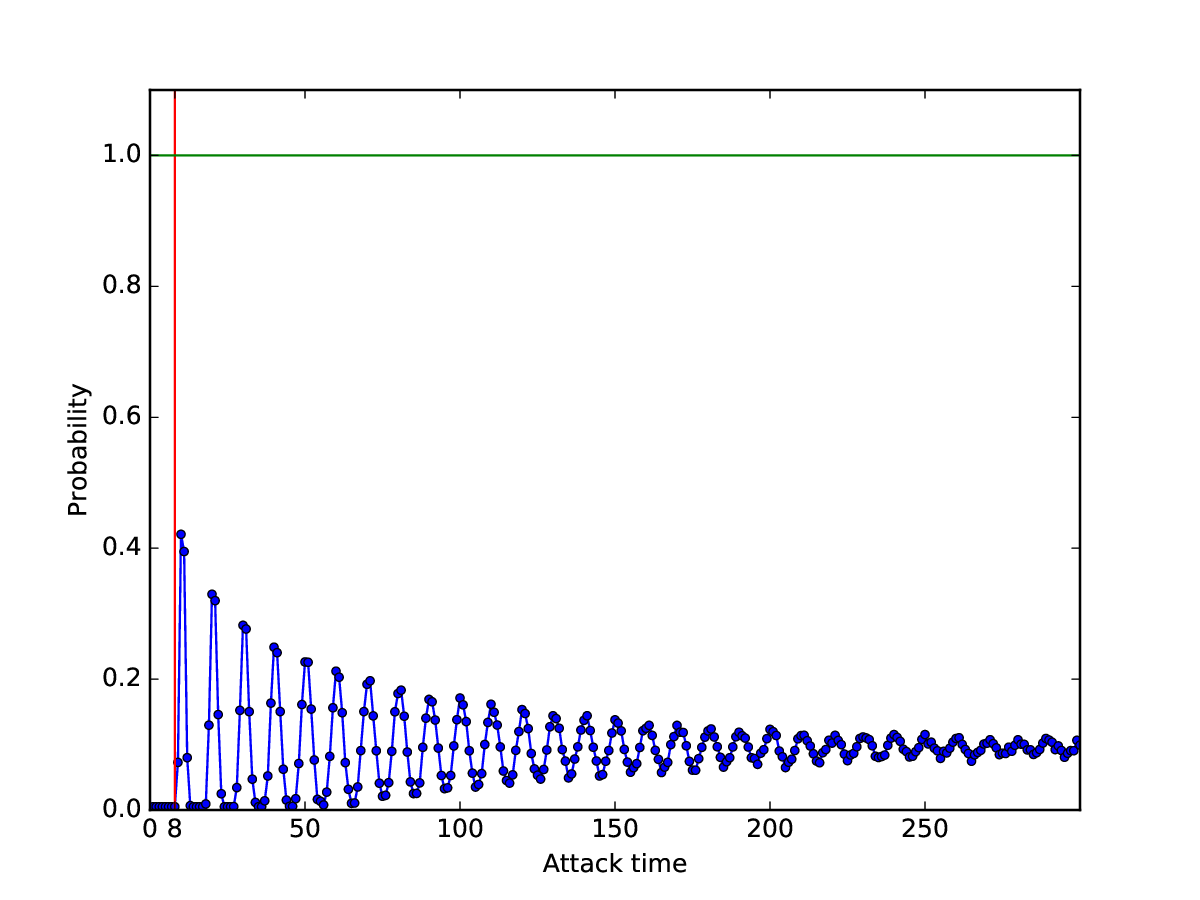} 
	\caption{Probability results of $\lozenge_{[0,m]}(\mathit{Cooling\_on} \, \wedge \, global\_clock \ge m)$  by varying $m$ in $1..300$}
	\label{fig:attack1}
\end{figure}

\begin{example}
	\label{exa:att:DoS2}
	Consider the following \emph{DoS/Integrity attack} to the sensor $s_{\mathrm{t}}$, of class $C \in [\I \rightarrow {\cal P}(2..\infty)]$ such that $C( \mbox{\Lightning}s_{\mathrm{t}}?) = \{ 2 \}$, 
	$C(\mbox{\Lightning}s_{\mathrm{t}}!) = 2..\infty$ and $C(\iota) = \emptyset$, for $ \iota \not \in \{\mbox{\Lightning}s_{\mathrm{t}}!, \mbox{\Lightning}s_{\mathrm{t}}?\}$. The attack \emph{begins} is activity in the time slot $m$, with $m > 8$, and then never stops: 
	\begin{displaymath}
	\begin{array}{rcl}
	A_m & = & \tick^{m-1}. A   \\[1pt]
	A   & = & 
	\lfloor \sniff x {s_{\mathrm{t}} }. \ifelse {x \leq 10} {B \langle x \rangle} {\tick.A}    \rfloor{}  \\[1pt]
	B(y) & = & \lfloor { \forge {y } {s_{\mathrm{t}} }.\tick.B \langle y \rangle } \rfloor {B \langle y \rangle } \, . 
	\end{array}
	\end{displaymath}
	
	Here, the attack $A_m$ behaves as follows. It 
	sleeps for $m-1$ time slots and then, in the following time slot,  it sniffs the current temperature at sensor $s_{\mathrm{t}}$. If the sensed temperature $v$ is greater than $10$, then it moves to the next time slot and restarts sniffing; otherwise  from that time on it will keep sending the same temperature $v$ to the logical components  (controller and IDS). 
	Actually, once the forgery activity starts, the process $\mathit{Ctrl}$ will always receive a temperature below $10$ and will never activate the cooling system (and consequently the IDS). 
	As a consequence, the system under attack $\mathit{Sys} \parallel A$ will first  move to an unsafe state until the invariant will be violated and the system will deadlock. 
	Indeed, in the worst execution scenario, already in the $m{+}1^\mathrm{th}$ time slot the temperature may exceed $10$ degrees, and after $4$ $\tick$-actions, in the $m{+}5^\mathrm{th}$ time slot, the system may violate the safety conditions emitting an $\unsafe$ action.
	Since the temperature will keep growing without any cooling activity, the deadlock of the \CPS{} cannot be avoided. 
	This is a \emph{lethal} attack, as it causes a shut down of the system; it is also a \emph{stealthy attack} as it remains undetected because the IDS never gets into action. 
\end{example}
\begin{proposition}
	\label{prop:att:dos-integrity}
	Let $\mathit{Sys}$ be our running example and $A_m$, for $m > 8$, be the attack defined in Example~\ref{exa:att:DoS2}. Then $\mathit{Sys} \parallel A_m \; \sqsubseteq _{{m+5}..\infty} \; \mathit{Sys}$. 
\end{proposition}

\begin{figure}[t]	
	\centering
	\includegraphics[scale=0.5]{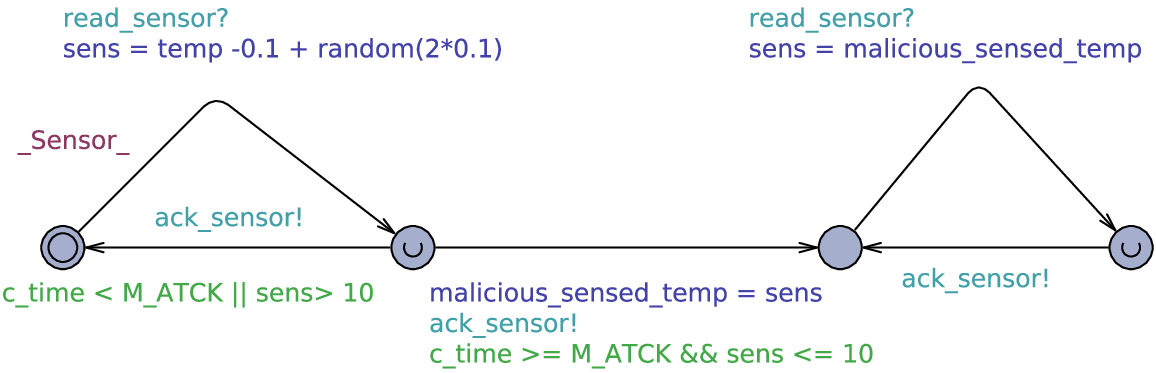} 
	\caption{\Uppaal{} SMC model for the attacker $A_m$ of Example~\ref{exa:att:DoS2}}
	\label{fig:attack2_auto}
\end{figure}

Here, we verify the \Uppaal{} SMC model of $\mathit{Sys}$ in which we assume that its sensor device is compromised (we recall that our MITM forgery attack on sensors or actuators can be assimilated to device compromise). The interested reader may find the proof in the appendix. In particular, we replace the \emph{\_Sensor\_} automaton of Figure~\ref{fig:physical} with a compromised one, provided in Figure~\ref{fig:attack2_auto}, and implementing the malicious activities of the MITM attacker $A_m$ of Example~\ref{exa:att:DoS2}.

We have done our analysis, with a $99\%$ accuracy, for execution traces that are at most $1000$ time units long and restricting the \emph{attack time} $m$ in the integer interval $9..300$. The results of our analysis \nolinebreak are:

\begin{itemize}
	\item up to the $m{+}4^\mathrm{th}$ time slot the system under attack remains safe, deadlock free, and alarm free, i.e., all three properties $\square_{[1,m+4]}(\mathit{safe})$, $\square_{[1,m+4]}(\neg\mathit{deadlock})$, and $\square_{[1,m+4]}(\neg\mathit{alarm})$ hold with  probability $0.99$; 
	
	\item in the $m{+}5^\mathrm{th}$ time slot the system under attack might become unsafe as the probability, for $m \in 9..300$, that the property $\lozenge_{[0,m+5]}(\neg \mathit{safe})$ is satisfied stabilises around $0.104$;
	
	\item the system under attack will eventually deadlock  not later that $80$ time slots after the attack time $m$, as the property $\square_{[m+80,1000]}(deadlocks)$ is satisfied with probability $0.99$; 
	
	\item finally, the attack is stealthy as the property $\square_{[1,1000]}(\neg \mathit{alarm})$ holds with probability $0.99$. 
\end{itemize}

Now, let us examine a similar but less severe attack. 
\begin{example}
	\label{exa:att:integrity}
	Consider the following DoS/Integrity attack  to sensor $s_{\mathrm{t}}$, of class $C \in [\I 
	\rightarrow {\cal P}(1..n)]$,  with $C(\mbox{\Lightning}s_{\mathrm{t}}!)=C(\mbox{\Lightning}s_{\mathrm{t}}?)=1..n$ and $C(\iota) = \emptyset$, for 
	$\iota \not \in \{\mbox{\Lightning}s_{\mathrm{t}}!,\mbox{\Lightning}s_{\mathrm{t}}?\}$:  
	\begin{displaymath}
	\begin{array}{rcl}
	A_n    & = &     \lfloor \sniff x {s_{\mathrm{t}} }. 
	\lfloor   \forge {x{-}4} {s_{\mathrm{t}} }.  \tick.A_{n-1 } \rfloor {A_{n-1 }} \rfloor {A_{n-1 }}, \:\textrm{  for $n>0$}\\
	A_0 &  = & \nil \, . 
	\end{array}
	\end{displaymath}
	
	In this attack, for $n$ consecutive time slots, 
	$A_n$ sends to the logical components (controller and IDS) the current sensed temperature decreased by 
	an offset $4$. The effect of this attack on the system depends on the
	\emph{duration} $n$ of the attack itself: 
	\begin{inparaenum}[(i)]
		\item for $n \leq 8$, the attack is harmless as the variable $\mathit{temp}$ may not reach a (critical) temperature above $9.9$;
		\item for $n=9$, the variable $\mathit{temp}$ might reach a temperature above $9.9$ in the $9^\mathrm{th}$ time slot, and the attack would delay the activation of the cooling system of one time slot; as a 
		consequence, the system  might get into an unsafe state in the time 
		interval $14..15$, but no alarm will be fired;  
		\item for $n \geq 10$, the system may get into an unsafe state in the time slot $14$ and in the following $n+11$ time slots; in this case, this would  not be \emph{stealthy attack} as the $\mathit{IDS}$ will fire the alarm with a delay of at most two time slots later, rather this is a \emph{temporary attack} that ends in the time slot $n+11$. 
	\end{inparaenum}
\end{example}

\begin{proposition}
	\label{prop:att:integrity}
	Let $\mathit{Sys}$ be our  use case  and $A_n$ be the attack defined in Example~\ref{exa:att:integrity}. Then: 
	\begin{itemize}
		\item $\mathit{Sys} \parallel  A_n \, \sqsubseteq \, \mathit{Sys}$, for $n \leq 8$,
		\item $\mathit{Sys} \parallel  A_n \; \sqsubseteq_{14..15} \; \mathit{Sys}$, for $n =9$, 
		\item $\mathit{Sys} \parallel A_n \; \sqsubseteq_{14..n{+}11} \; \mathit{Sys}$, for $n 
		\geq 10$.
	\end{itemize}
\end{proposition}

\begin{figure}[t]	
	\centering
	\includegraphics[scale=0.5]{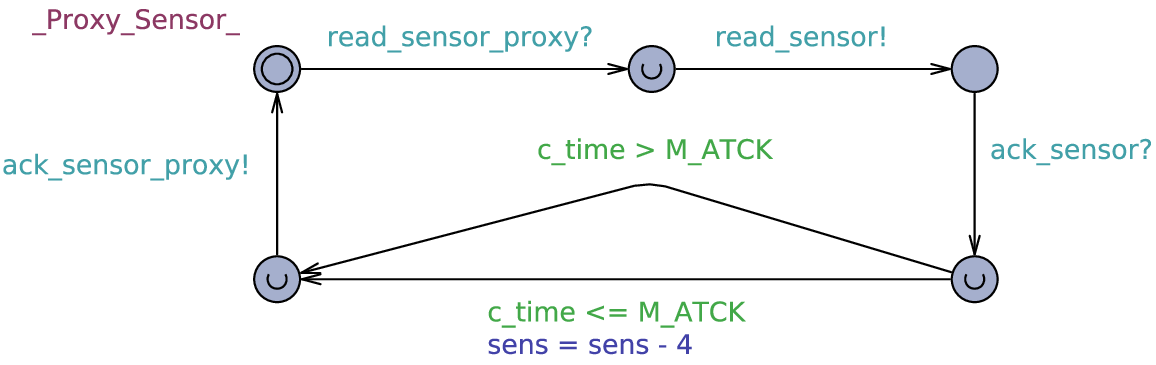} 
	\caption{\Uppaal{} SMC model for the attacker $A_n$ of Example~\ref{exa:att:integrity}}
	\label{fig:attack3_auto}
\end{figure}

Here, we verify the \Uppaal{} SMC model of $\mathit{Sys}$ in which we replace the \emph{\_Proxy\_Sensor\_} automaton of Figure~\ref{fig:network} with a compromised one, provided in Figure~\ref{fig:attack3_auto}, and implementing the MITM activities of the attacker $A_n$ of Example~\ref{exa:att:integrity}. The interested reader may find the proof in the appendix.

We have done our analysis, with a $99\%$ accuracy, for execution traces that are at most $1000$ time units long, and assuming that the \emph{duration of the attack} $n$ may vary in the integer interval $1..300$. The results of our analysis are: 

\begin{itemize}
	\item when $n \in 1..8$, the system under attack remains safe, deadlock free, and alarm free, i.e., all three properties $\square_{[1,1000]}(\mathit{safe})$, $\square_{[1,1000]}(\neg\mathit{deadlock})$, and $\square_{[1,1000]}(\neg\mathit{alarm})$ hold with  probability $0.99$; 
	\item when $n = 9$, we have the following situation: 
	\begin{itemize}
		\item the system under attack is deadlock free, i.e., the property $\square_{[1,1000]}(\neg\mathit{deadlock})$ holds with probability \nolinebreak $0.99$; 
		\item the system remains safe and alarm free, except for the time interval $14..15$, i.e., all the following properties $\square_{[1,13]}(\mathit{safe})$, $\square_{[1,13]}(\neg\mathit{alarm})$,  $\square_{[16,1000]}(\mathit{safe})$,  and $\square_{[16,1000]}(\neg\mathit{alarm})$ hold with  probability $0.99$;  
		\item in the time interval $14..15$, we may have violations of safety conditions, as the property $\lozenge_{[0,14]}(\neg \mathit{safe} \, \wedge \, global\_clock \ge 14)$	is satisfied with a probability $0.62$, while the property $\lozenge_{[0,15]}(\neg \mathit{safe} \, \wedge \, global\_clock \ge 15)$ is satisfied with probability $0.21$; both violations are \emph{stealthy} as the property $\square_{[14,15]}(\neg \mathit{alarm})$ holds with probability $0.99$; 
	\end{itemize}
	\item when $n \geq 10$, we have the following situation: 
	\begin{itemize}
		\item the system is deadlock free, i.e., the property {{\small $\square_{[1,1000]}(\neg\mathit{deadlock})$}} holds with probability \nolinebreak $0.99$; 
		\item the system remains safe  except for the time interval $14..n{+}11$, i.e., the two properties $\square_{[1,13]}(\mathit{safe})$ and  $\square_{[n+12,1000]}(\mathit{safe})$  hold with  probability $0.99$; 
		\item the system is alarm free except for the time interval $n{+}1..n{+}11$,
		i.e., the two properties   $\square_{[0,n]}(\neg\mathit{alarm})$ and $\square_{[n+12,1000]}(\neg\mathit{alarm})$ hold with  probability $0.99$; 
		\item in the $14^\mathrm{th}$ time slot the system under attack may reach an unsafe state as the probability, for $n \in 10..300$, that the property 
		$\lozenge_{[0,14]}(\neg \mathit{safe} \, \wedge \, global\_clock \ge 14)$ is satisfied stabilises around $0.548$;
		\item  once the attack has terminated, in the time interval $n{+}1..n{+}11$, 
		the system under attack has good chances to reach an unsafe state as the probability, for $n \in 10..300$, that the property 
		$\lozenge_{[0,n{+}11]}(\neg \mathit{safe} \, \wedge \, n{+}1 \leq global\_clock \leq n{+}11)$ is satisfied stabilises around $0.672$;  
		\item the violations of the safety conditions remain completely stealthy only up to the duration $n$ of the attack (we recall that $\square_{[0,n]}(\neg\mathit{alarm})$ is satisfied with probability $0.99$); the probability, for $n \in 10..300$,  that the property $\lozenge_{[0,{n{+}11}]}(\mathit{alarm})$ is satisfied stabilises around $0.13$; thus, in the time interval $n{+}1..n{+}11$, only a small portion of violations of safety conditions are detected by the IDS while a great majority of them remains stealthy. 
	\end{itemize}
\end{itemize}

\subsection{A technique for proving attack tolerance/vulnerability}
In this subsection, we provide sufficient criteria to prove attack
tolerance/vulnerability to attacks of an arbitrary class $C$. Actually, we
do more than that: we provide sufficient criteria to prove attack
tolerance/vulnerability to all attacks of any class $C'$ that is somehow
``weaker'' than a given class $C$.
\begin{definition} Let $C_i \in  [ \I \rightarrow {\cal P}(m_i..n_i) ]$, 
	for $i \in \{1,2 \}$, be 
	two classes of attacks, with $m_1..n_1 \subseteq m_2..n_2$. We say that $C_1$ is \emph{weaker} than $C_2$, written 
	$C_1 \preceq C_2$, if $C_1(\iota) \subseteq C_2(\iota)$ for any $\iota 
	\in \I$.
\end{definition}
Intuitively, if $C_1 \preceq C_2$ then: 
\begin{inparaenum}[(i)]
	\item the attacks of class $C_1$ might achieve fewer malicious activities than any attack of class $C_2$ (formally, there may be $\iota \in \I$ such that $C_1(\iota)=\emptyset$ and $C_2(\iota)\neq\emptyset$); 
	\item for those malicious activities $\iota \in \I$ achieved by the attacks of both classes $C_1$ and $C_2$ (i.e., $C_1(\iota) \neq \emptyset$ and $C_2(\iota) \neq \emptyset$), if they may be perpetrated by the  attacks of class $C_1$ at some time slot $k \in m_1..n_1$ (i.e., $k \in C_1(\iota)$) then all  attacks of class $C_2$ may do the same activity $\iota$ at the same time $k$ (i.e., $k \in C_2(\iota)$).
\end{inparaenum}

The next objective is to define a notion of \emph{most powerful attack} (also called \emph{top attacker}) of a given class $C$, such that, if a \CPS{} $M$ tolerates the most powerful attack of class $C$ then it also tolerates \emph{any} attack of class $C'$, with $C' \preceq C$. We will provide a similar condition for attack vulnerability: let $M$ be a \CPS{} vulnerable to $\mathit{Top}(C)$ in the time interval $m_1..n_1$; then, for any attack $A$ of class $C'$, with $C' \preceq C$, if $M$ is 
vulnerable to $A$ then it is so for a smaller time interval $m_2..n_2 \subseteq m_1..n_1$.

Our notion of top attacker has two extra ingredients with respect to the
physics-based attacks seen up to now: (i)~\emph{nondeterminism}, and (ii)~time-unguarded recursive processes. This extra power of the top attacker is not a problem as we are looking for sufficient criteria.

With respect to nondeterminism, we assume a generic procedure $\mathit{rnd}()$ that given an arbitrary set ${\cal Z}$ returns an element of ${\cal Z}$ chosen in a nondeterministic manner. This procedure allows us to express \emph{nondeterministic choice}, $P \oplus Q$, as an abbreviation for the process $\ifelse {\mathit{rnd}(\{\true,\false\})}  P Q $. Thus, let $\iota \in \{ \mbox{\Lightning}p ? : p \in {\cal S} \cup {\cal A} \} \cup \{ \mbox{\Lightning}p ! \, : p \in {\cal S} \cup {\cal A} \}$,  $m \in \mathbb{N}^{+}$, $n \in \mathbb{N}^{+} \cup
\{\infty\}$, with $m \leq n$, and ${\cal T} \subseteq m..n$, we define the attack process $ \mathit{Att}(\iota , k, {\cal T})$\footnote{In case of sensor sniffing, we might avoid to add this specific attack process as our top attacker process can forge any possible value without need to read sensors.} as the attack which may achieve the malicious activity $\iota$, at the time slot $k$, and which tries to do the same in all subsequent time slots of ${\cal T}$.  Formally, 
\begin{displaymath}
{\small
	\begin{array}{l}
	\mathit{Att}( \mbox{\Lightning}a?, k, {\cal T})   = 
	\ifthen {k \in {\cal T}} {  
		(\lfloor\drop x {a}.\mathit{Att}( \mbox{\Lightning}a?, k, {\cal T}) \rfloor 
		{\mathit{Att}( \mbox{\Lightning}a?, k{+}1, {\cal T})})  \, \oplus \,   \tick.  \mathit{Att}( \mbox{\Lightning}a? , k{+}1, {\cal T})} \:
	\\
	\hspace*{3.55cm}  \mathsf{else}\: \{  \ifelse { k < \mathrm{sup}({\cal T})}
	{\tick. \mathit{Att}( \mbox{\Lightning}a?, k{+}1, {\cal T})}
	{\nil} \}\\[2pt]
	\mathit{Att}( \mbox{\Lightning}s?, k, {\cal T})   = 
	\ifthen {k \in {\cal T}} {  
		(\lfloor\sniff x {s}.\mathit{Att}( \mbox{\Lightning}s?, k, {\cal T}) \rfloor 
		{\mathit{Att}( \mbox{\Lightning}s?, k{+}1, {\cal T})})  \, \oplus \,   \tick.  \mathit{Att}( \mbox{\Lightning}s? , k{+}1, {\cal T})} \:
	\\
	\hspace*{3.55cm}  \mathsf{else}\:  \{ \ifelse { k < \mathrm{sup}({\cal T})}
	{\tick. \mathit{Att}( \mbox{\Lightning}s?, k{+}1, {\cal T})}
	{\nil}  \}\\[2pt]
	\mathit{Att}( \mbox{\Lightning}p!, k, {\cal T})   =  \ifthen {k \in {\cal T}} {
		(\lfloor \forge {\mathit{rnd}(\mathbb{R})} {p}.\mathit{Att}( \mbox{\Lightning}p!, k, {\cal T}) \rfloor  
		{\mathit{Att}( \mbox{\Lightning}p!, k{+}1, {\cal T})})  \oplus   \tick.  \mathit{Att}( \mbox{\Lightning}p! , k{+}1, {\cal T})} \: 
	\\
	\hspace*{3.55cm} \mathsf{else}  \: \{  \ifelse {k < \mathrm{sup}({\cal T})}
	{\tick. \mathit{Att}( \mbox{\Lightning}p! , k{+}1, {\cal T})} {\nil} \}
	\, . 
	\end{array}
}
\end{displaymath}

Note that, for ${\cal T} = \emptyset$, we assume $\mathrm{sup}({\cal T})=- \infty$. 
We can now use the definition above to formalise the notion of most powerful attack of a given class $C$. 
\begin{definition}[Top attacker]
	Let $C \in [ \I \rightarrow {\cal P}(m..n)]$ be a class of attacks. We define 
	\begin{displaymath}
	\mathit{Top}(C)  \q = \q  \prod_{\iota \in \I } \mathit{Att}( 
	\iota , 1 , C(\iota))
	\end{displaymath}
	as the \emph{most powerful attack}, or \emph{top attacker}, of class $C$. 
\end{definition}

The following theorem provides soundness criteria for attack tolerance and attack vulnerability. 
\begin{theorem}[Soundness criteria]
	\label{thm:sound}
	Let $M$ be an honest and sound \CPS{}, $C$ an arbitrary class of attacks, and $A$ an attack of a class $C'$, with $C' \preceq C$.
	\begin{itemize}
		\item If $M \parallel \mathit{Top}(C) \, \sqsubseteq \, M$ then  
		$M \parallel A \, \sqsubseteq  \, M$. 
		\item If $M \parallel \mathit{Top}(C) \sqsubseteq_{{m_1}..{n_1}} M$  then either $M \parallel A \sqsubseteq M$ or $M \parallel A \sqsubseteq_{{m_2}..{n_2}} M$, with $m_2..n_2 \subseteq m_1..n_1$. 
	\end{itemize}
\end{theorem}

\begin{corollary}
	Let $M$ be an honest and sound \CPS{}, and $C$ a class of attacks. If
	$\mathit{Top}(C)$ is not lethal for $M$ then any attack $A$ of class $C'$, with $C' \preceq C$, is not lethal for $M$. If $\mathit{Top}(C)$
	is not a permanent attack for $M$, then any attack $A$ of class $C'$, with $C' \preceq C$, is not a permanent attack for $M$.
\end{corollary}

The following  example illustrates 
how Theorem~\ref{thm:sound} could be used to infer attack tolerance/vulnerability with respect to an entire class of attacks. 
\begin{example}\label{exa:top-attacker}
	Consider our running example $\mathit{Sys}$ and a class of attacks
	$C_m$, for $m \in \mathbb{N}$, such that 
	$C_m(\mbox{\Lightning}cool?)=C_m(\mbox{\Lightning}cool!)=\{ m \} $ and $C_m(\iota) = \emptyset$, for $\iota \not \in \{
	\mbox{\Lightning}cool? , \mbox{\Lightning}cool! \}$. 
	Attacks of class $C_m$ may tamper with the actuator $\mathit{cool}$ only in the time slot $m$ (i.e., in the time interval $m..m$). 
	The attack $A_m$ of Example~\ref{exa:att:DoS} is of class $C_m$.
\end{example}

In the following analysis in \Uppaal{} SMC of the top attacker $\mathit{Top}(C_m)$, we will show that both the vulnerability window and the probability of successfully attacking the system represent an upper bound for the attack $A_m$ of Example~\ref{exa:att:DoS} of class $C_m$. Technically, we verify the \Uppaal{} SMC 
model of $\mathit{Sys}$ in which we replace the  \emph{\_Proxy\_Actuator\_} automaton of Figure~\ref{fig:network} with a compromised one, provided in Figure~\ref{fig:attack_top}, and implementing the activities of the top attacker $\mathit{Top}(C_m)$. We carry out our analysis with a $99\%$ accuracy,  for execution traces that are at most $1000$ time slots long, limiting the \emph{attack time} $m$ to the integer interval $1..300$.

\begin{figure}[t]	
	\centering
	\includegraphics[scale=0.5]{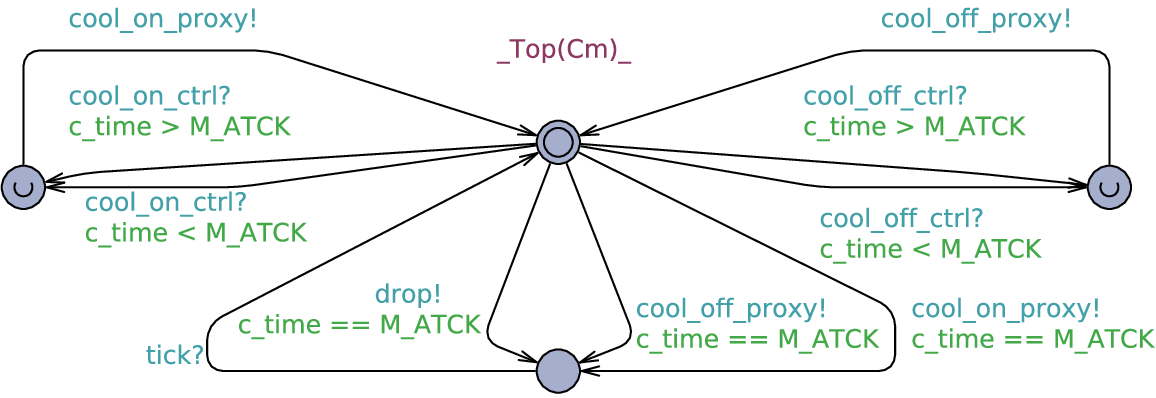} 
	\caption{\Uppaal{} SMC model for the top attacker $\mathit{Top}(C_m)$ of Example~\ref{exa:top-attacker}}
	\label{fig:attack_top}
\end{figure}

\begin{figure}[t]	
	\centering
	\includegraphics[scale=0.45]{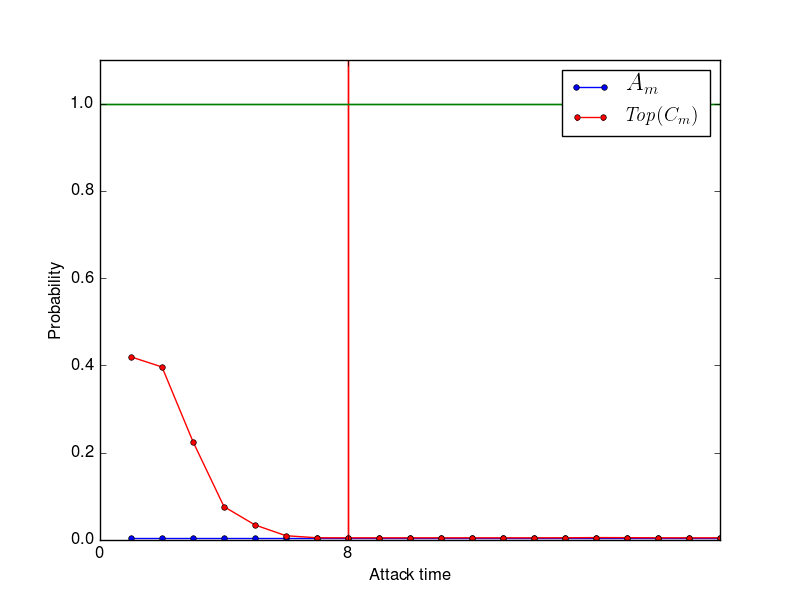} 
	\caption{Results of 
		$\lozenge_{[0,m+3]}(\mathit{deadlock} \wedge \mathit{global\_clock}\ge m+1 )$ by varying the attack time $m$}
	\label{fig:top_attacker_early deadlock}
\end{figure}

\begin{figure}[t]	
	\centering
	\includegraphics[scale=0.45]{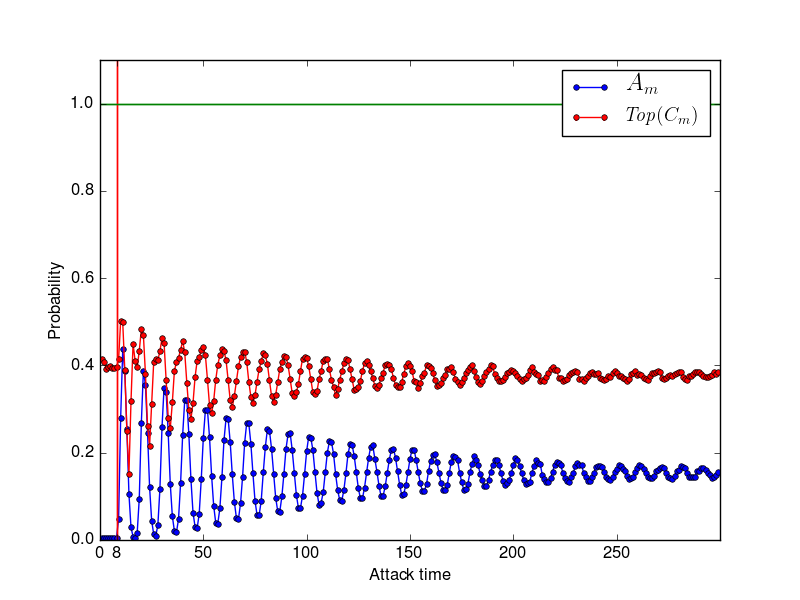} 
	\caption{Results of 
		$\lozenge_{[0,1000]}(\mathit{deadlock} \wedge \mathit{global\_clock}\ge m+4 )$
		by varying the attack  time $m$}
	\label{fig:top_attacker_deadlock}
\end{figure}

To explain our analysis further, let us provide details on how 
$\mathit{Top}(C_m)$ affects $\mathit{Sys}$ when compared to the  attacker $A_m$ of class $C_m$ seen in the Example~\ref{exa:att:DoS}. 
\begin{itemize}
	\item In the time interval $1..m$, the attacked system remains safe, deadlock free, and alarm free. Formally, the three properties $\square_{[1,m]}(\mathit{safe})$, $\square_{[1,m]}(\neg\mathit{deadlock})$ and $\square_{[1,m]}(\neg\mathit{alarm})$  hold with  probability $0.99$. Thus, in this time interval, the top attacker is harmless, as well as $A_m$.  
	\item In the time interval $m{+}1..m{+}3$, the system exposed to the top attacker may deadlock when $m \in 1..8$; for $m>8$ the system under attack is deadlock free (see Figure~\ref{fig:top_attacker_early deadlock}). This is because the top attacker, unlike the  attacker $A_m$, can forge in the first $8$ time slots cool-on commands turning on the cooling and dropping the temperature below zero in the time interval $m{+}1..m{+}3$. 
	Note that no alarms or unsafe behaviours occur in this case, as neither the safety process nor the IDS check whether the temperature drops below a certain threshold. Formally, the properties $\square_{[m+1,m+3]}(\mathit{safe})$ and $\square_{[m+1,m+3]}(\neg\mathit{alarm})$ hold with probability $0.99$, as already seen for the attacker $A_m$.
	\item In the time interval $m{+}4..1000$,  the top attacker has better chances to deadlock the system when compared with the attacker  $A_m$ (see Figure~\ref{fig:top_attacker_deadlock}). With respect to safety and alarms, 
	the top attacker and the attacker $A_m$ have the same probability of success (the properties $\square_{[m+4,1000]}(\mathit{safe})$ and $\square_{[m+4,1000]}(\neg\mathit{alarm})$ return the same probability results). 
\end{itemize}

This example shows how the verification of a top attacker $\mathit{Top}(C)$  provides an upper bound of the effectiveness of the entire class of attacks $C$,  in terms of both vulnerability window and probability of 
successfully attack the system. Of course, the accuracy of such approximation cannot be  estimated a priori.

\section{Impact of a physics-based attack}  
\label{sec:impact}

In the previous section,  we have grouped physics-based attacks by
focussing on the physical devices under attack and the timing aspects of
the attack (Definition~\ref{def:attacker-class}). 
Then, we have provided a formalisation of when a \CPS{} should be considered tolerant/vulnerable to an attack  
(Definition~\ref{def:attack-tolerance}). 
In this section, we show that these two formalisations are important not only to demonstrate the tolerance (or vulnerability) of a CPS with respect to certain attacks, but also to evaluate the disruptive impact of those attacks on the target \CPS{}~\cite{GeKiHa2015,Milosevic2018a}.

The goal of this section is to provide a \emph{formal metric} to
estimate the impact of a \emph{successful attack} on the \emph{physical behaviour} of a \CPS{}. 
In particular, we focus on the ability that an attack may have to drag a CPS out of the correct behaviour modelled by its evolution map, with the given uncertainty. 

Recall that $\evolmap{}$ is 
\emph{monotone} with respect to the uncertainty. Thus, as stated in Proposition~\ref{prop:monotonicity}, an increase of the uncertainty may translate into a widening of the range of the possible behaviours of the \CPS{}.
In the following, given the physical environment $E =  \envCPS
{\evolmap{}}
{\measmap{}}
{\invariantfun{}}
{\safefun{}}
{\uncertaintyfun{}}{\errorfun{}}
$, we write $\replaceENV
E {\uncertaintyfun{}} {\uncertaintyfun'{}}$ as an abbreviation for $ \envCPS
{\evolmap{}}
{\measmap{}}
{\invariantfun{}}
{\safefun{}}
{\uncertaintyfun'{}}{\errorfun{}}
$; similarly, 
for $M = \confCPS {E;S}  P$ we write $\replaceENV
M {\uncertaintyfun{}} {\uncertaintyfun'{}}$ for  $\confCPS {{\replaceENV E {\uncertaintyfun{}} {\uncertaintyfun'{}}};S}  P$.
\begin{proposition}[Monotonicity] 
	\label{prop:monotonicity}
	Let $M$ be an honest and sound \CPS{} with uncertainty $\uncertaintyfun{}$. If $\uncertaintyfun{} \leq \uncertaintyfun'{}$ and $M \trans{t} M'$ then $\replaceENV M {\uncertaintyfun{}}  {\uncertaintyfun'{}} \trans{t} \replaceENV {M'} {\uncertaintyfun{}} {\uncertaintyfun'{}}$.
\end{proposition}

However, a wider uncertainty in the model does not always correspond to a widening of the possible behaviours of the \CPS{}. In fact, this depends on the \emph{intrinsic tolerance} of a \CPS{} with respect to changes in the uncertainty function. In the following, we will write $\uncertaintyfun{} + \uncertaintyfun'{}$ to denote the function $\uncertaintyfun''{} \in \mathbb{R}^{\mathcal X}$ such that $\uncertaintyfun''{}(x) = \uncertaintyfun{}(x) + \uncertaintyfun'{}(x)$, for any $x \in \mathcal X$.

\begin{definition}[System $\xi$-tolerance] 
	An honest and sound \CPS{} $M$ with uncertainty $\uncertaintyfun{}$   is said \emph{$\xi$-tolerant}, for $\xi \in \mathbb{R}^{{\mathcal X}}$ and $\xi \geq 0$, if
	\begin{displaymath}
	\xi \, = \, \sup \big\{ \xi'  : \,
	\replaceENV M {\uncertaintyfun{}}  {{\uncertaintyfun{}}+{\eta}}
	\sqsubseteq M, \text{ for any } 0 \leq \eta \leq \xi'     
	\big\}.
	\end{displaymath}
\end{definition}

Intuitively, if a \CPS{} $M$ has been designed with a given uncertainty
$\uncertaintyfun{}$, but $M$ is actually $\xi$-tolerant, with $\xi > 0$, then the uncertainty $\uncertaintyfun{}$ is somehow underestimated: the real uncertainty of  $M$  is given by $\uncertaintyfun{} + \xi$. 
This information is quite important when trying to estimate the impact of an attack on a \CPS{}. In fact, if a system  $M$ has been designed with a given uncertainty $\uncertaintyfun{}$, but $M$ is actually $\xi$-tolerant, with $\xi > 0$, then an attack has (at least) a ``room for manoeuvre'' $\xi$ to degrade the whole \CPS{} without being observed (and hence detected). 

Let $\mathit{Sys}$ be our running example. In the rest of the section, with an abuse of notation, we will write $\replaceENV {\mathit{Sys}}  {\delta} {\gamma} $ to denote $\mathit{Sys}$ where the uncertainty $\delta$ of the variable $\mathit{temp}$ has been \nolinebreak replaced \nolinebreak with \nolinebreak $\gamma$.
\begin{example}
	\label{exa:toll}
	The \CPS{} $\mathit{Sys}$  is $\frac{1}{20}$-tolerant as 
	\begin{math}
	\sup \big\{ \xi' : \replaceENV {\mathit{Sys}} \delta {\delta {+} \eta}  \sqsubseteq  \mathit{Sys} , \text{ for } 0 \leq \eta \leq  \xi' \big\}
	\end{math} is equal to $\frac{1}{20}$. 
	Since $ \delta + \xi = \frac{8}{20} + \frac{1}{20}=\frac{9}{20}$, then
	this statement relies on the following proposition whose proof can be found in the appendix. 

\end{example}

\begin{proposition} We have
	\label{prop:toll} \
	\begin{itemize}
		\item $ \replaceENV {\mathit{Sys}} \delta  \gamma \, \sqsubseteq \, \mathit{Sys} $, for $\gamma \in (\frac{8}{20}, \frac{ 9}{20})$, 
		\item $ \replaceENV  {\mathit{Sys}} \delta \gamma \, \not\sqsubseteq \,  \mathit{Sys}$, for $\gamma >\frac{ 9}{20}$.  
	\end{itemize}
\end{proposition}

Now everything is in place to define our metric to estimate the 
impact of an attack. 
\begin{definition}[Impact] 
	\label{def:attack-xi-tolerance}
	Let  $M$ be an honest and sound \CPS{} with uncertainty $\uncertaintyfun{}$. We say that an attack $A$ has \emph{definitive impact} $\xi$ on the system $M$ if 
	\begin{displaymath}
	\xi  = \inf \big\{ \xi'  :  \xi' \in \mathbb{R}^{{\mathcal X}} 
	\: \wedge \: \xi'{>}0 \:  \wedge \: M \parallel  A  \, \sqsubseteq \,  
	{\replaceENV M {\uncertaintyfun{}}  {{\uncertaintyfun{}}+{\xi'}}  }  \big\}.
	\end{displaymath}%
	It has \emph{pointwise impact} $\xi$ on the system $M$ at time $m$ if 
	{\small 
		\begin{displaymath}
		\xi  {=}  \inf \big\{ \xi'  :  \xi' \in \mathbb{R}^{{\mathcal X}} 
		\, \wedge \, \xi'{>}0 
		\,  \wedge \, M \parallel  A   \, \sqsubseteq_{m..n} 
		\replaceENV M  {\uncertaintyfun{}}  {{\uncertaintyfun{}} + {\xi'}},  n \in \mathbb{N}^+ \cup \{\infty\}   \big\}.
		\end{displaymath}
	}
\end{definition}

Intuitively, the impact of an attacker $A$ on a system $M$ measures the perturbation introduced by the presence of the attacker in the compound system $ M \parallel A$ with respect to the original system $M$. 
With this definition, we can establish either the definitive (and hence maximum) impact of the attack $A$ on the system $M$, or the impact at a specific time $m$.
In the latter case, by definition of $\sqsubseteq_{m..n}$, there are two possibilities: either the impact of the attack keeps growing after time $m$, or in the time interval $m{+}1$, the system under attack deadlocks.

The impact of $\mathit{Top}(C)$ provides an upper bound for the impact of all attacks of class $C'$, $C' \preceq C$, as shown in the following theorem (proved in the appendix).
\begin{theorem}[Top attacker's impact]
	\label{thm:sound2}
	Let  $M$ be an honest and sound \CPS{}, and 
	$C$  an arbitrary class of attacks. Let $A$ be an arbitrary attack of class $C'$, with $C' \preceq C$.  
	\begin{itemize}[noitemsep]
		\item The definitive impact of $\mathit{Top}(C)$ on $M$ is greater than or equal to the definitive impact of $A$ on $M$.
		\item If $\mathit{Top}(C)$ has pointwise impact $\xi$ on $M$ at time $m$, and $A$ has pointwise impact $\xi'$  on $M$ at time $m'$, with  
		$m' \leq m$, then $\xi' \leq \xi $.  
	\end{itemize}
\end{theorem}

In order to help the intuition on the impact metric defined in Definition~\ref{def:attack-xi-tolerance}, we give a couple of examples. Here, we focus on the role played by the size of the vulnerability window.
\begin{example}
	\label{exa:effect1}
	Let us consider the attack $A_n $ of Example~\ref{exa:att:integrity}, for $n \in \{ 8, 9, 10 \}$. 
	Then, 
	\begin{itemize}[noitemsep]
		\item $A_8$ has definitive impact $0$ on $\mathit{Sys}$,
		\item $A_9$ has definitive impact $0.2\overline{3}$ on $\mathit{Sys}$,
		\item $A_{10}$ has definitive impact $0.4$ on $\mathit{Sys}$.
	\end{itemize} 
	Formally, the impacts of these three attacks are obtained by calculating  
	\begin{displaymath}
	\inf \{ \xi' :\, \xi'> 0 \, \wedge \, \mathit{Sys}  \parallel  A_{n} \, \sqsubseteq \, \replaceENV {\mathit{Sys}} {\delta} {\delta + {\xi'}} 
	\}\,, 
	\end{displaymath}
	for $n \in \{ 8, 9, 10 \}$. Attack $A_9$ has a very low impact on $\mathit{Sys}$ as it may drag the system into a \emph{temporary unsafe state} in the time interval $14..15$, whereas $A_{10}$ has a slightly stronger impact as it may induce a \emph{temporary unsafe state} during the larger time interval $14..21$. 
	Technically, since $\delta + \xi = 0.4 + 0.4=0.8$,
	the calculation of the impact of $A_{10}$ relies on the following proposition whose proof can be found in the appendix. 

\end{example}
\begin{proposition}
	\label{prop:effect1}
	Let $A_{10}$ be the attack defined in Example~\ref{exa:att:integrity}. Then:
	\begin{itemize}
		\item 
		$\mathit{Sys} \parallel A_{10} \, \not \sqsubseteq \,  \replaceENV {\mathit{Sys}}  {\delta} {\gamma} $, for $\gamma \in (0.4, 0.8)$, 
		\item 
		$\mathit{Sys} \parallel A_{10} \, \sqsubseteq \, \replaceENV {\mathit{Sys}}
		{\delta} {\gamma} $, for $\gamma > 0.8$. 
	\end{itemize}
\end{proposition}

On the other hand, the attack provided in Example~\ref{exa:att:DoS2}, driving the system to a (\emph{permanent}) \emph{deadlock state}, has a much stronger impact on the \CPS{} $\mathit{Sys}$ than the attack of Example~\ref{exa:att:integrity}.

\begin{example}
	\label{exa:effect2}
	Let us consider the attack $A_m$ of Example~\ref{exa:att:DoS2}, for $m > 8$.  
	As already discussed, this is a stealthy lethal attack that has a very severe and high impact. In fact, it has a definitive impact of $8.5$ on the \CPS{} $\mathit{Sys}$.  Formally, 
	\begin{displaymath}
	8.5= \inf \big \{ \, \xi'  :\: \xi'> 0 \: \wedge \: \mathit{Sys} \parallel  A_m \, \sqsubseteq \, \replaceENV {\mathit{Sys}} \delta  {\delta +
		\xi'} \big \}. 
	\end{displaymath}
	Technically, since $\delta + \xi = 0.4+8.5=8.9$, what  stated in this example relies on the following proposition whose proof can be found in the appendix.

\end{example}
\begin{proposition}
	\label{prop:effect2}
	Let  $A_m$ be the attack defined in Example~\ref{exa:att:DoS2}. Then:
	\begin{itemize}[noitemsep]
		\item $\mathit{Sys} \parallel A_m \, \not \sqsubseteq \, 
		\replaceENV {\mathit{Sys}}
		\delta  \gamma$, for $\gamma \in (0.4,8.9)$, 
		\item $\mathit{Sys} \parallel A_m \, \sqsubseteq \, \replaceENV {\mathit{Sys}}
		\delta  \gamma$, for $\gamma >8.9$. 
	\end{itemize} 
\end{proposition}

Thus, Definition~\ref{def:attack-xi-tolerance} provides an instrument to estimate the impact of a \emph{successful attack} on a CPS in terms of the perturbation introduced both on its physical and on its logical processes. However, there is at least another question that  
a \CPS{} designer could ask: ``Is there a way to
estimate the chances that an attack will be successful during the
execution of my \CPS{}?'' To paraphrase in a more operational manner: how
many execution traces of my \CPS{} are prone to be attacked by a specific
attack? As argued in the future work, we believe that \emph{probabilistic metrics} might reveal to be very useful in this respect~\cite{ifm2018,MaMe14}. 

\section{Conclusions, related and future work}
\label{sec:conclusions}

\subsection{Summary}
We have provided \emph{theoretical foundations} to reason about
and \emph{formally} detect attacks to physical devices of \CPS{s}.
A straightforward utilisation of these methodologies is for 
\emph{model-checking} or \emph{monitoring} in order to be able to formally analyse security properties of \CPS{s} either before system deployment or, when static analysis is not feasible, at runtime to promptly detect undesired behaviours. 
To that end, we have proposed a hybrid process calculus, called \cname{}, as a formal \emph{specification language} to model physical and cyber components of \CPS{s} as well as MITM physics-based attacks. Note that our calculus is general enough to represent \emph{Supervisory Control And Data Acquisition} (SCADA) systems as cyber components which can easily interact with controllers and IDSs via channel communications. SCADA systems are the main technology used by system engineers to supervise the activities of  complex CPSs. 

Based on \cname{} and its labelled transition semantics, we have formalised a threat model for \CPS{s} by grouping physics-based attacks in classes, 
according to the target physical devices and two timing parameters: begin and duration of the attacks. Then, we developed two different \emph{compositional} trace semantics for \cname{} to assess \emph{attack tolerance/vulnerability} with respect to a given attack. Such a tolerance may hold \emph{ad infinitum} or for a limited amount of time. In the latter case, the \CPS{} under attack is vulnerable and the attack affects the observable behaviour of the system only after a certain point in time, when the attack itself may already be achieved or still working.

Along the lines of GNDC~\cite{FM99}, we have defined a notion of \emph{top attacker}, $\mathit{Top}(C)$, of a given class of attacks $C$, which has been used to provide sufficient criteria to prove attack tolerance/vulnerability to all attacks of class $C$ (and weaker ones).

Then, we have provided a metric to estimate the \emph{maximum impact} introduced in the system under attack with respect to its genuine behaviour, according to its evolution law and the uncertainty of the model. We have proved that the impact of the most powerful attack $\mathit{Top}(C)$ represents an upper bound for the impact of any attack $A$ of class $C$ \nolinebreak (and \nolinebreak weaker \nolinebreak ones).

Finally, we have formalised a \emph{running example} in \Uppaal{} SMC~\cite{David:2015:UST:2802769.2802840}, the statistical extension of the \Uppaal{} model checker~\cite{DBLP:conf/sfm/BehrmannDL04}. 
Our goal was to test \Uppaal{} SMC as an automatic tool for the \emph{static security analysis} of a simple but significant \CPS{} exposed to a number of different physics-based attacks with different impacts on the system under attack. Here, it is important to note that, although we have verified most of the properties stated in the paper, we have not been  able to capture time properties on the responsiveness of the IDS to violations of the safety conditions. Examples of such properties are: 
\begin{inparaenum}[(i)]
	\item there are time slots $m$ and $k$ such that the system may have an unsafe state at some time $n > m$, and the IDS detects this violation with a delay of at least $k$ time slots ($k$ being a lower bound of the reaction time of the IDS), or 
	\item  there is a time slot  $n$ in which the IDS fires an alarm  but neither an unsafe state nor a deadlock occurs in the time interval $n{-}k..n{+}k$: this would provide a tolerance of the occurrence of  \emph{false positive}.
\end{inparaenum}
Furthermore, \Uppaal{} SMC does not support the verification of nested formulae. Thus, although from a designer's point of view it would have been much more practical to verify a logic formula of the form $\exists\lozenge(\square_{[t, t+5]} temp > 9.9)$ to check safety and invariant conditions, in \Uppaal{} SMC we had to implement a \_Safety\_ automaton that is not really part of our \CPS{} (for more details see the  discussion of related work).

\subsection{Related work}
\label{sec:related}
A number of approaches have been proposed for modelling \CPS{s} using
\emph{hybrid process algebras}~\cite{CuRe05,BergMid05,vanBeek06,RouSong03,HYPE}.
Among these approaches, our calculus \cname{} shares some similarities with the $\phi$-calculus~\cite{RouSong03}. However, unlike \cname{}, in the $\phi$-calculus, given a hybrid system $(E,P)$, the process $P$ can dynamically change the evolution law in $E$. Furthermore, the $\phi$-calculus does not have a representation of physical devices and measurement law, which are instead crucial for us to model physics-based attacks that operate in a timely fashion on sensors and actuators. More recently, Galpin et al.~\cite{HYPE} have proposed a process algebra in which the continuous part of the system is represented by appropriate variables whose changes are  determined by active influences (i.e., commands on actuators).

Many good surveys on the security of cyber-physical systems have been published recently (see, e.g., \cite{ACM-survey2018,survey-CPS-security-2016,SurveyComputerIndustry2018,survey-CPS-security-2019}), including a survey of 
surveys~\cite{SurveyOfSurverys2017}. 
In particular, the surveys~\cite{survey-CPS-security-2019,survey-CPS-security-2016} provide a systematic categorisation of $138$ selected papers on \CPS{} security. 
Among those $138$ papers, $65$  adopt a discrete notion of time similar to ours, $26$ a continuous one, $55$ a quasi-static time model, and the rest use a hybrid time model. 
This study encouraged us in adopting a discrete time model for physical processes rather than a continuous 
one.  
Still, one might wonder what is actually lost when one adopts a discrete rather than a continuous time model, in particular when the attacker has the possibility to move in a continuous time setting. A continuous time model is, of course, more expressive. For instance, Kanovich et al.~\cite{Kanovich2015} identified a novel vulnerability in the context of \emph{cryptographic 
	protocols} for \CPS{s} in which the attacker works in a continuous-time setting to fool discrete-time verifiers. However, we believe that, for \emph{physics-based attacks}, little is lost by adopting a discrete time model.  In fact, 
sensor measurements and actuator commands are elaborated within controllers, which are digital devices with an intrinsic discrete notion of time. 
In particular, with respect to dropping of actuator commands and forging of sensor measurements, there are no differences between discrete-time and continuous-time attackers given that to achieve those malicious activities the attacker has to synchronise with the controller. 
Thus, there remain only two potential malicious activities: sensor sniffing and forging of actuator commands. Can a continuous-time attacker, able to carry out these two malicious activities, be more disruptive than a similar attacker adopting a discrete-time model? This would only be the case when dealing with very rare physical processes changing their physical state in an extremely fast way, faster than the controller which is the one dictating the discrete time of the CPS. However, we believe that CPSs of this kind would be hardly controllable as they would pose serious safety issues even in the absence of any attacker.

The survey~\cite{ACM-survey2018} provides an exhaustive review of papers on physics-based anomaly detection proposing a unified taxonomy, whereas the survey~\cite{SurveyComputerIndustry2018} presents the main  solutions in the estimation of the consequences of cyber-attacks, attacks modelling and detection, and the development of security architecture (the main types of attacks and threats against \CPS{s} are analysed and grouped in a tree structure). 

Huang et al.~\cite{HCALTS2009} were among the first to propose \emph{threat models} for \CPS{s}. Along with~\cite{KrCa2013,BestTime2014}, they stressed the role played by timing parameters on integrity and DoS attacks.

Gollmann et al.~\cite{GGIKLW2015} discussed possible goals
(\emph{equipment damage}, \emph{production damage}, \emph{compliance
	violation}) and \emph{stages} (\emph{access}, \emph{discovery},
\emph{control}, \emph{damage}, \emph{cleanup}) of physics-based attacks.
In this article, we focused on the ``damage'' stage, where the attacker already has a rough idea of the plant and the control architecture of the target \CPS{}.As we remarked in Section~\ref{sec:introduction}, here we focus on an attacker who has already entered the CPS, without considering how the attacker gained access to the system, which could have happened in several ways, for instance by 
attacking an Internet-accessible controller or
one of the communication protocols.

Almost all papers discussed in the surveys mentioned above \cite{survey-CPS-security-2019,ACM-survey2018,SurveyComputerIndustry2018} investigate attacks on \CPS{s} and their protection by relying on \emph{simulation test systems} to validate the results, rather than \emph{formal methodologies}.
We are aware of a number of works applying \emph{formal methods} to \CPS{} security, although they apply methods, and most of the time have goals, that are quite different from ours. We discuss the most significant ones on the following.

Burmester et al.~\cite{BuMaCh2012} employed \emph{hybrid timed automata} to give a threat 
framework based on the traditional Byzantine faults model for
crypto-security. 
However, as remarked in~\cite{TeShSaJo2015}, physics-based attacks and faults have inherently distinct characteristics.
Faults are considered as physical events that affect the system behaviour where simultaneous events 
don't act in a coordinated way, whereas  
cyber attacks may be performed over a significant number of attack points and in \nolinebreak a \nolinebreak coordinated \nolinebreak way.

In~\cite{Vig2012}, Vigo presented an attack scenario that addresses some
of the peculiarities of a cyber-physical adversary, and discussed how this scenario relates to other attack models popular in the security protocol literature. Then, in~\cite{VNN2013} Vigo et al.\ proposed an untimed calculus of broadcasting processes 
equipped with 
notions of 
failed and unwanted communication. These works differ quite considerably from ours, e.g., they focus on \emph{DoS attacks} without taking into consideration timing aspects or impact of the attack.

C\'ombita et al.~\cite{Cardenas2015} and Zhu and
Basar~\cite{game-theory-CPS2015} applied \emph{game theory} to capture the conflict of goals between an attacker who seeks to maximise the damage inflicted to a \CPS{}'s security and a defender who aims to minimise it~\cite{game-theory-2013}.

Rocchetto and Tippenhauer~\cite{RocchettoTippenhauer2016a} introduced a taxonomy of the diverse attacker models proposed for \CPS{} security and outline requirements for generalised attacker models; in~\cite{RocchettoTippenhauer2016b}, they then proposed an extended \emph{Dolev-Yao attacker model} suitable for \CPS{s}. 
In their approach, physical layer interactions are modelled as abstract
interactions between logical components to support reasoning on the
physical-layer security of \CPS{s}. This is done by introducing additional orthogonal channels. Time is not represented.

Nigam et al.~\cite{Nigam-Esorics2016} worked around the notion of \emph{Timed Dolev-Yao Intruder Models for Cyber-Physical Security Protocols} by bounding the number of intruders required for the automated verification of such protocols. Following a tradition in security protocol analysis, they provide an answer to the question: How many intruders are enough for verification and where should they be placed? They also extend the strand space model to \CPS{} protocols by allowing for the symbolic representation of time, so that they can use the tool Maude~\cite{RT-MAUDE} along with SMT support. Their notion of time is however different from ours, as they focus on the time a message needs to travel from an agent to another. The paper does not mention physical devices, such as sensors and/or actuators.

There are a few approaches that carry out \emph{information flow security analysis} on discrete/continuous models for \CPS{s}.  Akella et al.\ \cite{Akella2010}  proposed an approach to perform information flow analysis, including both trace-based analysis and automated analysis through process algebra specification. This approach has been used to verify process algebra models of a gas pipeline system and a smart electric power grid system.
Bodei et al.~\cite{BDFG17} proposed a process calculus supporting a control flow analysis that safely approximates the abstract behaviour of IoT systems. Essentially, they track how data spread from sensors to the logics of the network, and how physical data are manipulated. In \cite{Bodei2018}, the same authors extend their work to infer \emph{quantitative measures} to establish the cost of possibly security countermeasures, in terms of time and energy. 
Another discrete model has been proposed by Wang~\cite{Wang2014}, where Petri-net models have been used to verify \emph{non-deducibility security properties} of a natural gas pipeline system. 
More recently, Bohrer and Platzer~\cite{LICS-Platzer2018} introduced dHL, a hybrid logic for verifying cyber-physical hybrid-dynamic information flows, communicating information through both discrete computation and physical dynamics, so security is ensured even when attackers observe \emph{continuously-changing values} in continuous time. 

Huang et al.~\cite{HZTYQ18}  proposed a \emph{risk assessment method} that uses a Bayesian network to model the attack propagation process and infers the probabilities of sensors and actuators to be compromised. These probabilities are fed into a stochastic hybrid system (SHS) model to predict the evolution of the physical process being controlled. Then, the security risk is quantified by evaluating the system availability with \nolinebreak the \nolinebreak  SHS \nolinebreak model. 

{As regards tools for the formal verification of CPSs, we remark that we tried to verify our case study using \emph{model-checking tools} for distributed systems 
	such as PRISM~\cite{PRISM}, \Uppaal{}~\cite{UPPAAL}, Real-Time Maude~\cite{RT-MAUDE}, and \textsf{prohver} within the {\textsc M{\scriptsize ODEST} T{\scriptsize OOLSET}}~\cite{MODEST-toolset}. 
	In particular, as our example adopts a discrete notion of time, we started looking at tools supporting discrete time. PRISM, for instance, relies on Markov decision processes or discrete-time Markov chains, depending on whether one is interested in modelling nondeterminism or not. It supports the verification of both CTL and LTL properties (when dealing with nonprobabilistic systems). 
	This allowed us to express the formula $\exists\lozenge(\square_{[t, t+5]} temp > 9.9)$ to verify violations of the safety conditions, avoiding the implementation of the \emph{\_Safety\_} automaton. 
	However, using integer variables to represent state variables with a fixed precision requires the introduction of extra transitions (to deal with nondeterministic errors), which significantly complicates the PRISM model. 
	In this respect, \Uppaal{}  appears to be more efficient than PRISM, as we have been able  to concisely express  the error occurring in integer state variables thanks to the \emph{select()} construct, 
	in which the user can fix the granularity adopted to approximate a dense interval. This discrete representation provides an \emph{under-approximation} of the system behaviour; thus, a finer granularity translates into an exponential increase of the complexity of the system, with obvious consequences on the verification performance. 
	Then, we tried to model our case study  in Real-Time Maude,  a completely different framework for real-time systems, based on \emph{rewriting logic}. The language supports object-like inheritance features that are quite helpful to represent complex systems in a modular manner. 
	We used communication channels to implement our attacks on the physical devices. Furthermore, we used 
	rational variables for a  more concise discrete representation of state variables.  We  have been able to verify  LTL and T-CTL properties, although the verification process resulted to be quite slow due to a proliferation of rewriting rules when fixing a reasonable granularity to approximate dense intervals. 
	As the verification logic is quite powerful, there is no need to implement an ad hoc process to check for safety. 
	Finally,  we also tried to model our case study in the \emph{safety model 
		checker} \textsf{prohver} within the {\textsc M{\scriptsize ODEST} T{\scriptsize OOLSET}} (see~\cite{ourFORTE2018}).  We specified our case study in the high-level language {\textsc {HM{\scriptsize{ODEST}}}}, supporting:
	(i) differential inclusion to model linear CPSs with constant bounded derivatives; 
	(ii) linear formulae to express nondeterministic assignments within a  dense interval; 
	(iii) a compositional programming style inherited from process algebra; 
	(iv) shared actions to synchronise parallel components. 
	However, we faced the same performance limitations encountered in \Uppaal{}. Thus, we decided to move to statistical model checking.

	Finally, this article 
	extends the preliminary conference version~\cite{CSF2017} in the following aspects:
	\begin{inparaenum}[(i)]
		\item the calculus has been slightly redesigned by distinguishing physical state and physical environment, adding specifying constructs to sniff, drop and forge packets, and removing, for simplicity,  protected physical devices; 
		\item the two trace semantics  have been proven to be compositional, i.e., preserved by properly defined contexts; 
		\item both our running example $\mathit{Sys}$ and the attacks proposed in Examples \ref{exa:att:DoS}, \ref{exa:att:DoS2},  \ref{exa:att:integrity} and \ref{exa:top-attacker} have been implemented and verified in \Uppaal{} SMC. 
	\end{inparaenum}

	\subsection{Future work} 
	While much is still to be done, we believe that our paper provides a
	stepping stone for the development of formal and automated tools to
	analyse the security of \CPS{s}. We will consider applying, possibly after proper enhancements, existing tools and frameworks for automated security protocol analysis, resorting to the development of a dedicated tool if existing ones prove not up to the task.
	We will also consider further security properties and concrete examples of \CPS{s}, as well as other kinds of physics-based attacks,such as \emph{delays in the communication} of measurements and/or commands,  and \emph{periodic attacks}, i.e., attacks that operate in a periodic fashion inducing periodic physical effects on the targeted system that may be easily confused by engineers with system malfunctions. This will allow us to refine the classes of attacks we have given here (e.g., by formalising a type system amenable to static analysis), and provide a formal definition of when a \CPS{} is more secure than another so as to be able to design, by progressive refinement, secure variants of a vulnerable \CPS{s}.
	
	We also aim to extend the behavioural theory of \cname{} by developing  suitable \emph{probabilistic metrics} to take into consideration the probability of a specific trace to actually occur. We have already done some progress in this direction for a variant of \cname{} with no security features in it, by defining ad hoc compositional \emph{bisimulation metrics}~\cite{LMT2018}. In this manner, we believe that our notion of impact might be refined by taking into account quantitative aspects of an attack such as the probability of being successful when targeting a specific \CPS{}. A first attempt on a (much) simpler IoT setting can be found in~\cite{ifm2018}.

	Finally, with respect to automatic approximations of the impact, while we have not yet fully investigated the problem, we believe that we can transform it into a ``minimum problem''. For instance, if the environment uses linear functions, then, by adapting techniques developed for linear hybrid automata (see, e.g., \cite{ALUR95}), the set of all traces with length at most $n$ (for a fixed $n$) can be characterised by a system of first degree inequalities, so the measure of the impact could be translated into a linear programming problem.
	
	\section*{Acknowledgements}
		We thank the anonymous reviewers for their insightful and careful reviews.
		Massimo Merro and Andrei Munteanu have been partially supported by the project ``Dipartimenti di Eccellenza 2018--2022'' funded by the Italian Ministry of Education, Universities and Research (MIUR).

\bibliographystyle{abbr}
\bibliography{IoT_bib}


\appendix
\section{Proofs}
\subsection{Proofs of Section~\ref{sec:calculus}}

As already stated in Remark~\ref{rem:deadlock}, our trace preorder $\sqsubseteq$ is deadlock-sensitive. Formally, 
\begin{lemma}
	\label{lem:trace-deadlock}
	Let $M$ and $N$ be two \CPS{s} in \cname{} such that $M \sqsubseteq N$. Then, $M$ satisfies its system invariant if and only if $N$ satisfies its system invariant. 
\end{lemma}
\begin{proof}
	This is because \CPS{s} that don't satisfy their invariant can only fire $\dead$ actions. 
\end{proof}

\begin{proof}[Proof of Theorem~\ref{thm:congruence-trace}]
	We prove the three statements separately.
	
	\begin{compactenum}
		\item Let us prove that $M \uplus O \trans{t} M' \uplus O' $ entails $N \uplus O \Trans{ \hat{t}} N' \uplus O' $. The proof is by induction on the length of the trace $M \uplus O \trans{t} M' \uplus O' $.
		
		As $M \sqsubseteq N $, by an application of Lemma~\ref{lem:trace-deadlock} it follows that either both $M$ and $N$ satisfy their respective invariants or they both don't. In the latter case, the result would be easy to prove as the systems can only fire $\dead$ actions.
		Similarly, if the system invariant of $O$ is not satisfied, then $M\uplus O$ and $N\uplus O$ can perform only $\dead$ actions and again the result would follow easily. 
		Thus, let us suppose that the system invariants of $M$, $N$ and $O$ are satisfied.
		
		\noindent 
		\emph{Base case.}
		We suppose  $M=\confCPS {E_1;S_1}{P_1}$, $N=\confCPS {E_2;S_2}{P_2}$,  and $O=\confCPS {E_3;S_3}{P_3}$. We proceed by case analysis on why 
		$M \uplus O \trans{\alpha} M' \uplus O' $, for some action $\alpha$. 
		\begin{itemize}
			\item $\alpha = \out c v$. Suppose $M \uplus O \trans{\out c v} M' \uplus O' $ is derived by an application of rule \rulename{Out}.
			We have two possible  cases: 
			\begin{itemize}
				\item either $P_1 \parallel P_3 \trans{\out c v} P_1 \parallel P_3'$, because $P_3 \trans{\out c v} P_3'$, for some $P_3'$, $O'=\confCPS {S_3}{P_3'}$, and $M'=M$, 
				\item or $P_1 \parallel P_3 \trans{\out c v} P'_1 \parallel P_3$, because 
				$P_1  \trans{\out c v} P'_1   $, for some $P_1'$,  and  $M=\confCPS {S_1}{P_1'}$ and  $O'=O$.
			\end{itemize}
			In the first case, by an application of rule \rulename{Par} we derive $P_2\parallel P_3 \trans{\out c v} P_2 \parallel P_3'$. Since both system invariants of $N$ and $O$ are satisfied, we can derive the required trace $N\uplus O \trans{\out c v} N \uplus O'$ by an application of rule \rulename{Out}. 
			In the second case, since $P_1 \trans{\out c v} P'_1$ and the invariant of $M$ is satisfied, by an application of rule \rulename{Out} we can derive $M \trans{\out c v} M'$.
			As $M \sqsubseteq N$, there exists a trace $N \Trans{ \widehat{\out c v}} N'$, for some system $N'$. Thus, by several applications of rule \rulename{Par} we can easily derive $N \uplus O \Trans{ \widehat{\out c v}} N' \uplus O = N' \uplus O' $, as required. 
			
			\item $\alpha = \inp c v$. Suppose $M \uplus O \trans{\inp c v} M' \uplus O'$ is derived by an application of  rule \rulename{Inp}. This case is similar to the previous one.
			
			\item $\alpha = \tau$. Suppose $M \uplus O \trans{\tau} M' \uplus O'$ is derived by an application of rule \rulename{SensRead}.
			We have two possible cases: 
			\begin{compactitem}
				\item either $P_1 \parallel P_3 \trans{\rcva s v} P_1 \parallel P_3'$ because $P_3 \trans{\rcva s v} P_3'$, for some $P_3'$, $P_1 \parallel P_3\ntrans{\snda{\mbox{\Lightning}s}{v}}$ (and hence $P_3\ntrans{\snda{\mbox{\Lightning}s}{v}}$), $M'=M$ and $O'=\confCPS {S_3}{P_3'}$, 
				\item or $P_1 \parallel P_3 \trans{\rcva s v} P'_1 \parallel P_3$ because $P_1  \trans{\rcva s v} P'_1$, for some $P_1'$, $P_1 \parallel P_3 \ntrans{\snda{\mbox{\Lightning}s}{v}}$ (and hence $P_1 \ntrans{\snda{\mbox{\Lightning}s}{v}}$) and $M'=\confCPS {S_1}{P_1'}$ and $O'=O$.
			\end{compactitem}
			In the first case, by an application of rule \rulename{Par} we derive $P_2\parallel P_3 \trans{\rcva s v} P_2 \parallel P_3'$. Moreover from $P_3\ntrans{\snda{\mbox{\Lightning}s}{v}}$ and since the sets of sensors are always disjoint, we can derive $P_2 \parallel P_3\ntrans{\snda{\mbox{\Lightning}s}{v}}$. Since both invariants of  $N$ and $O$ are satisfied,  we can derive $N\uplus O \trans{\tau} N \uplus O'$ by an application of rule \rulename{SensRead}, as required.
			In the second case, since $P_1 \trans{\rcva s v} P_1'$ and the invariant of $M$ is satisfied, by an application of  rule \rulename{SensRead} we can derive $M  \trans{\tau} M' $ with $M'= \confCPS{S_1}{P_1'}$. As $M \sqsubseteq N$, there exists a derivation $N  \Trans{ \hat{\tau}} N' $, for some $N'$. Thus, we can derive the required trace $N \uplus O \Trans{ \hat{\tau}} N' \uplus O$  by an application of rule \rulename{Par}. 
			
			\item $\alpha = \tau$. Suppose that $M \uplus O \trans{\tau} M' \uplus O'$ is derived by an application of rule \rulename{$\mbox{\Lightning}$SensSniff$\mbox{\,\Lightning}$}. This case is similar to the previous one. 
			
			\item $\alpha = \tau$. Suppose that $M \uplus O \trans{\tau} M' \uplus O'$ is derived by an application of rule \rulename{ActWrite}. This case is similar to the case \rulename{SensRead}.
			
			\item $\alpha = \tau$. Suppose that $M \uplus O \trans{\tau} M' \uplus O' $ is derived by an application of rule \rulename{$\mbox{\Lightning}$AcIntegr$\mbox{\,\Lightning}$}. This case is similar to the case \rulename{$\mbox{\Lightning}$SensSniff$\mbox{\,\Lightning}$}.
			
			\item $\alpha = \tau$. Suppose that $M \uplus O \trans{\tau} M' \uplus O'$ is derived by an application of rule \rulename{Tau}. We have four possible cases:
			\begin{compactitem}
				\item $P_1\parallel P_3 \trans{\tau} P_1' \parallel P_3'$ by an application of rule \rulename{Com}.
				We have two sub-cases: either $P_1 \trans{\out c v} { P'_1}$ and $P_3 \trans{\inp c v} { P_3'}$, or $P_1 \trans{\inp c v} { P'_1} $ and $P_3 \trans{\out c v} { P_3'}$, for some $P_1'$ and $P_3'$.
				We prove the first case, the second one is similar.
				As the invariant of $M$ is satisfied, by an application of rule \rulename{Out} we can derive $M \trans{\out c v} M' $. As $M \sqsubseteq N$, there exists a trace $N \Trans{ \hat{\tau}} \trans{\out c v} \Trans{ \hat{\tau}} N'$, for some $N'=\confCPS {E_2;S_2'}{P_2'}$. 
				As $P_3 \trans{\inp c v} P_3'$, by several applications of rule \rulename{Par} and one of rule \rulename{Com} we derive  $N \uplus O \Trans{ \widehat{\out c v}} N' \uplus O'$, as required.

				\item $P_1\parallel P_3 \trans{\tau} P_1 \parallel P_3'$  or $P_1\parallel P_3 \trans{\tau} P'_1 \parallel P_3$ by an application of  \rulename{Par}.
				This case is easy.
				
				\item $P_1\parallel P_3 \trans{\tau} P_1' \parallel P_3'$ by an application of either rule \rulename{$\mbox{\Lightning}$ActDrop$\mbox{\,\Lightning}$} or rule \rulename{$\mbox{\Lightning}$SensIntegr$\mbox{\,\Lightning}$}.
				This case does not apply as the sets of actuators of $M$ and $O$ are disjoint.
				
				\item $P_1\parallel P_3 \trans{\tau} P_1' \parallel P_3'$ by the application of on rule among \rulename{Res}, \rulename{Rec}, \rulename{Then} and \rulename{Else}.
				This case does not apply to parallel processes.
			\end{compactitem}
			
			\item $\alpha = \dead$. Suppose that $M \uplus O \trans{\dead} M' \uplus O' $ is derived by an application of rule \rulename{Deadlock}. 
			This case is not admissible as the invariants of $M$, $N$ and $O$ are satisfied.
			
			\item $\alpha = \tick$. Suppose that $M \uplus O \trans{\tick} M' \uplus O' $ is derived by an application of rule \rulename{Time}. This implies $P_1\parallel  P_3 \trans{\tick} P_1'\parallel P_3'$, for some $P_1'$ and $P_3'$, $M'=\confCPS {E_1;S_1'}{P_1'}$ and $O=\confCPS {E_3;S_3'}{P_3'}$, with $S'_1 \in \operatorname{next}(E_1;S_1)$ and $S'_3 \in \operatorname{next}(E_3;S_3)$.
			As $P_1\parallel P_3 \trans{\tick} P_1'\parallel P_3'$ can only be derived by an application of rule \rulename{TimePar}, it follows that $P_1 \trans{\tick} P_1'$ and $P_3 \trans{\tick} P_3'$. Since the invariant of $M$ is satisfied, by an application of rule \rulename{Time} we can derive $M \trans{\tick} M'$ with $M'=\confCPS{E_1;S_1'}{P_1'}$. As $M \sqsubseteq N$, there exists a derivation $N \Trans{ \hat{\tau}} N'' \trans{\tick}N''' \Trans{ \hat{\tau}} N'$, for some $N'=\confCPS {E_2;S_2'}{P_2'}$, $N''=\confCPS {E_2;S_2''}{P_2''}$, $N'''=\confCPS {E_2;S_2'''}{P_2'''}$, with $S_2''' \in \operatorname{next}(E_2;S_2'')$.
			By several applications of rule \rulename{Par} we can derive that $N \uplus O \Trans{ \hat{\tau}} N'' \uplus O$ and $N''' \uplus O' \Trans{ \hat{\tau} }N' \uplus O'$. In order to conclude the proof, it is sufficient to prove $N'' \uplus O \trans{\tick} N''' \uplus O' $.
			By the definition of rule \rulename{Time}, from  $N'' \trans{\tick}N'''$ it follows that $P_2'' \trans{\tick} P_2'''$.  As $P_3 \trans{\tick}   P_3'$, by an application of rule \rulename{TimePar} it follows that $P_2'' \parallel  P_3  \trans{\tick} P_2''' \parallel P_3'$. Since $S_2''' \in \operatorname{next}(E_2;S_2'')$ and $S'_3 \in \operatorname{next}(E_3;S_3)$ we can derive that $S_2''' \uplus S_3' \in \operatorname{next}(E_2;S_2'')\cup \operatorname{next}(E_3;S_3)$. By an application of rule \rulename{Time} we have $N'' \uplus O \trans{\tick} N''' \uplus O' $  and hence $N \uplus O \Trans{\widehat{\tick}} N' \uplus O' $, as required. 
			
			\item $\alpha= \unsafe$. Suppose that $M \uplus O \trans{\unsafe} M' \uplus O'$ is derived by an application of rule \rulename{Safety}. This is similar to the case $\alpha = \out c v$ by considering the fact that $\statefun{} \not \in  \safefun{}$ implies that $\statefun{}  \cup \statefun{}  '  \not \in \safefun{} \cup \safefun{}'$, for any $\statefun{}  '$ and any $ \safefun{}'$.
		\end{itemize}
		\noindent 
		\emph{Inductive case.}
		We have to prove that $M\uplus O=M_0\uplus O_0 \trans{\alpha_1}\dots \trans{\alpha_n} M_{n}\uplus O_{n}$ implies 
		$N\uplus O=N_0\uplus O_0 \Trans{ \widehat{\alpha_1}}\dots \Trans{ \widehat{\alpha_n}} N_{n}\uplus O_{n}$. We can use the inductive hypothesis to easily deal with the first $n-1$ actions and resort to the base case to handle the $n^\mathrm{th}$ action.

		\item We have to prove that $M \sqsubseteq N$ implies $M \parallel P \sqsubseteq N \parallel P$, for any pure-logical process $P$. This is a special case of (1) as $M \parallel P=M\uplus \left(\confCPS{\emptyset;\emptyset}{P}\right)$ and $N \parallel P=N\uplus \left(\confCPS{\emptyset;\emptyset}{P}\right)$, where $\confCPS{\emptyset;\emptyset}{P}$ is a \CPS{} with no physical process in it, only logics. 
		
		\item We have to prove that $M \sqsubseteq N$ implies $M \backslash c \; \sqsubseteq \; N \backslash c$, for any channel $c$. For any derivation $M \backslash c \trans{t} M' \backslash c $ we can easily derive that $M \trans{t} M'$ with $c$ not occurring in $t$. Since $M \sqsubseteq N$, it follows that $N \Trans{ \hat{t}} N' $, for some $N'$. Since $c$ does not appear in $t$, we can easily derive that $N\backslash c \Trans{ \hat{t} } N' \backslash c $, as required. 
	\end{compactenum}
\end{proof}

In order to prove Theorem~\ref{thm:congruence-traceupto} we adapt to \cname{} two standard lemmata used in process calculi theory to compose and  decompose the actions performed by a compound system. 

\begin{lemma}[Decomposing system actions]
	\label{lem:cong}
	Let $M$ and $N$ be two \CPS{}s in \cname{}. Then, 
	\begin{compactitem}
		\item if $M \uplus N\trans{\tick} M'\uplus N'$, for some $M'$ and $N'$, then $M \trans{\tick} M'$ and  $N \trans{\tick} N'$;
		\item if $M \uplus N\trans{\dead} M \uplus N$, then  $M \trans{\dead} M$ or $N \trans{\dead} N$;
		\item if $M \uplus N\trans{\tau} M'\uplus N'$, for some  $M'$ and $N'$, due to a channel synchronisation between $M$ and $N$,
		then either $ M \trans{\out c v}  { M'} $ and  $N \trans{\inp c v}  { N'} $, or $ M \trans{\inp c v}  {M'} $ and $N \trans{\out c v}  {N'} $, for some \nolinebreak channel \nolinebreak $c$; 
		\item if $M \uplus N\trans{\alpha} M'\uplus N'$, for some $M'$ and $N'$,
		$\alpha \neq \tick$, not due to a channel synchronisation between $M$ and $N$, then either $M \trans{\alpha} M$ and $N=N'$, or $N \trans{\alpha} N$ and $M=M'$.
	\end{compactitem}
\end{lemma}

\begin{lemma}[Composing system actions]
	\label{lem:cong2}
	Let $M$ and $N$ be two \CPS{s} of \cname{}. Then, 
	\begin{compactitem}
		\item If $M \trans{\tick} M'$ and $N \trans{\tick} N'$, for some $M'$ and $N'$, then $M \uplus N\trans{\tick} M'\uplus N'$; 
		\item If $N \ntrans{\dead}$ and $M \trans{\alpha} M'$, for some $M'$ and $\alpha \neq \tick $, then $M \uplus N\trans{\alpha} M'\uplus N$ and $N \uplus M\trans{\alpha} N \uplus M'$.
	\end{compactitem}
\end{lemma} 

\begin{proof}[Proof of Theorem~\ref{thm:congruence-traceupto}]
	Here, we prove case (1) of the theorem. The proofs of cases (2) and (3) are similar to the corresponding ones of Theorem \ref{thm:congruence-trace}.
	
	We prove that $M \sqsubseteq_{m..n} N$ implies that there are $m',n' \in \mathbb{N}^+ \cup \infty $, with $ m'..n' \subseteq m..n $ such that $M \uplus O \sqsubseteq_{m'..{n'}} N \uplus O$.  
	We prove separately that $m' \geq m$ and $n' \leq n$. 
	
	\begin{itemize}
		\item $m'\geq m$.
		We recall that $m,m' \in \mathbb{N}^+ $.
		If $m=1$, then we trivially have $m'\geq 1 =m$.
		Otherwise, since $m$ is the minimum integer for which there is a trace $t$, with $\#\tick(t)=m-1$, such that $M \trans{t}$ and $N \not\!\!\!\Trans{\hat{t}}$, then for any trace $t$, with $\#\tick(t)<m-1$ and such that $M \trans{t}$, it holds that $N \Trans{\hat{t}}$. As done in the proof of case (1) of Theorem~\ref{thm:congruence-trace}, we can derive  that for any trace $t$, with $\#\tick(t)<m-1$ and such that $M \uplus O \trans{t}$ it holds that $N \uplus O \Trans{\hat{t}}$. This implies the required condition, $m' \geq m$.
		
		\item $n' \leq n$.
		We recall that $n$ is the infimum element of $\mathbb{N}^+ \cup \{ \infty \}$, $n \geq m$, such that whenever $M \trans{t_1}M'$, with $\#\tick(t_1)=n-1$, there is $t_2$, with $\#\tick(t_1)=\#\tick(t_2)$, such that $N \trans{t_2}N'$, for some $N'$, and $M' \sqsubseteq N'$.
		Now, if $M\uplus O \trans{t} M'\uplus O'$, with $\#\tick(t)=n-1$, by Lemma~\ref{lem:cong} we can split the trace $t$ by extracting the actions performed by $M$ and those performed by $O$. Thus, there exist two traces $M \trans{t_1} M'$ and $O \trans{t_3} O'$, with $\#\tick(t_1)=\#\tick(t_3)=n-1$ whose combination has generated the trace $M\uplus O \trans{t} M'\uplus O'$.
		As $M \sqsubseteq_{m..n} N$, from $M \trans{t_1} M'$ we know that 
		there is a trace $t_2$, with $\#\tick(t_1)=\#\tick(t_2)$, such that $N \trans{t_2}N'$, for some $N'$, and $M' \sqsubseteq N'$. Since $N \trans{t_2} N'$ and $O \trans{t_3} O'$, by an application of Lemma~\ref{lem:cong2} we can build a trace $N \uplus O\trans{t'}N' \uplus O'$, for some $t'$ such that $\#\tick(t)=\#\tick(t')= n-1 $. As $M' \sqsubseteq N'$, by Theorem~\ref{thm:congruence-trace} we can derive that $M' \uplus O'\sqsubseteq N' \uplus O'$. This implies that $n'\leq n$. 
	\end{itemize}
\end{proof}

\subsection{Proofs of Section~\ref{sec:running_example}}

In order to prove Proposition~\ref{prop:sys} and Proposition~\ref{prop:X}, we use the following lemma that formalises the invariant properties binding 
the state variable $\mathit{temp}$ with the activity of the cooling system.

Intuitively, when the cooling system is inactive the value of the state
variable $\mathit{temp}$ lays in the real interval $[0, 11.5]$.
Furthermore, if the coolant is not active and the variable $\mathit{temp}$
lays in the real interval $(10.1, 11.5]$, then the cooling will be turned
on in the next time slot. Finally, when active the cooling system will
remain so for $k\in1..5$ time slots (counting also the current time slot)
with the variable $\mathit{temp}$ being in the real interval $(
9.9-k{*}(1{+}\delta) , 11.5-k{*}(1{-}\delta)]$.

\begin{lemma} 
	\label{lem:sys}
	Let $\mathit{Sys}$ be the system defined in Section~\ref{sec:running_example}.
	Let
	\begin{small}
		\begin{displaymath}
		\mathit{Sys} = \mathit{Sys_1} \trans{t_1}\trans\tick 
		\mathit{Sys_2}\trans{t_2}\trans\tick  \dots 
		\trans{t_{n-1}}\trans\tick  \mathit{Sys_n}
		\end{displaymath}
	\end{small}%
	such that the traces $t_j$ contain no $\tick$-actions, for any $j \in  1 .. n{-}1 $, and for any  $i \in  1 .. n $, $\mathit{Sys_i}= \confCPS {S_i}{P_i} $ with 
	$S_i = \langle \statefun^i{} ,  \sensorfun^i{} , \actuatorfun^i{} \rangle$.
	Then, for any $i \in 1 .. n{-}1 $, we have the following:
	\begin{enumerate}
		
		\item \label{uno}
		if   $ \actuatorfun^i{}(\mathit{cool})= \off $ then
		$\statefun^i{}(\mathit{temp})  \in [0, 11.1+\delta ]$;
		with  $\statefun^i{}(\mathit{stress})=0$ if $ \statefun^i{}(\mathit{temp})  \in [0, 10.9+\delta ] $, and  $\statefun^i{}(\mathit{stress})=1 $, otherwise; 
		
		\item \label{due}
		if   $ \actuatorfun^i{}(\mathit{cool})= \off $ and 
		$\statefun^i{}(\mathit{temp})\in (10.1, 11.1+\delta ]$ then, in the next time slot,  $\actuatorfun^{i{+}1}{}(\mathit{cool})=\on$
		and $\statefun^{i{+}1}{}(\mathit{stress})  \in 1..2$;

		\item \label{tre}
		if  $ \actuatorfun^i{}(\mathit{cool})=\on$ then   $\statefun^i{}(\mathit{temp}) \in ( 9.9-k {*}(1{+}\delta) , 11.1+\delta -k{*}(1{-}\delta)] $, 
		for some  $k  \in 1 .. 5 $   such that $\actuatorfun^{i-k}{}(\mathit{cool})=\off $ and 
		$\actuatorfun^{i-j}{}(\mathit{cool}) =\on $, for $j \in 0..k{-}1$;
		moreover,  if  $k\in 1.. 3$ then     $\statefun^i{}(\mathit{stress})  \in  1..k{+}1  $, otherwise,  
		$\statefun^i{}(\mathit{stress}) =0$. 
	\end{enumerate}
\end{lemma}
\begin{proof}
	Let us write $v_i$ and $s_i$ to denote the values of the state variables
	$\mathit{temp}$ and $\mathit{stress}$, respectively, in the systems
	$\mathit{Sys_i}$, i.e., $\statefun^i{} (\mathit{temp})=v_i $ and
	$\statefun^i{} (\mathit{stress})=s_i $. Moreover, we will say that the
	coolant is active (resp., is not active) in $\mathit{Sys_i}$ if
	$\actuatorfun^i{}(\mathit{cool})=\on$ (resp.,
	$\actuatorfun^i{}(\mathit{cool})=\off$).
	
	The proof is by mathematical induction on $n$, i.e., the number of
	$\tick$-actions of our traces.
	
	The \emph{case base} $n=1$ follows directly from the definition of $\mathit{Sys}$. 
	
	Let us prove the \emph{inductive case}. 
	We assume that the three statements hold for $n-1$ and prove that they  
	also hold for $n$.
	\begin{enumerate}[noitemsep]
		\item Let us assume that the cooling  is not active  in $\mathit{Sys_{n}}$.
		In this case, we prove that $v_n \in [0, 11.1+\delta ]$, with and $s_n=0$ if $ v_n  \in [0, 10.9+\delta ] $, and $s_n=1$ otherwise.

		We consider separately the cases in which the coolant is active or not in $\mathit{Sys_{n-1}}$
		\begin{itemize}[noitemsep]
			\item Suppose the coolant is not active in $\mathit{Sys_{n{-}1}}$ (and
			not active in $\mathit{Sys_{n}}$).
			
			By the induction hypothesis we have 
			$v_{n-1} \in [0, 11.1+\delta ]$; with $s_{n{-}1}=0$ if $ v_{n{-}1}  \in [0, 10.9+\delta ] $, and $s_{n{-}1}=1 $ otherwise. Furthermore,  if   
			$v_{n-1} \in (10.1, 11.1+\delta ]$, then, by the induction hypothesis, the coolant must be active in $\mathit{Sys_{n}}$.
			Since we know that in $\mathit{Sys_n}$ the cooling is not active,
			it follows that $v_{n-1} \in [0, 10.1]$ and $s_n =0$.
			Furthermore, in $\mathit{Sys_{n}}$ the temperature
			will increase of a value laying in the real interval $[1-\delta,1+\delta]=[0.6,1.4]$. Thus, $v_{n}$ will be in 
			$ [0.6, 11.1+\delta ]\subseteq[0, 11.1+\delta ]$.  
			Moreover, if $v_{n-1} \in [0, 9.9]$, then the state variable $\mathit{stress}$ is not incremented and hence $s_n=0$ with   
			$ v_n  \in [0+1-\delta \, , \,  9.9+ 1+\delta ]=[0.6 \, , \, 10.9+\delta]\subseteq [0 \, , \, 10.9+\delta ] $. Otherwise, 
			if $v_{n-1} \in (9.9,10.1]$, then the state variable $\mathit{stress}$ is incremented, and hence $s_n=1$.
			
			\item Suppose the coolant is active in $\mathit{Sys_{n{-}1}}$ (and not
			active in $\mathit{Sys_{n}}$).
			
			By the induction hypothesis, $v_{n-1} \in ( 9.9-k *(1+\delta) , 11.1+\delta
			-k*(1-\delta)]$ for some $k \in 1..5$ such that the coolant is not active
			in $\mathit{Sys_{n{-}1{-}k}}$ and is active in $\mathit{Sys_{n{-}k}},
			\ldots, \mathit{Sys_{n-1}}$.
			
			The case $k \in \{1,\ldots,4\}$  is not admissible. 
			In fact if  $k \in \{1,\ldots,4\}$ then the coolant would be active for less than $5$ $\tick$-actions as we know that 
			$\mathit{Sys_{n}}$ is not active. 
			Hence, it must be $k=5$. Since $\delta=0.4$ and $k=5$, it holds that $v_{n-1 }\in (9.9-5*1.4, 11.1+0.4 -5*0.6]=(2.8, 8.6] $
			and $s_{n{-}1} =0$. Moreover, since
			the coolant is active for $5$ time slots, in $\mathit{Sys_{n{-}1}}$ the controller and the $\mathit{IDS}$ synchronise together via channel $\mathit{sync}$ and hence the $\mathit{IDS}$ checks the
			temperature. Since $v_{n-1} \in (2.8, 8.6]$ the $\mathit{IDS}$ process 
			sends to the controller a command to 
			$\mathsf{stop}$ the cooling, and the controller will switch off the 
			cooling system. Thus, in the next time slot, the temperature
			will increase of a value laying in the real interval $[1-\delta,1+\delta]=[0.6,1.4]$. As
			a consequence, in $\mathit{Sys_{n}}$ we will have $v_{n}
			\in [2.8+0,6, 8.6+1.4]=[3.4,10] \subseteq [0, 11.1+\delta ]$.
			Moreover, since $v_{n-1} \in (2.8, 8.6]$ and $s_{n{-}1} =0$, we derive that the state variable $\mathit{stress}$ is not increased and hence
			$s_{n } =0$, with $v_{n} \in [3.4,10] \subseteq [0, 10.9+\delta ]$. 
		\end{itemize}
		
		\item Let us assume that the coolant is not active in $\mathit{Sys_{n}}$
		and $v_n \in (10.1, 11.1+\delta ]$; we prove that the coolant is active in
		$\mathit{Sys_{n{+}1}}$ with $s_{n {+}1} \in 1..2 $. Since the coolant is
		not active in $\mathit{Sys_{n}}$, then it will check the temperature
		before the next time slot. Since $v_n \in (10.1, 11.1+\delta ]$ and
		$\epsilon=0.1$, then the process $\mathit{Ctrl}$ will sense a temperature
		greater than $10$ and the coolant will be turned on. Thus, the coolant
		will be active in $\mathit{Sys_{n{+}1}}$. Moreover, since $v_n \in (10.1,
		11.1+\delta ]$, and $s_{n}$ could be either $0$ or $1$, the state variable
		$\mathit{stress}$ is increased and therefore $s_{n {+}1} \in 1..2$.
		
		\item Let us assume that the coolant is active in $\mathit{Sys_{n}}$; we
		prove that $v_{n} \in ( 9.9-k *(1+\delta), 11.1+\delta -k*(1-\delta)] $
		for some $k \in 1..5 $ and the coolant is not active in
		$\mathit{Sys_{n{-}k}}$ and active in $\mathit{Sys_{n-k+1}}, \dots,
		\mathit{Sys_{n}}$. Moreover, we have to prove that if $k\leq 3$ then $s_n
		\in 1..k{+}1 $, otherwise, if $k > 3$ then $s_n =0$.
		
		We prove the first statement. That is, we prove that $v_{n} \in ( 9.9-k
		*(1+\delta), 11.1+\delta -k*(1-\delta)] $, for some $k \in 1..5 $, and the
		coolant is not active in $\mathit{Sys_{n{-}k}}$, whereas it is active in
		the systems $\mathit{Sys_{n-k+1}}, \dots, \mathit{Sys_{n}}$.

		
		We separate the case in which the coolant is active in
		$\mathit{Sys_{n{-}1}}$ from that in which is not active.
		
		\begin{itemize}[noitemsep]
			\item Suppose the coolant is not active in $\mathit{Sys_{n{-}1}}$ (and active in $\mathit{Sys_{n}}$).
			
			In this case $k=1$ as the coolant is not active in $\mathit{Sys_{n-1}}$
			and it is active in $\mathit{Sys_{n}}$. Since $k=1$, we have to prove $v_n
			\in (9.9-(1+\delta), 11.1+\delta-(1-\delta)]$.
			
			However, since the coolant is not active in $\mathit{Sys_{n-1}}$ and is
			active in $\mathit{Sys_{n}}$ it means that the coolant has been switched
			on in $\mathit{Sys_{n-1}}$ because the sensed temperature was above $10$
			(since $\epsilon=0.1$ this may happen only if $v_{n-1} > 9.9$). By the
			induction hypothesis, since the coolant is not active in
			$\mathit{Sys_{n-1}}$, we have that $v_{n-1} \in [0, 11.1+\delta ]$.
			Therefore, from $v_{n-1} > 9.9$ and $v_{n-1} \in [0, 11.1+\delta ]$ it
			follows that $v_{n-1} \in (9.9, 11.1+\delta ]$. Furthermore, since the
			coolant is active in $\mathit{Sys_{n}}$, the temperature will decrease of
			a value in $[1-\delta,1+\delta]$ and therefore $v_n \in (9.9-(1+\delta),
			11.1+\delta-(1-\delta)]$, which concludes this case of the proof.
			
			\item Suppose the coolant is active in $\mathit{Sys_{n{-}1}}$ (and active
			in $\mathit{Sys_{n}}$ as well).
			
			By the induction hypothesis, there is $h \in 1..5$ such that $v_{n-1} \in
			( 9.9-h *(1+\delta) , 11.1+\delta -h*(1-\delta)] $ and the coolant is not
			active in $\mathit{Sys_{n{-}1{-}h}}$ and is active in
			$\mathit{Sys_{n{-}h}}, \ldots, \mathit{Sys_{n{-}1}}$.
			
			The case $h=5$ is not admissible. In fact, since $\delta=0.4$, if $h=5$
			then $v_{n-1 }\in (9.9-5*1.4, 11.1+\delta -5*0.6]=(2.8, 8.6] $.
			Furthermore, since the cooling system has been active for $5$ time
			instants, in $\mathit{Sys_{n{-}1}}$ the controller and the IDS synchronise
			together via channel $\mathit{sync}$, and the $\mathit{IDS}$ checks the
			received temperature. As $v_{n-1 }\in (2.8, 8.6] $, the $\mathit{IDS}$
			sends to the controller via channel $\mathit{ins}$ the command
			$\mathsf{stop}$. This implies that the controller should turn off the
			cooling system, in contradiction with the hypothesis that the coolant is
			active in $\mathit{Sys_{n }}$.
			
			Hence, it must be $h \in 1 .. 4$. Let us prove that for $k=h+1$ we obtain
			our result. Namely, we have to prove that, for $k=h+1$, (i) $v_{n} \in (
			9.9-k *(1+\delta), 11.1+\delta -k*(1-\delta)] $, and (ii) the coolant is
			not active in $\mathit{Sys_{n{-}k}}$ and active in $\mathit{Sys_{n-k+1}},
			\dots, \mathit{Sys_{n}}$.
			
			Let us prove the statement (i). By the induction hypothesis, it holds that
			$v_{n-1} \in ( 9.9-h *(1+\delta) , 11.1+\delta -h*(1-\delta)] $. Since the
			coolant is active in $\mathit{Sys_{n}}$, the temperature will decrease
			Hence, $v_{n } \in ( 9.9-(h+1) *(1+\delta) , 11.1+\delta
			-(h+1)*(1-\delta)] $. Therefore, since $k=h+1$, we have that $v_{n} \in (
			9.9-k *(1+\delta) , 11.1+\delta -k*(1-\delta)] $.
			
			Let us prove the statement (ii). By the induction hypothesis the coolant
			is not active in $\mathit{Sys_{n-1-h}}$ and it is active in
			$\mathit{Sys_{n-h}}, \ldots, \mathit{Sys_{n-1}}$. Now, since the coolant
			is active in $\mathit{Sys_{n}}$, for $k=h+1$, we have that the coolant is
			not active in $\mathit{Sys_{n-k}}$ and is active in $\mathit{Sys_{n-k+1}},
			\ldots, \mathit{Sys_{n}}$, which concludes this case of the proof.
		\end{itemize}
		
		Thus, we have proved that $v_{n} \in ( 9.9-k *(1+\delta), 11.1+\delta
		-k*(1-\delta)] $, for some $k \in 1..5 $; moreover, the coolant is not
		active in $\mathit{Sys_{n{-}k}}$ and active in the systems
		$\mathit{Sys_{n-k+1}}, \dots, \mathit{Sys_{n}}$.
		
		It remains to prove that $s_n \in 1..k{+}1 $ if $k\leq 3$, and $s_n =0$,
		otherwise.
		
		By inductive hypothesis, since the coolant is not active in
		$\mathit{Sys_{n{-}k}}$, we have that $s_{n{-}k} \in 0..1$. Now, for $k \in
		[1..2]$, the temperature could be greater than $9.9$. Hence if the state
		variable $\mathit{stress}$ is either increased or reset, then $s_n \in
		1..k{+}1 $, for $k\in 1.. 3$. Moreover, since for $k\in 3..5 $ the
		temperature is below $9.9$, it follows that $s_n =0$ for $k> 3$.
	\end{enumerate}
\end{proof}

\begin{proof}[Proof of Proposition~\ref{prop:sys}]
	Since $\delta=0.4$, by Lemma~\ref{lem:sys} the value of the state variable
	$\mathit{temp}$ is always in the real interval $[0, 11.5]$. As a
	consequence, the invariant of the system is never violated and the system
	never deadlocks. Moreover, after $5$ time units of cooling, the state
	variable $\mathit{temp}$ is always in the real interval $( 9.9-5 *1.4 ,
	11.1+0.4-5*0.6]=(2.9, 8.5]$. Hence, the process $\mathit{IDS}$ will never
	transmit on the channel $\mathit{alarm}$.
	
	Finally, by Lemma~\ref{lem:sys} the maximum value reached by the state variable $\mathit{stress}$ is $4$ and therefore the system does not reach unsafe states.
\end{proof}

\begin{proof}[Proof of  Proposition~\ref{prop:X}]
	Let us prove the two statements separately. 
	\begin{itemize}
		\item Since $\epsilon=0.1$, if process $\mathit{Ctrl}$ senses a
		temperature above $10$ (and hence $\mathit{Sys}$ turns on the cooling)
		then the value of the state variable $\mathit{temp}$ is greater than
		$9.9$. By Lemma~\ref{lem:sys}, the value of the state variable
		$\mathit{temp}$ is always less than or equal to $11.1+\delta $. Therefore, if
		$\mathit{Ctrl}$ senses a temperature above $10$, then the value of the
		state variable $\mathit{temp}$ is in $(9.9,11.1+\delta ]$.
		
		\item By Lemma~\ref{lem:sys} (third item), the coolant can be active for no
		more than $5$ time slots. Hence, by Lemma~\ref{lem:sys}, when $\mathit{Sys}$
		turns off the cooling system the state variable $\mathit{temp}$ ranges
		over $( 9.9-5 *(1+\delta) , 11.1+\delta-5*(1-\delta)]$. 
	\end{itemize}
\end{proof}

\subsection{Proofs of Section~\ref{sec:cyber-physical-attackers}}

\begin{proof}[Proof of Proposition~\ref{prop:att:DoS}]
	We distinguish the two cases, depending on $m$.
	\begin{itemize}[noitemsep]
		\item Let $m\leq 8$.
		We recall that the cooling system is activated only when the sensed temperature is above $10$. Since $\epsilon = 0.1$, when this happens the state variable $\mathit{temp}$ must be at least $9.9$. Note that after $m{-}1 \leq 7$ $\tick$-actions, when the attack tries to interact with the
		controller of the actuator $\mathit{cool}$, the variable $\mathit{temp}$ may reach at most $7* (1 + \delta)= 7 * 1.4=9.8$ degrees. Thus, the cooling system will not be activated and the attack will not have any effect.
		
		\item Let $m>8$. 
		By Proposition~\ref{prop:sys}, the system $\mathit{Sys}$ in
		isolation may never deadlock, it does not get into an unsafe state, and it may never emit an output on channel $\mathit{alarm}$. Thus, any execution trace of the system $\mathit{Sys}$ consists of a sequence of $\tau$-actions and $\tick$-actions.

		In order to prove the statement it is enough to show the following four facts:
		\begin{itemize}[noitemsep]
			\item the system $\mathit{Sys} \parallel A_m$ may not deadlock in the first
			$m+3$ time slots; 
			
			\item the system $\mathit{Sys} \parallel A_m$ may not emit any output in the first $m+3$ time slots; 
			
			\item the system $\mathit{Sys} \parallel A_m$ may not enter in an unsafe
			state in the first $m+3$ time slots;
			
			\item the system $\mathit{Sys} \parallel A_m$ has a trace reaching un unsafe state from the $(m{+}4)$-th time slot on, and until the invariant gets violated and the system deadlocks. 
		\end{itemize}
		

		The first three facts are easy to show as the attack may steal the command
		addressed to the actuator $\mathit{cool}$ only in the $m$-th time slot.
		Thus, until time slot $m$, the whole system behaves correctly. In
		particular, by Proposition~\ref{prop:sys} and Proposition~\ref{prop:X}, no alarms,
		deadlocks or violations of safety conditions occur, and the temperature
		lies in the expected ranges. Any of those three actions requires at least
		further $4$ time slots to occur. Indeed, by Lemma~\ref{lem:sys}, when the
		cooling is switched on in the time slot $m$, the variable
		$\mathit{stress}$ might be equal to $2$ and hence the system might not
		enters in an unsafe state in the first $m+3$ time slots. Moreover, an
		alarm or a deadlock needs more than $3$ time slots and hence no alarm can
		occur in the first $m+3$ time slots.
		
		Let us show the fourth fact, i.e., that there is a trace where the system
		$\mathit{Sys} \parallel A_m$ enters into an unsafe state starting from the
		$(m{+}4)$-th time slot and until the invariant gets violated.
		
		Firstly, we prove that for all time slots $n$, with $9\leq n < m$, there
		is a trace of the system $\mathit{Sys} \parallel A_m$ in which the state
		variable $\mathit{temp}$ reaches the values $10.1$ in the time slot $n$.
		
		The fastest trace reaching the temperature of $10.1$ degrees requires
		$\lceil \frac{10.1}{1 + \delta }\rceil = \lceil \frac{10.1}{1.4 }\rceil
		=8$ time units, whereas the slowest one $\lceil \frac{10.1}{1 - \delta
		}\rceil = \lceil \frac{10.1}{0.6 }\rceil =17$ time units. Thus, for any
		time slot $n$, with $9 \leq n \leq 18$, there is a trace of the system
		where the value of the state variable $\mathit{temp}$ is $10.1$. Now, for
		any of those time slots $n$ there is a trace in which the state variable
		$\mathit{temp}$ is equal to $10.1$ in all time slots $n+10i < m$, with
		$i\in \mathbb{N}$. Indeed, when the variable $\mathit{temp}$ is equal to
		$10.1$ the cooling might be activated. Thus, there is a trace in which the
		cooling system is activated. We can always assume that during the cooling
		the temperature decreases of $1+\delta$ degrees per time unit, reaching at
		the end of the cooling cycle the value of $5$. This entails that the trace
		may continue with $5$ time slots in which the variable $\mathit{temp}$ is
		increased of $1+\delta$ degrees per time unit; reaching again the value
		$10.1$. Thus, for all time slots $n$, with $9 \leq n < m$, there is a
		trace of the system $\mathit{Sys} \parallel A_m$ in which the state
		variable $\mathit{temp}$ is $10.1$ in $n$.
		
		As a consequence, we can suppose that in the $m{-}1$-th time slot there is
		a trace in which the value of the variable $\mathit{temp}$ is $10.1$.
		Since $\epsilon=0.1$, the sensed temperature lays in the real interval
		$[10,10.2]$. Let us focus on the trace in which the sensed temperature is
		$10$ and the cooling system is not activated. In this case, in the $m$-th
		time slot the system may reach a temperature of $10.1 + (1 + \delta)=11.5$
		degrees and the variable $\mathit{stress}$ is $1$.
		
		The process $\mathit{Ctrl}$ will sense a temperature above $10$ sending
		the command $\snda {cool} {\on}$ to the actuator $\mathit{cool}$. Now,
		since the attack $A_m$ is active in that time slot ($m > 8$), the command
		will be stolen by the attack and it will never reach the actuator. Without
		that dose of coolant, the temperature of the system will continue to grow.
		As a consequence, after further $4$ time units of cooling, i.e.\ in the
		$m{+}4$-th time slot, the value of the state variable $\mathit{stress}$
		may be $5$ and the system enters in an unsafe state.
		
		After $1$ time slots, in the time slot $m+5$, the controller and the
		$\mathit{IDS}$ synchronise via channel $\mathit{sync}$, the $\mathit{IDS}$
		will detect a temperature above $10$, and it will fire the output on
		channel $\mathit{alarm}$ saying to process $\mathit{Ctrl}$ to keep
		cooling. But $\mathit{Ctrl}$ will not send again the command $\snda {cool}
		{\on}$. Hence, the temperature would continue to increase and the system
		remains in an unsafe state while the process $\mathit{IDS}$ will keep
		sending of $\mathit{alarm}$(s) until the invariant of the environment gets
		violated.
	\end{itemize}
\end{proof}

\begin{proof}[Proof of Proposition~\ref{prop:att:dos-integrity}] 
	As proved for Proposition~\ref{prop:att:DoS}, 
	we can prove that  there
	is a trace of the system $\mathit{Sys} \parallel A_m$ in which the state
	variable $\mathit{temp}$ reaches a value greater than $9.9$ in the time slot $m$, for all    $ m >8$.
	The process $\mathit{Ctrl}$ never activates the $\mathit{Cooling}$ component  as it will always detect a temperature below $10$.
	Hence, after further $5$ $\tick$-actions
	(in the $m+5$-th time slot)
	the system will violate the safety conditions emitting an $\unsafe$ 
	action. So, since the process $\mathit{Ctrl}$ never activates the $\mathit{Cooling}$,
	the temperature will increase its value until the invariant will be violated and the
	system will deadlock.  
	Thus, $\mathit{Sys} \parallel A_m \; \sqsubseteq _{m+5..\infty} \; \mathit{Sys}$. 
\end{proof}

In order to prove Proposition~\ref{prop:att:integrity}, we introduce Lemma~\ref{lem:sys2}. This is a variant of Lemma~\ref{lem:sys} in which the \CPS{} $\mathit{Sys} $ runs in parallel with the attack $A_n$ defined in Example~\ref{exa:att:integrity}. Here, due to the presence of the attack, the temperature is $4$ degrees higher when compared to the system $\mathit{Sys}$ in isolation. 
\begin{lemma} 
	\label{lem:sys2}
	Let $\mathit{Sys}$ be the system defined in Section~\ref{sec:running_example} and $A_n$ be the attack of Example~\ref{exa:att:integrity}. 
	Let
	\begin{small}
		\begin{displaymath}
		\mathit{Sys} \parallel A_n= \mathit{Sys_1}  \trans{t_1}\trans\tick \dots \mathit{Sys_{n-1}}  \trans{t_{n-1}}\trans\tick \mathit{Sys_n} 
		\end{displaymath}
	\end{small}%
	such that the traces $t_j$ contain no $\tick$-actions, for any $j \in  1 .. n{-}1 $,  and for any  $i \in  1 .. n $ $\mathit{Sys_i}= \confCPS {S_i}{P_i} $ with 
	$S_i = \langle \statefun^i{} ,  \sensorfun^i{} , \actuatorfun^i{} \rangle$.
	Then, for any $i \in 1 .. n{-}1 $ we have the following:
	\begin{itemize}[noitemsep]
		\item  if   $ \actuatorfun^i{}(\mathit{cool})= \off $, then
		$\statefun^i{}(\mathit{temp})\in [0, 11.1+4+\delta ]$; 
		
		\item if $ \actuatorfun^i{}(\mathit{cool})= \off $ and
		$\statefun^i{}(\mathit{temp})\in (10.1+4, 11.1+4+\delta ]$, then we have $
		\actuatorfun^{i+1}{}(\mathit{cool}) =\on$;

		\item  if  $ \actuatorfun^i{}(\mathit{cool})=\on$, then   $\statefun^i{}(\mathit{temp}) \in ( 9.9+4-k *(1+\delta) , 11.1+4+\delta -k*(1-\delta)] $, 
		for some  $k  \in 1..5$,   such that $\actuatorfun^{i-k}{}(\mathit{cool}) =\off $ and $ \actuatorfun^{i-j}{}(\mathit{cool}) =\on $, for $j \in 0..k{-}1$. 
		
	\end{itemize}
\end{lemma}
\begin{proof}
	Similar to the proof of    Lemma~\ref{lem:sys}.
\end{proof}

Now, everything is in place to prove Proposition~\ref{prop:att:integrity}. 
\begin{proof}[Proof of Proposition~\ref{prop:att:integrity}]
	Let us proced by case analysis. 
	\begin{itemize}[noitemsep]
		\item 
		Let $0 \leq n \leq 8$.  
		In the proof of Proposition~\ref{prop:att:DoS}, we remarked that the system $\mathit{Sys}$ in isolation may sense a temperature greater than $10$  
		only after $8$ $\tick$-actions, i.e., in the $9$-th time slot.
		However, the life of the attack is $n \leq 8$, and in the $9$-th time
		slot the attack is already terminated. As a consequence, starting
		from the $9$-th time slot the system will correctly sense the
		temperature and it will correctly activate the cooling system.
		\item Let $n=9$. 
		The maximum value that may be reached by the state variable
		$\mathit{temp}$ after $8$ $\tick$-actions, i.e., in the $9$-th time
		slot, is $8 * (1+ \delta)=8*1.4= 11.2$. However, since in the $9$-th
		time slot the attack is still alive, the process $\mathit{Ctrl}$
		will sense a temperature below $10$ and the system will move to the next
		time slot and the state variable   $\mathit{stress}$ is incremented. Then, in the $10$-th time slot, when the attack is
		already terminated, the maximum temperature the system may reach is $11.2
		+ (1+ \delta)=12.6$ degrees and the state variable $\mathit{stress}$ is equal to  $1$. Thus, the process $\mathit{Ctrl}$ will sense
		a temperature greater than $10$, activating the cooling system  and incrementing the state variable $\mathit{stress}$. 
		As a consequence, during the following $4$ time units of cooling, the value of the state variable $\mathit{temp}$ will be at most
		$12.6 - 4*(1-\delta)= 12.6- 4*0.6=10.2$, and hence  in the $14$-th time
		slot, the value of the state variable $\mathit{stress}$ is $5$ and  the system will enter in an unsafe state. 
		In the $15$-th time
		slot, the value of the state variable $\mathit{stress}$ is still equal to $5$ 
		and the system will still be in an unsafe state.
		However, 
		the value of the state variable $\mathit{temp}$ will be at most
		$12.6 - 5*(1-\delta)= 12.6- 5*0.6=9.6$
		which will be sensed by process $\mathit{IDS}$  as at most
		$9.7$ (sensor error $\epsilon=0.1$). As a consequence,  no alarm will be turned on and the   variable $\mathit{stress}$ will be reset.
		Moreover, the invariant will be obviously always preserved. 
		
		As in the current time slot the attack has already terminated, from this
		point in time on, the system will behave correctly with neither deadlocks
		or alarms.

		\item 
		Let $n \geq 10$. 
		In order to prove that  $\mathit{Sys} \parallel A_{n} \, \sqsubseteq_{[14,n{+}11]} \,
		\mathit{Sys}$,   it is enough to show that:
		\begin{itemize}[noitemsep]
			
			\item  the system $\mathit{Sys} \parallel A_n$ does not deadlock;
			
			\item the system $\mathit{Sys} \parallel A_{n}$ may not emit any output in the first $13$ time slots; 
			\item there is a trace in which the system $\mathit{Sys} \parallel A_{n}$ enters in an unsafe state in the $14$-th time slot;
			
			\item 
			there is a trace in which the system $\mathit{Sys} \parallel A_n$ is in an unsafe state in the $(n{+}11)$-th time slot;

			\item   the system $\mathit{Sys} \parallel A_n$ does not have any execution
			trace emitting an output along channel $\mathit{alarm}$
			or entering in an unsafe state after   the $n+11$-th time slot. 
			
		\end{itemize}
		
		As regards the first fact, since $\delta=0.4$, by Lemma~\ref{lem:sys2} the
		temperature of the system under attack will always remain in the real
		interval $[0, 15.5]$. Thus, the invariant is never violated and the trace
		of the system under attack cannot contain any $\dead$-action. Moreover,
		when the attack terminates, if the temperature is in $[0,9.9]$, the system
		will continue his behaviour correctly, as in isolation. Otherwise, since
		the temperature is at most $15.5$, after a possible sequence of cooling
		cycles, the temperature will reach a value in the interval $[0,9.9]$, and
		again the system will continue its behaviour correctly, as in isolation.
		
		Concerning the second and the third facts, the proof is analogous to that
		of case $ n=9$.
		
		
		Concerning the fourth fact, firstly we prove that for all time slots $m$,
		with $9 < m \leq n$, there is a trace of the system $\mathit{Sys}
		\parallel A_n$ in which the state variable $\mathit{temp}$ reaches the
		values $14$ in the time slot $m$. Since the attack is alive at that time,
		and $\epsilon=0.1$, when the variable $\mathit{temp}$ will be equal to
		$14$ the sensed temperature will lay in the real interval $[9.9,10.1]$.
		
		The fastest trace reaching the temperature of $14$ degrees requires
		$\lceil \frac{14}{1 + \delta }\rceil = \lceil \frac{14}{1.4 }\rceil =   \frac{14}{1.4 } =10$
		time units, whereas the slowest one $\lceil \frac{14}{1 - \delta }\rceil =
		\lceil \frac{14}{0.6 }\rceil =24$ time units. Thus, for any time slot $m$,
		with $9 < m \leq 24$, there is a trace of the system where the value of
		the state variable $\mathit{temp}$ is $14$. Now, for any of those time
		slots $m$ there is a trace in which the state variable $\mathit{temp}$ is
		equal to $14$ in all time slots $m+10i < n$, with $i\in \mathbb{N}$. As
		already said, when the variable $\mathit{temp}$ is equal to $14$ the
		sensed temperature lays in the real interval $[9.9, 10.1]$ and the cooling
		might be activated. Thus, there is a trace in which the cooling system is
		activated. We can always find a trace where during the cooling the
		temperature decreases of $1+\delta$ degrees per time unit, reaching at the
		end of the cooling cycle the value of $5$. Thus, the trace may continue
		with $5$ time slots in which the variable $\mathit{temp}$ is increased of
		$1+\delta$ degrees per time unit; reaching again the value $14$. Thus, for
		all time slots $m$, with $9 < m \leq n$, there is a trace of the system
		$\mathit{Sys} \parallel A_n$ in which the state variable $\mathit{temp}$
		has value $14$ in the time slot $m$.
		
		Therefore, we can suppose that in the $n$-th time slot the variable
		$\mathit{temp}$ is equal to $14$ and, since the maximum increment of
		temperature is $1.4$, the the variable $\mathit{stress}$ is at least equal
		to $1$. Since the attack is alive and $\epsilon=0.1$, in the $n$-th time
		slot the sensed temperature will lay in $[9.9,10.1]$. We consider the case
		in which the sensed temperature is less than $10$ and hence the cooling is
		not activated.
		Thus, in the $n{+}1$-th time slot the system may reach a temperature of
		$14 + 1 + \delta=15.4$ degrees and the process $\mathit{Ctrl}$ will sense
		a temperature above $10$, and it will activate the cooling system. In this
		case, the variable $\mathit{stress}$ will be increased. As a consequence,
		after further $5$ time units of cooling, i.e.\ in the $n{+}6$-th time
		slot, the value of the state variable $\mathit{temp}$ may reach $15.4 -
		5*(1-\delta)=12.4$ and the alarm will be fired and the variable
		$\mathit{stress}$ will be still equal to $5$. After $4$ time unit in the $n{+}10$-th
		time slot
		the state variable $\mathit{temp}$ may reach $12.4 -
		4*(1-\delta)=10$ and   the variable
		$\mathit{stress}$ will be still equal to $5$ and
		the system will be in an unsafe state.
		So  in the $n{+}11$-th
		time slot $\mathit{stress}$ will be still equal to $5$ and
		the system will be in an unsafe state.
		
		Concerning the fifth fact,
		%
		%
		by Lemma~\ref{lem:sys2}, in the $n{+}1$-th time slot the attack will be terminated and the system may reach a temperature that is, in the worst case, at most $15.5$. Thus, the cooling system may be activated and the variable $\mathit{stress}$ will be increased. As a consequence, in the $n{+}11$-th time slot, the value of the state variable $\mathit{temp}$ may be at most $15.5-10*(1-\delta)=15.5-10*0.6=9.5$ and the variable $\mathit{stress}$ will be reset to $0$. Thus, after the $n+11$-th time slot, the system will behave correctly, as in isolation.
	\end{itemize}
\end{proof}

In order to prove Theorem~\ref{thm:sound}, we introduce the following lemma. 
\begin{lemma}
	\label{lem:top}
	Let $M$ be an honest and sound  \CPS{}, $C$ an arbitrary class
	of attacks, and $A$ an attack of a class $C' \preceq C$.
	Whenever $M\parallel A\trans t M'\parallel A'$, then 
	\begin{math}
	M\parallel \mathit{Top}(C)  \Trans{\hat{t}} M'  \parallel  \prod_{\iota \in \I} \mathit{Att}(\iota, \#\tick(t){+}1, C(\iota)) .
	\end{math}
\end{lemma} 
\begin{proof}
	Let us define $\mathit{Top}^h(C)$ as the attack process
	$\prod_{\iota \in \I } \mathit{Att}( \iota , h , C(\iota))$. 
	Then, $\mathit{Top}^1 (C)=\mathit{Top} (C)$. 
	The proof is by mathematical induction on the length $k$ of the trace $t$.

	\noindent 
	\emph{Base case.} $k=1$.
	This means $t=\alpha $, for some action $\alpha$.
	We proceed by case analysis on $\alpha$.
	\begin{itemize}
		\item $\alpha = \out {c} v$. As the attacker $A$ does not use 
		communication channels, from $M\parallel A\trans {\out {c} v} M'\parallel A'$ we can derive that $A=A'$ and $M \trans {\out {c} v} M'$. Thus, by applications of rules \rulename{Par} and \rulename{Out} we derive $M \parallel \mathit{Top}(C) \trans {\out {c} v} M' \parallel \mathit{Top}^1(C)=M' \parallel \mathit{Top}(C)$.
		
		\item $\alpha = \inp {c} v$. This case is similar to the previous one.
		
		\item $\alpha =\tau$. There are five sub-cases. 
		\begin{itemize}

			\item Let $M\parallel A\trans {\tau}  M'\parallel A'$ be derived by 
			an application of rule \rulename{SensRead}. 
			Since the attacker $A$ performs only malicious actions, from $M\parallel A\trans {\tau} M'\parallel A'$ we can derive that $A=A'$ and
			$P \trans {\rcva { s} v} P'$ for some process $P$ and $P$' such that
			$M=\confCPS {E; S} P$ and $M'=\confCPS {E; S} {P'}$.
			By considering $\mathit{rnd}(\{\true,\false\})=\false$ for any process $\mathit{Att}(\iota, 1, C(\iota))$, we have that $\mathit{Top}(C)$ can only perform a $\tick $ action, and $\mathit{Top}(C) \ntrans { \snda {\mbox{\Lightning}s} v}$. Hence, by an application of rules \rulename{Par} and \rulename{SensRead} we derive $M \parallel \mathit{Top}(C) \trans  {\tau} M' \parallel \mathit{Top}^1(C)=M' \parallel \mathit{Top}(C)$.
			
			\item Let $M\parallel A\trans {\tau} M'\parallel A'$ be derived by an 
			application of rule \rulename{ActWrite}. 
			This case is similar to the previous one.
			
			\item Let $M\parallel A\trans {\tau} M'\parallel A'$ be derived by 
			an application of rule \rulename{$\mbox{\Lightning}$SensSniff$\mbox{\,\Lightning}$}. 
			Since $M$ is sound it follows that $M=M'$ and $A\trans {\rcva {\mbox{\Lightning}s} v }A'$. This entails $1 \in C'({\mbox{\Lightning}s?}) \subseteq C({\mbox{\Lightning}s?}) $. 
			By assuming $\mathit{rnd}(\{\true,\false\})=\true$ for the process $\mathit{Att}(  {\mbox{\Lightning}s?}, 1, C( {\mbox{\Lightning}s?}))$, 
			it follows that 
			\begin{math} 
			\mathit{Top}(C) \trans { \rcva {\mbox{\Lightning}s} v} \mathit{Top}^1(C) =\mathit{Top}(C) . 
			\end{math}
			Hence, by applying the rules \rulename{Par} and \rulename{$\mbox{\Lightning}$SensRead$\mbox{\,\Lightning}$} we derive $M \parallel \mathit{Top}(C) \trans {\tau} M' \parallel \mathit{Top}^1(C)=M' \parallel \mathit{Top}(C)$.
			
			\item Let $M\parallel A\trans {\tau} M'\parallel A'$ be derived by 
			an application of rule \rulename{$\mbox{\Lightning}$ActIntgr$\mbox{\,\Lightning}$}. Since $M$ is sound it follows that $M=M'$ and $A\trans {\snda {\mbox{\Lightning}a} v }A'$. 
			As a consequence, $1 \in C'({\mbox{\Lightning}a}!) \subseteq C({\mbox{\Lightning}a}!)$. 
			By assuming 
			$\mathit{rnd}(\{\true,\false\}){=}\true$ and $\mathit{rnd}(\mathbb{R})=v$ for the process $\mathit{Att}( {\mbox{\Lightning}a}!, 1, C({\mbox{\Lightning}a}!))$, it follows that 
			\begin{math} 
			\mathit{Top}(C) \trans { \snda {\mbox{\Lightning}a} v} \mathit{Top}^1(C) =\mathit{Top}(C)  . 
			\end{math} 
			Thus, by applying the rules \rulename{Par} and \rulename{$\mbox{\Lightning}$ActIntegr$\mbox{\,\Lightning}$} we derive $M \parallel \mathit{Top}(C) \trans  {\tau} M' \parallel \mathit{Top}^1(C)=M' \parallel \mathit{Top}(C)$.
			
			\item Let $M\parallel A\trans {\tau}  M'\parallel A'$ be derived by 
			an application of rule \rulename{Tau}. 
			Let $M=\confCPS {E; S} P$ and $M'=\confCPS {E'; S} {P'}$. 
			First, we consider the case when $P \parallel A \trans \tau P' \parallel A'$ is derived by an application of either rule \rulename{$\mbox{\Lightning}$SensIntegr$\mbox{\,\Lightning}$} or rule \rulename{$\mbox{\Lightning}$ActDrop$\mbox{\,\Lightning}$}.
			Since $M$ is sound and $A$ can perform only malicious actions, we have that:
			\begin{inparaenum}[(i)]
				\item either $P\trans {\rcva { s} v }P'$ and $A\trans  {\snda {\mbox{\Lightning}s} v }A'$ 
				\item or $P\trans  {\snda { a} v }P'$ and $A\trans  {\rcva {\mbox{\Lightning}a} v }A'$.
			\end{inparaenum}
			We focus on the first case as the second one is similar. 
			
			Since $A\trans {\snda {\mbox{\Lightning}s} v }A'$, we derive $1 \in C'(\mbox{\Lightning}s!) \subseteq C(\mbox{\Lightning}s!)$, and $\mathit{Top}(C) \trans { \snda {\mbox{\Lightning}s} v} \mathit{Top}^1(C) =\mathit{Top}(C)$, by assuming $\mathit{rnd}(\{\true,\false\})= \true$ and $\mathit{rnd}(\mathbb{R})= v$ for the process $\mathit{Att}( {\mbox{\Lightning}s}!, 1, C({\mbox{\Lightning}s}!))$. 
			Thus, by applying the rules \rulename{$\mbox{\Lightning}$SensIntegr$\mbox{\,\Lightning}$} and \rulename{Tau} we derive $M \parallel \mathit{Top}(C) \trans {\tau} M' \parallel \mathit{Top}^1(C)=M' \parallel \mathit{Top}(C)$, as required. 
			
			To conclude the proof we observe that if $P \parallel A \trans \tau P' \parallel A'$ is derived by an application of a rule different from \rulename{$\mbox{\Lightning}$SensIntegr$\mbox{\,\Lightning}$} 
			and \rulename{$\mbox{\Lightning}$ActDrop$\mbox{\,\Lightning}$}, then 
			by inspection of Table~\ref{tab:lts_processes} and by definition of attacker, it follows that $A$ can't perform a $\tau$-action since $A$ does not use channel communication and performs only malicious actions. Thus, the only possibility is that the $\tau$-action is performed by $P$ in isolation. As a consequence, by applying the rules \rulename{Par} and \rulename{Tau}, we derive $M \parallel \mathit{Top}(C) \trans  {\tau} M' \parallel \mathit{Top}^1(C)=M' \parallel \mathit{Top}(C)$.
		\end{itemize} 
		\item $\alpha = \tick$. In this case the transition $M\parallel A\trans \tick  M'\parallel A'$ is derived by an application of rule \rulename{Time} because $M\trans \tick M'$ and $A \trans \tick A'$.
		Hence, it suffices to prove that $\mathit{Top}(C) \trans \tick \mathit{Top}^2(C)$.
		We consider two cases: $1 \in C(\iota)$ and $1 \not\in C(\iota)$.
		If $1 \in C(\iota)$, then the transition $\mathit{Att}( \iota, 1 , C(\iota)) \trans \tick \mathit{Att}( \iota, 2, C(\iota))$ can be derived by assuming $\mathit{rnd}(\{\true,\false\})=\false$. 
		Moreover, since $\mathit{rnd}(\{\true,\false\})=\false$ the process $ \mathit{Att}( \iota, 1, C(\iota)) $ can only perform a $\tick$ action.
		If $1 \not \in C(\iota)$, then the process $ \mathit{Att}( \iota, 1 , C(\iota)) $ can only perform a $\tick $ action. As a consequence, $  \mathit{Att}( \iota, 1 , C(\iota))  \trans \tick \mathit{Att}( \iota, 2 , C(\iota)) $ and $\mathit{Top}(C) \trans \tick  \mathit{Top}^2(C) $. 
		By an application of rule \rulename{Time}, we derive $M \parallel \mathit{Top}(C) \trans  {\tick} M' \parallel \mathit{Top}^2(C)$.

		\item $\alpha = \dead$. This case is not admissible because $M\parallel A\trans {\dead} M'\parallel A'$ would entail $M \trans {\dead} M'$. However, $M$ is sound and it can't deadlock. 
		
		\item $\alpha = \unsafe$. Again, this case is not admissible because 
		$M$ is sound. 
	\end{itemize}
	\emph{Inductive case} ($k>1$). 
	We have to prove that $M\parallel A\trans t M'\parallel A'$ implies 
	$M\parallel \mathit{Top}(C) \Trans{\hat{t}} M'\parallel \mathit{Top}^{ \#\tick(t)+1  }(C)$. 
	Since the length of $t$ is greater than $1$, it follows that $t=t' \alpha$, for some trace $t'$ and  some action $\alpha$.
	Thus, there exist $M''$ and $A''$ such that 
	\begin{math} 
	M\parallel A\trans {t'}   M''\parallel A'' \trans \alpha  M'\parallel A'
	.
	\end{math} 
	By inductive hypothesis, it follows that $M\parallel \mathit{Top}(C) \Trans{\hat{t'}} M''\parallel\mathit{Top}^{\#\tick(t')+1} (C)$. To conclude the proof, it is enough to show that $M''\parallel A'' \trans \alpha M'\parallel A'$ implies $M'' \parallel  \mathit{Top}^{\#\tick(t' )+1 } (C) \Trans{\hat{\alpha}} M'\parallel \mathit{Top}^{\#\tick(t  )+1 }(C)$.
	The reasoning is similar to that followed in the base case, except 
	for actions $\alpha=\dead$ and $\alpha=\unsafe$ that need to be treated separately. We prove the case $\alpha=\dead$ as the case $\alpha=\unsafe$ is similar.

	Let $M=\confCPS {E; S} P$. The transition $M''\parallel A\trans {\dead}  M'\parallel A'$ must be derived by an application of rule \rulename{Deadlock}. This implies that $M''=M'$, $A''=A'$ and the state function of $M$ is not in the invariant set $ \invariantfun{}$. Thus, by an application of rule \rulename{Deadlock} we derive  
	\begin{displaymath} 
	M'' \parallel  \mathit{Top}^{\#\tick(t' )+1 } (C)
	\trans {\dead}  M'\parallel \mathit{Top}^{\#\tick(t' )+1 }(C) .
	\end{displaymath} 
	Since $\#\tick(t) + 1 = \#\tick(t') + \#\tick(\dead) +1=\#\tick(t' )+1$, it follows, as required, that 
	\begin{displaymath} 
	M'' \parallel \mathit{Top}^{\#\tick(t' )+1 } (C)
	\trans {\dead} M'\parallel \mathit{Top}^{\#\tick(t  )+1 }(C) \, . 
	\end{displaymath} 
\end{proof}

Everything is finally in place to prove Theorem~\ref{thm:sound}.
\begin{proof}[Proof of Theorem~\ref{thm:sound}] 
	We have to prove that either $M \parallel A \sqsubseteq M$ or $M \parallel A \sqsubseteq_{{m_2}..{n_2}} M$, for some $m_2$ and $n_2$ such that $m_2..n_2 \subseteq m_1..n_1$ ($m_2=1$ and $n_2=\infty$ if the two systems are completely unrelated). The proof proceeds by contradiction. 
	Suppose that $M \parallel A \not\sqsubseteq M$ and $M \parallel A \sqsubseteq_{{m_2}..{n_2}} M$, with $m_2..n_2 \not \subseteq m_1 .. n_1$. We distinguish two cases: either $n_1=\infty$ or $n_1 \in \mathbb{N}^+$. 
	
	If $n_1=\infty$, then it must be $m_2 < m_1$. Since $M \parallel A
	\sqsubseteq_{{m_2}..{n_2}} M$, by Definition~\ref{Time-bounded-trace-equivalence}
	there is a trace $t$, with $\#\tick(t)=m_2{-} 1$, such that $M \parallel A
	\trans t$ and $M \not\!\Trans{\hat{t}}$. By Lemma~\ref{lem:top}, this
	entails $ M \parallel \mathit{Top}(C)\Trans{\hat{t}}$. Since $M
	\not\!\Trans{\hat{t}}$ and $\#\tick(t)=m_2{-}1 < m_2 < m_1$, this contradicts
	$M \parallel \mathit{Top}(C) \sqsubseteq_{{m_1}..{n_1}} M$.
	
	If $n_1 \in \mathbb{N}^+$, then $m_2 < m_1$ and/or $n_1 < n_2$, and we
	reason as in the previous case.
\end{proof}

\subsection{Proofs of Section~\ref{sec:impact}}

In order to prove Proposition~\ref{prop:toll}, we need a couple of lemmas. 

Lemma~\ref{lem:sys3} is a variant of Lemma~\ref{lem:sys}. Here the behaviour
of $\mathit{Sys} $ is parametric on the uncertainty.

\begin{lemma} 
	\label{lem:sys3}
	Let $\mathit{Sys}$ be the system defined in Section~\ref{sec:running_example}, and 
	$0.4 < \gamma \leq \frac{9}{20}$. 
	Let
	\begin{small}
		\begin{displaymath}
		\replaceENV {Sys} \delta  \gamma   =\mathit{Sys_1} \trans{t_1}\trans\tick 
		\mathit{Sys_2}  \dots 
		\trans{t_{n-1}}\trans\tick  \mathit{Sys_n}
		\end{displaymath}
	\end{small}%
	such that the traces $t_j$ contain no $\tick$-actions, for any $j \in  1 .. n{-}1 $, and for any  $i \in  1 .. n $ $\mathit{Sys_i}= \confCPS {S_i}{P_i} $ with 
	$S_i = \langle \statefun^i{} ,  \sensorfun^i{} , \actuatorfun^i{} \rangle$.
	Then, for any $i \in 1 .. n{-}1 $ we have the following:
	\begin{itemize}[noitemsep]
		\item  if   $ \actuatorfun^i{}(\mathit{cool})= \off $ then
		$\statefun^i{}(\mathit{temp})\in [0, 11.1+ \gamma ]$
		and $\statefun^i{}(\mathit{stress})=0$ if $ \statefun^i{}(\mathit{temp})  \in [0, 10.9+\gamma ] $ and, otherwise,  $\statefun^i{}(\mathit{stress})=1 $;

		\item   if   $\actuatorfun^i{}(\mathit{cool})= \off $ and 
		$\statefun^i{}(\mathit{temp})\in (10.1, 11.1+ \gamma ]$ then   $ \actuatorfun^{i+1}{}(\mathit{cool}) =\on$
		and $\statefun^{i{+}1}{}(\mathit{stress})  \in 1..2$;

		\item  if  $ \actuatorfun^i{}(\mathit{cool})=\on$ then   $\statefun^i{}(\mathit{temp}) \in 
		( 9.9-k *(1+\gamma), 11.1+\gamma -k*(1- \gamma)] $, 
		for some  $k  \in 1..5$,   such that $\actuatorfun^{i-k}{}(\mathit{cool}) =\off $ and $ \actuatorfun^{i-j}{}(\mathit{cool}) =\on $, for $j \in 0..k{-}1$;
		moreover,  if  $k\in 1.. 3$ then     $\statefun^i{}(\mathit{stress})  \in  1..k{+}1  $, otherwise,  
		$\statefun^i{}(\mathit{stress}) =0$.

	\end{itemize}
\end{lemma}
\begin{proof}
	Similar to the proof of    Lemma~\ref{lem:sys}.
	The crucial difference w.r.t.\ the proof of   Lemma~\ref{lem:sys} 
	is limited to 
	the second part of the third item. In particular the part saying that
	$\statefun^i{}(\mathit{stress}) =0$, when $k \in 4..5$.
	Now, after $3$ time units of cooling, the state variable 
	$ \mathit{stress} $ lays  in the integer interval $  1..k{+}1=1..4$.
	Thus, in order to  have  $\statefun^i{}(\mathit{stress}) =0$, when $k \in 4..5$, the   temperature in the
	third time slot of  the cooling must be less than or equal to $9.9$.
	However, from the first statement of the third item 
	we deduce that,   in the
	third time slot of cooling,  the  state variable $\mathit{temp}$
	reaches at most 
	$ 11.1+\gamma -3*(1- \gamma)  =  8.1+4\gamma $. Thus, 
	Hence we have that $ 8.1+4\gamma \leq 9.9$ for $\gamma \leq \frac{9}{20}$.
\end{proof}

The following lemma is a variant of Proposition~\ref{prop:sys}. 

\begin{lemma} 
	\label{lem:sys:damage}
	Let $\mathit{Sys}$ be the system defined in Section~\ref{sec:running_example} and
	$\gamma$ such that 
	$0.4 < \gamma \leq \frac{9}{20}$. 
	If $\replaceENV {\mathit{Sys}}   \delta \gamma    \trans{t} Sys'$, for some $t=\alpha_1 \ldots \alpha_n$,  then
	$\alpha_i \in \{ \tau , \tick \}$, for any $i \in 1 .. n$. 
\end{lemma}
\begin{proof}
	By Lemma~\ref{lem:sys3}, the temperature will always lay in the real
	interval $ [0, 11.1+ \gamma ]$. As a consequence, since $ \gamma \leq
	\frac{9}{20}$, the system will never deadlock.
	
	Moreover, after $5$ $\tick$ action of coolant the state variable
	$\mathit{temp}$ is in $( 9.9-5 *(1+\gamma), 11.1+\gamma -5*(1- \gamma)]
	=(4.9 -5\gamma \, , \, 6.1+6\gamma]$. Since $\epsilon = 0.1$, the value
	detected from the sensor will be in the real interval $(4.8 -5\gamma \, ,
	\, 6.2+6\gamma]$. Thus, the temperature sensed by $\mathit{IDS}$ will be
	at most $6.2 + 6\gamma \leq 6.2+6*\frac{9}{20}\leq 10 $, and no alarm will
	be fired.

	Finally, the maximum value that can be reached  by the state variable 
	$ \mathit{stress} $ is $   k{+}1$m for $k=3$. As a consequence, 
	the system will not reach an unsafe state. 
\end{proof}

The following Lemma is a variant of Proposition~\ref{prop:X}.   Here the behaviour of $\mathit{Sys} $ is  parametric on the 
uncertainty. 

\begin{lemma}
	\label{lem:X3}  
	Let $\mathit{Sys}$ be the system defined in Section~\ref{sec:running_example} and
	$\gamma$ such that  $0.4 < \gamma \leq \frac{9}{20}$. 
	Then, for 
	any execution trace of $\replaceENV  {Sys}  \delta \gamma$ we have the following:
	\begin{itemize}[noitemsep]
		\item if either process $\mathit{Ctrl}$ or process $\mathit{IDS}$ senses a temperature above $10$ then the value of
		the state variable $\mathit{temp}$ ranges over $(9.9, 11.1+\gamma]$;
		\item 
		when the process  $\mathit{IDS}$  tests the temperature the value of
		the state variable $\mathit{temp}$ 
		ranges over $( 9.9-5 *(1+\gamma), 11.1+\gamma -5*(1- \gamma)] $. 
	\end{itemize}
\end{lemma}
\begin{proof}
	As to the first statement,  since $\epsilon=0.1$, if either process $\mathit{Ctrl}$ or process  $\mathit{IDS}$ senses  a temperature above $10$ then the value of
	the state variable $\mathit{temp}$  is above $9.9$. 
	By Lemma~\ref{lem:sys3}, the state variable $\mathit{temp}$ is less than or equal to $11.1+\gamma$.
	Therefore, \emph{if either process $\mathit{Ctrl}$ or process  $\mathit{IDS}$ sense }  a temperature above $10$ then the value of
	the state variable $\mathit{temp}$ is in $(9.9,11.1+\gamma]$.
	
	Let us prove now the second statement. When the process $\mathit{IDS}$
	tests the temperature then the coolant has been active for $5$ $\tick$
	actions. By Lemma~\ref{lem:sys3}, the state variable $\mathit{temp}$ ranges
	over $( 9.9-5 *(1+\gamma), 11.1+\gamma -5*(1- \gamma)] $. 
\end{proof}

Everything is finally in place to prove Proposition~\ref{prop:toll}.

\begin{proof}[Proof of Proposition~\ref{prop:toll}] 
	For (1) we have to show that  $ \replaceENV  {\mathit{Sys}} \delta \gamma \, \sqsubseteq \,  \mathit{Sys}$, for $\gamma \in (\frac{8}{20} ,\frac{9}{20})$.  
	But this obviously holds by Lemma~\ref{lem:sys:damage}.  

	As regards item (2), we have to prove that $ \replaceENV  {\mathit{Sys}} \delta \gamma \, \not \sqsubseteq \,  \mathit{Sys}$, for $\gamma
	> \frac{ 9}{20} $. By Proposition~\ref{prop:sys} it is enough to show that
	the system $\replaceENV {\mathit{Sys}}   \delta \gamma$ has a trace which either
	(i) sends an alarm, or (ii) deadlocks, or (iii) enters in an unsafe state. We can easily build up a trace for
	$\replaceENV {\mathit{Sys}}  {\delta} \gamma$ in which, after $10$
	$\tick$-actions, in the $11$-th time slot, the value of the state
	variable $\mathit{temp}$ is $10.1$. In fact, it is enough to increase the
	temperature of $ 1.01$ degrees for the first $10$ rounds. Notice that this
	is an admissible value since, $ 1.01 \in [ 1-\gamma,1+\gamma ]$, for any $
	\gamma > \frac{ 9}{20}$. Being $10.1$ the value of the state variable
	$\mathit{temp}$, there is an execution trace in which the sensed
	temperature is $10$ (recall that $\epsilon=0.1$) and hence the cooling
	system is not activated but the state variable $\mathit{stress}$ will be increased. 
	In the following time slot, i.e.,
	the $12$-th time slot, the temperature may reach at most the value
	$10.1 + 1+\gamma$ and the state variable  $\mathit{stress}$ is $1$.  Now, if $10.1 + 1+\gamma>50$ then the system deadlocks.
	Otherwise, the controller will activate the cooling system, and after $3$ time
	units of cooling, in the $15$-th time slot, the state variable  $\mathit{stress}$ will be $4$ and the variable 
	$\mathit{temp}$ will be at most $11.1+\gamma -3(1-\gamma)=8.1+4\gamma$.  
	Thus, there is an execution trace in which the
	temperature is $ 8.1+4\gamma$, which will be greater than $9.9$ being
	$\gamma> \frac{ 9}{20}$. As a consequence, in the 
	next time slot, the state variable  $\mathit{stress}$ will be $5$
	and the system will enter in an unsafe state.
	
	This is enough to derive that $ \replaceENV  {\mathit{Sys}} \delta \gamma \, \not \sqsubseteq \,  \mathit{Sys}$, for $\gamma
	> \frac{ 9}{20} $. 
\end{proof}

\begin{proof}[Proof of Theorem~\ref{thm:sound2}]
	We consider the two parts of the statement separately. 
	
	\emph{Definitive impact.}
	By an application of Lemma~\ref{lem:top} we have that $M \parallel A\trans t$ entails $ M \parallel \mathit{Top}(C)\Trans{\hat{t}}$. This implies $M \parallel A \sqsubseteq M \parallel \mathit{Top}(C)$. Thus, if $ M \parallel \mathit{Top}(C) \sqsubseteq {\replaceENV M {\uncertaintyfun{}} {{\uncertaintyfun{}}{+}{\xi}}}$, for $\xi \in \mathbb{R}^{\hat{\cal X}}$, $\xi >0$, then, by transitivity of $\sqsubseteq$, it follows that \( M \parallel A \sqsubseteq   {\replaceENV M {\uncertaintyfun{}} {{\uncertaintyfun{}}{+}{\xi}}}\).
	
	\smallskip
	\emph{Pointwise impact}. 
	The proof proceeds by contradiction.
	Suppose $\xi' > \xi $. Since $ \mathit{Top}(C) $ has a pointwise impact $\xi$ at time $m$, it follows that $\xi$ is given by:
	\begin{displaymath}
	\inf \big\{ \xi''  :  \xi'' {\in} \mathbb{R}^{\hat{\cal X}} 
	\: \wedge \: M \parallel   \mathit{Top}(C)    \sqsubseteq_{m ..n} 
	\replaceENV M  {\uncertaintyfun{}}  {{\uncertaintyfun{}} {+} {\xi''}},  n \in \mathbb{N}^+ \cup \infty   \big\}.
	\end{displaymath}
	Similarly, since $A$ has a pointwise impact $\xi'$ at time $m'$, it
	follows that $\xi'$ is given by
	\begin{displaymath}
	\inf \big\{ \xi''  :  \xi'' {\in} \mathbb{R}^{\hat{\cal X}} 
	\,  \wedge \, M \parallel  A   \sqsubseteq_{m'..n} 
	\replaceENV M  {\uncertaintyfun{}}  {{\uncertaintyfun{}} {+} {\xi''}},  n \in \mathbb{N}^+ \cup \infty   \big\}.
	\end{displaymath}
	
	Now, if $m=m'$, then $\xi \geq \xi'$ because $M \parallel A \trans{t}$ entails $M \parallel \mathit{Top}(C)\Trans{\hat{t}}$ due to an application of Lemma~\ref{lem:top}. 
	This is contradiction with the fact that $\xi < \xi'$. 
	Thus, it must be $m' < m$. Now, since both $\xi $ and  $\xi'$ are the infimum functions and since $\xi' > \xi $, there are $\overline{\xi}$ 
	and $\overline{ \xi'}$, with $\xi {\leq} \overline{\xi}{\leq} \xi' {\leq} \overline{ \xi'}$ such that: 
	\begin{inparaenum}[(i)]
		\item $M \parallel \mathit{Top}(C) \sqsubseteq_{m..n} \replaceENV M {\uncertaintyfun{}} {{\uncertaintyfun{}} {+} {\overline{\xi}}}$, for some $n$; 
		\item  $M \parallel A \sqsubseteq_{m'..n'} \replaceENV M {\uncertaintyfun{}} {{\uncertaintyfun{}} {+} {\overline{\xi'}}}$, for some $n'$.
	\end{inparaenum}

	From $M \parallel A \sqsubseteq_{m'..n'} \replaceENV M {\uncertaintyfun{}} {{\uncertaintyfun{}} {+} {\overline{\xi'}}}$ it follows that there exists a trace $t$ with $\#\tick(t)=m'-1$ such that $M \parallel A  \trans{t}$ and $\replaceENV M  {\uncertaintyfun{}}  {{\uncertaintyfun{}} {+} {\overline{\xi'}}}  \not\!\Trans{\hat{t}}$. 
	Since $\overline{\xi} \leq \overline{\xi'} $, by monotonicity (Proposition~\ref{prop:monotonicity}), we deduce that $\replaceENV M  {\uncertaintyfun{}}  {{\uncertaintyfun{}} {+} {\overline \xi}}  \not\!\Trans{\hat{t}}$. 
	Moreover,  by Lemma~\ref{lem:top} $ M \parallel A \trans{t}$ entails 
	$ M \parallel \mathit{Top}(C)\Trans{\hat{t}}$. 
	
	Summarising, there exists a trace $t'$ with $\#\tick(t')=m'-1$ such that $M \parallel \mathit{Top}(C) \trans{t'}$ and $\replaceENV M
	{\uncertaintyfun{}} {{\uncertaintyfun{}} {+} {\overline \xi }}
	\not\!\Trans{\hat{t'}}$. However, this, together with $m' < m$, is in
	contradiction with the fact (i) above saying that $M \parallel \mathit{Top}(C) \sqsubseteq_{m ..n} \replaceENV M {\uncertaintyfun{}} {{\uncertaintyfun{}} {+} {\overline \xi }}$, for some $n$. As a consequence it must be $\xi' \leq \xi$ and $m' \leq m$. 
	This concludes the proof. 
\end{proof}

\begin{proof}[Proof of Proposition~\ref{prop:effect1}]
	Let us prove the first sub-result. 
	From Proposition~\ref{prop:att:integrity}  we know that 
	$\mathit{Sys}  \parallel  A_{10}  \, \sqsubseteq _{14..21}   \, Sys$. In particular we showed that 
	the system $\mathit{Sys}  \parallel  A_{10} $ has an execution trace which  
	is in an unsafe state from the $14$-th to the $21$-th time interval and fires only one alarm in the $16$-th time slot, and which cannot be matched by $\mathit{Sys}$.

	Hence, if in the $14$-th time slot the the system is in unsafe state, then the temperature in the
	$9$-th time slot  must be greater than $9.9$.
	Moreover, to fire an allarm in the  in the $16$-th time slot,   the cooling must be activated in the $11$-th time slot and hence  the temperature in the
	$10$-th time slot  must be less or equal than $10.1$ (recall that $\epsilon=0.1$).    
	But this is impossible since in the
	$9$-th time slot  $\mathit{temp}$ is greater than $9.9$, and, the minimum increasing of the temperature is $1-\gamma=0.2$.

	As a consequence, 
	for  $\gamma \leq 0.8$  we have 
	\[
	Sys  \parallel  A_{10}   \not \sqsubseteq \,   Sys   {\subst { {\gamma}}{\delta}}
	\enspace .   
	\]

	Let us prove the second sub-result. 
	That is,  \[ Sys  \parallel  A_{10}    \sqsubseteq \,   Sys   {\subst { {\gamma}}{\delta}} \]
	for $\gamma > 0.8$. Formally, we have to demonstrate that 
	whenever $\mathit{Sys} \parallel  A_{10} \trans{t}$, for some trace $t$, then 
	$\mathit{Sys}   {\subst { {\gamma}}{\delta}} \Trans{\hat t}$ as well. 
	Let us do a case analysis on the structure of the trace $t$. 
	We notice that,   since $A_{10}$ is a temporanely attack, the trace $t$ does not contain  $\dead$-action. 
	We distinguish four possible cases.

	\begin{itemize}[noitemsep]

		\item The trace $t$ contains contains only $\tau$-, $\tick$-, $\overline{alarm}$- and  $\unsafe$-actions.
		Firstly, 
		notice that  the system $\mathit{Sys}  \parallel  A_{10} $ may produce \emph{only one}  output on 
		channel $alarm$, in the $16$-th time slot. 
		After that, the trace will have
		only $\tau$- and $\tick$-actions. 
		In fact,  in Proposition \ref{prop:att:integrity} we provided a trace 
		in which, in the $16$-th time slot, the state variable $\mathit{temp}$
		reaches the value $10.5$ and an output on channel $alarm$ is emitted. 
		This is the maximum possible value for variable $\mathit{temp}$ in that point in time. 
		After the transmission of the $alarm$, the 
		system $\mathit{Sys}  \parallel  A_{10} $ activates the cooling for the following $5$
		time slots. Thus, in the $21$-th time slot, the temperature will
		be at most $10.5- 5*(1-\delta)=10.5-5*(0.6)=7.5$, and no alarm is fired. From that time on,  since the attack $A_{10} $ terminated is life in the $10$-th time slot,  no other alarms will be fired.
		Moreover, in Proposition \ref{prop:att:integrity}, we have showed that $\mathit{Sys}  \parallel  A_{10} $ is in an unsafe state from the $14$-th to the $21$-th time interval.
		
		Summarising $\mathit{Sys}  \parallel  A_{10} $ is in an unsafe state from the $14$-th to the $21$-th time interval and fires only one alarm in the $16$-th time slot.

		By monotonicity (Proposition~\ref{prop:monotonicity}), it is enough 
		to show that such a trace exists 
		for $\replaceENV {\mathit{Sys}  } \delta \gamma$, with 
		$0.8 < \gamma \leq 0.81$. In fact, if this trace exists for $0.8 < \gamma \leq 0.81$, 
		then it would also exist for
		$  \gamma > 0.81$. 
		In the following, we show how to build the trace of $\replaceENV {\mathit{Sys}  } \delta \gamma$ which simulates  the trace $t$ of $ \mathit{Sys}  \parallel A_{10}$. 
		
		We can easily build up a trace for $\mathit{Sys}   {\subst {  {\gamma} }{\delta}}$ 
		in which, after $8$ $\tick$-actions, in the $9$-th time slot, 
		the value of the state variable  $\mathit{temp}$  is in  $9.1+\gamma$. 
		In fact, it is enough to increase the temperature of  a value   $\frac{9.1+\gamma}{8}$ degrees
		for the first $8$ rounds. Notice that this is an admissible value since,    $\frac{9.1+\gamma}{8}$  is in $ [  1-\gamma,1+\gamma ]$, for any  $0.8 < \gamma \leq 0.81$. Moreover, since $\gamma >0.8$, in the $9$-th time slot  we have that $9.1+\gamma>9.9$.
		Now, in the in the $10$-th time slot,    $\mathit{temp}$  may reach  $9.1+\gamma+(1-\gamma)=10.1$.
		Being $10.1$ the value  of the state variable $\mathit{temp}$, there is an execution
		trace  
		in which the  sensed temperature is $10$ (recall that $\epsilon=0.1$) and hence 
		the cooling system is not activated. However, 
		in the following time slot, i.e.\ the $11$-th time slot,   
		the temperature may reach   the value
		$11.8$, imposing the activation of the cooling system (notice $1.7$ is an admissible increasing). 
		
		Summarizing, in the $9$-th time slot  the temperature is greater than $9.9$ and in the $11$-th time slot the cooling system is activated with a temperature equal to $11.8$.
		The thesis  follows from the following two facts:
		\begin{itemize}
			\item Since in the $9$-th time slot  the temperature is greater than $9.9$, then in the $ 14$-th time slot the system enters in an unsafe state. Since in the $11$-th time slot the cooling system is activated with a temperature equal to $11.8$, then, in the $20$-th time slot, after $9$ time units of cooling, the  temperature may reach the value  
			$11.8 -9(1-\gamma)=2.8+9 \gamma$ which will be greater than $9.9$ being 
			$\gamma> 0.8$. 
			Hence,  in the $21$-th time slot,    the system still be in an unsafe state.
			Finally, in the $21$-th time slot, after $10$ time units of cooling, the  temperature may reach the value  
			$11.8 -10(1-\gamma)=1.8+10\gamma$ which will be less  or equal to  $9.9$ being 
			$\gamma\leq  0.81$. 
			Hence,  in the $22$-nd time slot,  the  variable $\mathit{stress}$ is reset to $0$   and  the system enters in a safe state. 
			From that time on, 
			since   $\mathit{Sys}   {\subst {\gamma}{\delta}}$ can mimic all traces of  $\mathit{Sys}$,   we can always choose a trace which does not enter in an unsafe state any more.
			
			Summarising $\mathit{Sys}  \parallel  A_{10} $ is in an unsafe state from the $14$-th to the $21$-th time interval.
			
			\item Since in the $11$-th time slot the cooling system is activated with a temperature equal to $11.8$, 
			then, in the $16$-th time slot, the  temperature may reach the value  
			$11.8 -5(1-\gamma)=6.8+5\gamma$.
			Since $\epsilon = 0.1$, the sensed temperature would be  in 
			the real interval $[6.7+ 5\gamma , 6.9+5\gamma] $. Thus,   the sensed temperature is  greater than $10$ being 
			$\gamma> 0.8$. 
			Thus,  the alarm will be transmitted, in the $16$-th time slot,  
			as required. 
			After the transmission on channel $alarm$, the 
			system $\mathit{Sys}   {\subst {  {\gamma} }{\delta}}$ activates the cooling for the following $5$
			time slots.   As a consequence,  in $21$-th time slot, the temperature
			will be at most $6.8+ 5\gamma- 5*(1-\gamma)= 1.8+10\gamma$. Since we assumed $ 0.8< \gamma \leq 0.81$ the temperature will be well below $10$ and no alarm will be sent. From that time on, 
			since   $\mathit{Sys}   {\subst {\gamma}{\delta}}$ can mimic all traces of  $\mathit{Sys}$,   we can always choose a trace which does not fire the alarm any more.
			
			Summarising $\mathit{Sys}  \parallel  A_{10} $  fires only one alarm in the $16$-th time slot. 
		\end{itemize}

		\item The trace $t$ contains contains only $\tau$-, $\tick$- and
		$\unsafe$-actions. This case is similar to the previous one.

		\item The trace $t$ contains only $\tau$-, $\tick$- and
		$\overline{alarm}$-actions. This case cannot occur. In fact, an $\overline{alarm}$-action
		can not occur without un $\unsafe$-action.
		
		\item The trace $t$ contains only $\tau$- and $\tick$-actions.
		If the system  $\mathit{Sys} \parallel  A_{10}$ has a trace
		$t$ which contains only $\tau$- and $\tick$-actions then,
		by Proposition \ref{prop:sys}, the system $\mathit{Sys}$ in isolation  must have a similar trace with the same number of  $\tick$-actions. 
		By an application of Proposition \ref{prop:monotonicity}, as  $\delta<\gamma$, 
		any trace of $\mathit{Sys} $ can be simulated by  
		$\mathit{Sys}   {\subst { {\gamma}}{\delta}} $. As a consequence, $\mathit{Sys}   {\subst { {\gamma}}{\delta}} \Trans{\hat t}$.

	\end{itemize}
	This is enough to derive that:
	\[
	Sys  \parallel  A_{10}    \sqsubseteq \,   Sys   {\subst { {\gamma}}{\delta}}
	\enspace .   
	\]
\end{proof}

\begin{proof}[Proof of Proposition~\ref{prop:effect2}]
	Let us prove the first sub-result. 
	As demonstrated  in Example~\ref{exa:att:DoS2},  we know that 
	$\mathit{Sys} \parallel A  \sqsubseteq_{14..\infty}  \mathit{Sys}$
	because  in the $14$-th time slot 
	the compound system  will violate the safety conditions emitting an $\unsafe$-action  until the invariant will  be violated.
	No alarm will be emitted.
	
	Since the system keeps violating the safety condition
	the temperature must remain greater than $9.9$.
	As proved for Lemma~\ref{lem:sys3} 
	we can prove that we have that the temperature is less than or equal to $11.1+\gamma $.  
	Hence, in the time slot before getting in deadlock, the temperature 
	of the system is in the real interval $(9.9,11.1+\gamma]$.
	To deadlock  with one $\tick$ action and from  a temperature in the real interval $(9.9,11.1+\gamma]$, either the temperature  reaches  a value greater than $50$ (namely, $11.1+\gamma+1+\gamma > 50$) or 
	the temperature reaches a value less than $ 0$ (namely, $9.9-1-\gamma <  0$ ). 
	Since $\gamma \leq 8.9$, both cases can not occur. Thus,  we have that 
	\[
	\mathit{Sys}  \parallel  A   \not \sqsubseteq \,  \replaceENV{\mathit{Sys}}  {\delta} {\gamma} 
	\enspace .   
	\]
	Let us prove the second sub-result. 
	That is,  \[ Sys  \parallel  A    \sqsubseteq \,   \replaceENV{\mathit{Sys}}  {\delta} {\gamma} \]
	for $\gamma >8.9$. 
	We demonstrate that 
	whenever $\mathit{Sys} \parallel  A \trans{t}$, for some trace $t$, then 
	$\replaceENV {\mathit{Sys}  } \delta \gamma\Trans{\hat t}$ as well. 
	We will proceed by case analysis on the kind of actions contained in $t$. 
	We distinguish four possible cases.

	\begin{itemize}[noitemsep]
		
		\item The trace $t$ contains contains only $\tau$-, $\tick$-, $\unsafe$-
		and $\dead$-actions. As discussed in Example~\ref{exa:att:DoS2},
		$\mathit{Sys} \parallel A \; \sqsubseteq_{14..\infty} \; \mathit{Sys}$
		because in the $14$-th time slot the system will violate the safety
		conditions emitting an $\unsafe$-action until the invariant will be
		broken. No alarm will be emitted. Note that, when $\mathit{Sys} \parallel
		A$ enters in an unsafe state then the temperature is at most
		$9.9+(1+\delta)+5(1+\delta)=9.9+6(1.4)=18.3$. Moreover, the fastest
		execution trace, reaching an unsafe state, deadlocks just after $\lceil
		\frac{ 50-18.3}{1 + \delta } \rceil = \lceil \frac{ 31,7}{1.4 } \rceil=23$
		$\tick$-actions. Hence, there are $m,n \in \mathbb{N}$, with $m\geq 14$ and
		$n\geq m+23$, such that the trace $t$ of $\mathit{Sys} \parallel A $
		satisfies the following conditions: (i) in the time interval $1..m-1$ the
		trace $t$ of is composed by $\tau$- and $\tick$-actions; (ii) in the time
		interval $m..(n-1)$, the trace $t$ is composed by $\tau$-, $\tick$- and
		$\unsafe$- actions; (iii) in the $n$-th time slot the trace $t$ deadlocks.
		
		By monotonicity (Proposition~\ref{prop:monotonicity}), it is enough 
		to show that such a trace exists 
		for $\replaceENV {\mathit{Sys}  } \delta \gamma$, with 
		$8.9 < \gamma < 9$. In fact, if this trace exists for  $ 8.9 < \gamma < 9$, then it would also exist for
		$  \gamma \geq 9$. 
		In the following, we show how to build the trace of $\replaceENV {\mathit{Sys}  } \delta \gamma$ which simulates  the trace $t$ of $ \mathit{Sys}  \parallel A$. 
		We build up the trace in three steps: (i)  
		the sub-trace from time slot $1$ to  time slot $m{-}6$;
		(ii) the sub-trace from the time slot $m{-}5$ to the time slot $n{-}1$;
		(iii) the final part of the trace reaching the deadlock.
		\begin{itemize}
			\item[(i)] 
			As $\gamma>8.9$ (and hence $1+\gamma>9.9$), the system may increment the
			temperature of $9.9$ degrees after a single $\tick$-action. Hence, we
			choose the trace in which the system $\replaceENV {\mathit{Sys} } \delta
			\gamma$, in the second time slot, reaches the temperature equal to $9.9$.
			Moreover, the system may maintain this temperature value until the
			$(m{-}6)$-th time slot (indeed $0$ is an admissible increasing since $0
			\in [1-\gamma,1+\gamma]\supseteq [-7.9,10.9]$) . Obviously, with a
			temperature equal to $9.9$, only $\tau$- and $\tick$-actions are possible.

			\item[(ii)]
			Let $k \in \mathbb{R}$ such that $0< k < \gamma-8.9  $ (such $k$ exists since $\gamma>8.9 $).  
			We may consider an increment of the temperature of $k$. 
			This implies that in the $(m{-}5)$-th time slot, the system 
			$\replaceENV {\mathit{Sys}  } \delta \gamma$ may reach the temperature $9.9+k$. 
			Note that $ k$ is an admissible increment since $0< k < \gamma-8.9  $ and $8.9 < \gamma < 9$ entails $k \in (0,0.1)$. 
			Moreover, the system may maintain this temperature value  until the $(n{-}1)$-th time slot
			(indeed, as said before, $0$ is an admissible increment).
			Summarising from the $(m{-}5)$-th time slot to the $(n{-}1)$-th time slot,
			the temperature may remain equal to $9.9+k \in (9.9,10)$. As a
			consequence, from the $m$-th time slot to the $(n{-}1)$-th time slot the
			system $\replaceENV {\mathit{Sys} } \delta \gamma$ may enter in an unsafe
			state. Thus, an $\unsafe$-action may be
			performed in the time interval $m..(n{-}1)$. Moreover, since
			$\epsilon=0.1$ and the temperature is e $9.9+k \in (9.9,10)$, we can
			always assume that the cooling is not activated until the $(n{-}1)$-th
			time slot. This implies that neither alarm nor deadlock occur.

			\item[(iii)]
			At this point, since in the $(n{-}1)$-th time slot the temperature is
			equal to $9.9 + k \in (9.9,10)$ (recall that $k \in (0,1)$), the cooling
			may be activated. We may consider a decrement of $1+\gamma$. In this
			manner, in the $n$-th time slot the system may reach a temperature of
			$9.9+k-(1+\gamma)< 9.9+0 -1 -8.9 =0$ degrees, and the system $\replaceENV
			{\mathit{Sys} } \delta \gamma$ will deadlock.
			
		\end{itemize}
		
		Summarising, for any $\gamma > 8.9 $ the system $\replaceENV {\mathit{Sys}
		} \delta \gamma$ can mimic any trace $t$ of $ \mathit{Sys} \parallel A$.
		
		\item The trace $t$ contains contains only $\tau$-, $\tick$- and
		$\unsafe$-actions. This case is similar to the previous one.

		\item The trace $t$ contains only $\tau$-, $\tick$- and
		$\overline{alarm}$-actions. This case cannot occur. In fact, as discussed
		in Example~\ref{exa:att:DoS2}, the process $\mathit{Ctrl}$ never activates the
		$\mathit{Cooling}$ component (and hence also the $\mathit{IDS}$ component,
		which is the only one that could send an alarm) since it will always
		detect a temperature below $10$.
		
		\item The trace $t$ contains only $\tau$- and $\tick$-actions. If the
		system $\mathit{Sys} \parallel A $ has a trace $t$ that contains only
		$\tau$- and $\tick$-actions, then, by Proposition~\ref{prop:sys}, the system
		$\mathit{Sys}$ in isolation must have a similar trace with the same number
		of $\tick$-actions. By an application of Proposition~\ref{prop:monotonicity}, as
		$\delta<\gamma$, any trace of $\mathit{Sys} $ can be simulated by
		$\replaceENV {\mathit{Sys} } \delta \gamma$. As a consequence,
		$\replaceENV {\mathit{Sys} } \delta \gamma\Trans{\hat t}$.
	\end{itemize}
	
	This is enough to obtain what required:
	\[
	\mathit{Sys}   \parallel  A    \sqsubseteq \,   \replaceENV {\mathit{Sys}  } \delta \gamma
	\, . 
	\]
\end{proof}

\end{document}